\documentclass[a4paper,UKenglish,cleveref, autoref, thm-restate]{lipics-v2021}
\pdfoutput=1
\usepackage[all]{xy}
\usepackage{bm}
\usepackage{xcolor}
\usepackage{xspace}
\usepackage{hyperref}
\usepackage{amsmath,amsthm,amssymb}
\usepackage{framed}

\newcommand{\indrulename}[1]{\texttt{\textup{#1}}}
\newcommand{\emptyPremise}{\vphantom{@}}
\newcommand{\indruleNPos}[4]{
\begin{array}[#1]{c@{}r}
\!\!
 #3
\vspace{-.1cm}
\\
& \,#2\!\hspace{-.2cm}\vspace{-.2cm} \\
\cline{1-1}\vspace{-.3cm} \\
 \!\!#4\,\, %\hspace{.5cm}\,
\end{array}
}
\newcommand{\indrule}[3]{{\small\indruleNPos{b}{#1}{#2}{#3}}}

\renewcommand{\theenumi}{\arabic{enumi}}
\renewcommand{\theenumii}{\arabic{enumii}}
\renewcommand{\theenumiii}{\arabic{enumiii}}
\renewcommand{\theenumiv}{\arabic{enumiv}}
\makeatletter
\renewcommand\p@enumii{\theenumi.}
\renewcommand\p@enumiii{\theenumi.\theenumii.}
\renewcommand\p@enumiv{\theenumi.\theenumii.\theenumiii.}
\makeatother

\newcommand{\llem}[1]{\label{lemma:#1}}
\newcommand{\rlem}[1]{Lem.~\ref{lemma:#1}}
\newcommand{\ldef}[1]{\label{def:#1}}
\newcommand{\rdef}[1]{Def.~\ref{def:#1}}
\newcommand{\lprop}[1]{\label{prop:#1}}
\newcommand{\rprop}[1]{Prop.~\ref{prop:#1}}
\newcommand{\lthm}[1]{\label{thm:#1}}
\newcommand{\rthm}[1]{Thm.~\ref{thm:#1}}
\newcommand{\lremark}[1]{\label{remark:#1}}
\newcommand{\rremark}[1]{Rem.~\ref{remark:#1}}
\newcommand{\lcoro}[1]{\label{coro:#1}}
\newcommand{\rcoro}[1]{Coro.~\ref{coro:#1}}
\newcommand{\lsec}[1]{\label{section:#1}}
\newcommand{\rsec}[1]{Section~\ref{section:#1}}

\newcommand{\lexample}[1]{\label{example:#1}}
\newcommand{\rexample}[1]{Ex.~\ref{example:#1}}

\newcommand{\NotNow}[1]{}

\newcommand{\eg}{{\em e.g.}\xspace}
\newcommand{\ie}{{\em i.e.}\xspace}
\newcommand{\etal}{{\em et al.}}
\newcommand{\ih}{IH\xspace}
\newcommand{\ST}{\ |\ }
\newcommand{\HS}{\hspace{.5cm}}

\newcommand{\emptyseq}{\epsilon}
\newcommand{\set}[1]{\{#1\}}
\newcommand{\Nat}{\mathbb{N}}
\newcommand{\Natz}{\mathbb{N}_0}
\newcommand{\eqdef}{\overset{\mathrm{def}}{=}}

\newcommand{\lam}[2]{\lambda#1.\,#2}
\newcommand{\ulam}[2]{\underline{\lambda}#1.\,#2}
\newcommand{\app}[2]{#1\,#2}
\newcommand{\sub}[2]{[#1:=#2]}
\newcommand{\fv}[1]{\mathsf{fv}(#1)}
\newcommand{\garb}[1]{\bm{\{}#1\bm{\}}}
\newcommand{\bin}[2]{#1\garb{#2}}

\newcommand{\symg}{\mathbf{m}}
\newcommand{\lambdaG}{\lambda^{\symg}}

\newcommand{\tobeta}{\to_\beta}
\newcommand{\tobetad}[1]{\xrightarrow{#1}_\beta}
\newcommand{\rtobetad}[1]{\xrightarrow{#1}^{\raisebox{-1em}{{\scriptsize$\ast$}}}_\beta}
\newcommand{\tog}{\to_\symg}
\newcommand{\ptog}{\Rightarrow_\symg}
\newcommand{\tod}[1]{\xrightarrow{#1}_\symg}
\newcommand{\todplus}[1]{\xrightarrow{#1}^{\raisebox{-1em}{{\scriptsize$+$}}}_\symg}
\newcommand{\rtod}[1]{\xrightarrow{#1}^{\raisebox{-1em}{{\scriptsize$\ast$}}}_\symg}
\newcommand{\ptod}[1]{\overset{\!#1\,}{\Longrightarrow}_\symg}
\newcommand{\ptodplus}[1]{\ptod{#1}^{\raisebox{-1em}{{\scriptsize$+$}}}}

\newcommand{\tctx}{\Gamma}

\newcommand{\sctx}{\mathtt{L}}

\newcommand{\ctxhole}{\Box}
\newcommand{\gctx}{\mathtt{C}}
\newcommand{\of}[2]{#1[#2]}

\newcommand{\var}{x}
\newcommand{\vartwo}{y}
\newcommand{\varthree}{z}
\newcommand{\varfour}{w}

\newcommand{\ltm}{M}
\newcommand{\ltmtwo}{N}
\newcommand{\ltmthree}{P}

\newcommand{\tm}{t}
\newcommand{\tmtwo}{s}
\newcommand{\tmthree}{u}
\newcommand{\tmfour}{r}

\newcommand{\bbtyp}{0}
\newcommand{\btyp}{\alpha}
\newcommand{\btyptwo}{\beta}

\newcommand{\typ}{A}
\newcommand{\typtwo}{B}

\newcommand{\judg}[3]{#1\vdash#2:#3}

\newcommand{\typeof}[1]{\mathsf{type}(#1)}
\newcommand{\height}[1]{\mathsf{h}(#1)}
\newcommand{\weight}[1]{\mathsf{w}(#1)}
\newcommand{\maxdeg}[1]{\mathsf{maxdeg}(#1)}

\newcommand{\simp}[2]{\mathtt{S}_{#1}(#2)}
\newcommand{\simpk}[1]{\simp{k}{#1}}
\newcommand{\simpd}[1]{\simp{d}{#1}}
\newcommand{\simpfull}[1]{\simp{*}{#1}}

\newcommand{\shone}{\mathrel{\rhd}}

\newcommand{\meassym}{\mathcal{W}}
\newcommand{\meas}[1]{\meassym(#1)}

\newcommand{\rulePVar}{\indrulename{p-var}}
\newcommand{\rulePAbs}{\indrulename{p-abs}}
\newcommand{\rulePAppOne}{\indrulename{p-app$_1$}}
\newcommand{\rulePAppTwo}{\indrulename{p-app$_2$}}
\newcommand{\rulePBin}{\indrulename{p-wrap}}
\newcommand{\rulePCtxHole}{\indrulename{p-ctx-hole}}
\newcommand{\rulePCtxBin}{\indrulename{p-ctx-wrap}}

\newcommand{\redex}{R}
\newcommand{\redextwo}{S}
\newcommand{\redexthree}{T}
\newcommand{\mstep}{\mathbf{R}}
\newcommand{\msteptwo}{\mathbf{S}}

\newcommand{\redseq}{\rho}
\newcommand{\redseqtwo}{\sigma}

\newcommand{\steptomstep}[1]{\mathtt{sim}(#1)}
\newcommand{\msteptoderiv}[1]{\mathtt{red}(#1)}

\newcommand{\protract}[2]{#1^\curvearrowright#2}
\newcommand{\retract}[2]{#1^\curvearrowleft#2}

\newcommand{\src}[1]{#1^{\mathsf{src}}}
\newcommand{\tgt}[1]{#1^{\mathsf{tgt}}}

\newcommand{\mset}{\mathfrak{m}}
\newcommand{\msettwo}{\mathfrak{n}}
\newcommand{\ms}[1]{[#1]}
\newcommand{\msempty}{\ms{\,}}
\newcommand{\msb}[2]{\ms{#1 \ ||\ #2}}
\newcommand{\Multi}[1]{\mathbb{M}(#1)}
\newcommand{\mgt}{\succ}
\newcommand{\mgeq}{\succeq}
\newcommand{\mgtmap}{\bm{:}\succ\bm{:}}
\newcommand{\mleq}{\preceq}
\newcommand{\mgtone}{\mathrel{\mgt^1}}
\newcommand{\mtimes}{\otimes}

\newcommand{\ameset}[1]{\mathbb{T}_{#1}}
\newcommand{\bmeset}[1]{\mathbb{R}_{#1}}

\newcommand{\amesym}{\mathcal{T}^{\symg}}
\newcommand{\bmesym}{\mathcal{R}^{\symg}}
\newcommand{\ame}[2]{\amesym_{\leq#1}(#2)}
\newcommand{\amefull}[1]{\amesym(#1)}
\newcommand{\bme}[2]{\bmesym_{#1}(#2)}
\newcommand{\eme}[3]{\amesym_{#1}(#2,#3)}
\newcommand{\turingmesym}{\mathcal{T}}
\newcommand{\turingme}[1]{\turingmesym(#1)}
\newcommand{\turingmep}[1]{\turingmesym'(#1)}
\newcommand{\turingmetwo}[2]{\turingmesym_{#1}(#2)}
\newcommand{\turingmethreea}[2]{\mathcal{T}^{\beta}_{\leq#1}(#2)}
\newcommand{\turingmethreeb}[2]{\mathcal{R}^{\beta}_{#1}(#2)}

\newcommand{\anon}{\mathtt{X}}

\newcommand{\semtyp}[1]{[\![#1]\!]}
\newcommand{\sem}[1]{[#1]}
\newcommand{\Item}[1]{\textcolor{lipicsGray}{\textsf{\textbf{\textup{#1}}}}}

\newcommand{\incf}[1]{|#1|}
\newcommand{\counter}{{\tau}}

\newcommand{\SeeAppendix}[1]{\textcolor{purple}{#1}}

\title{Two Decreasing Measures for Simply Typed $\lambda$-Terms (Extended Version)} %TODO Please add

\author{Pablo Barenbaum}{ICC, Universidad de Buenos Aires, Argentina \and Universidad Nacional de Quilmes (CONICET), Argentina}{pbarenbaum@dc.uba.ar}{}{}

\author{Cristian Sottile}{ICC, Universidad de Buenos Aires (CONICET), Argentina \and Universidad Nacional de Quilmes, Argentina}{csottile@icc.fcen.uba.ar}{}{}

\authorrunning{P. Barenbaum and C. Sottile} %TODO mandatory. First: Use abbreviated first/middle names. Second (only in severe cases): Use first author plus 'et al.'

\Copyright{Pablo Barenbaum and Cristian Sottile} %TODO mandatory, please use full first names. LIPIcs license is "CC-BY";  http://creativecommons.org/licenses/by/3.0/


 %TODO mandatory; please add comma-separated list of keywords

\category{} %optional, e.g. invited paper

\relatedversion{} %optional, e.g. full version hosted on arXiv, HAL, or other respository/website
\funding{This work was partially supported by project grants PICT-2021-I-A-00090, PICT-2021-I-INVI-00602,
PIP 11220200100368CO, PICT 2019-1272, PUNQ 2218/22, and PUNQ 2219/22.}%optional, to capture a funding statement, which applies to all authors. Please enter author specific funding statements as fifth argument of the \author macro.

\acknowledgements{To Giulio Manzonetto for fruitful discussions that led to the development of this work. To Eduardo Bonelli and the anonymous reviewers for feedback on earlier versions of this paper. The second author would like to thank his advisors Alejandro Díaz-Caro and Pablo E. Martínez López.}

\nolinenumbers %uncomment to disable line numbering
\hideLIPIcs

\EventEditors{John Q. Open and Joan R. Access}
\EventNoEds{2}
\EventLongTitle{42nd Conference on Very Important Topics (CVIT 2016)}
\EventShortTitle{CVIT 2016}
\EventAcronym{CVIT}
\EventYear{2016}
\EventDate{December 24--27, 2016}
\EventLocation{Little Whinging, United Kingdom}
\EventLogo{}
\SeriesVolume{42}
\ArticleNo{23}
\begin{document}

\maketitle

\begin{abstract}
This paper defines two decreasing measures for terms of the simply typed
$\lambda$-calculus, called the $\meassym$-measure and the $\amesym$-measure.
A decreasing measure is a function that
maps each typable $\lambda$-term to an element of a
well-founded ordering, in such a way that
contracting any $\beta$-redex decreases the value of the function,
entailing strong normalization.
Both measures are defined constructively, relying on an auxiliary calculus,
a non-erasing variant of the $\lambda$-calculus.
In this system, dubbed the $\lambdaG$-calculus,
each $\beta$-step creates a ``wrapper'' containing a copy of
the argument that cannot be erased and cannot interact with the context
in any other way.
Both measures rely crucially on the observation, known to Turing
and Prawitz, that contracting a redex cannot create redexes of higher degree,
where the degree of a redex is defined as the height of the type of its
$\lambda$-abstraction.
The $\meassym$-measure maps each $\lambda$-term to a natural number,
and it is obtained by evaluating the term in the $\lambdaG$-calculus
and counting the number of remaining wrappers.
The $\amesym$-measure maps each $\lambda$-term to a structure
of nested multisets, where the nesting depth is proportional
to the maximum redex degree.
\end{abstract}

\section{Introduction}

In this paper we revisit a fundamental question,
that of {\bf strong normalization} of the simply typed $\lambda$-calculus (STLC).
We begin by recalling
that a reduction relation is {\em weakly normalizing} (WN)
if every term can be reduced to normal form in a finite number
of steps,
whereas it is {\em strongly normalizing} (SN) if there are no infinite
reduction sequences ($a_1 \to a_2 \to a_3 \to \hdots$).
Let us review three proof techniques for proving strong
normalization of the STLC.
\medskip

%%%%%%%%%%%%%%%%%%%%%%%%%%%%%%%%%%%%%%%%

One of the better known ways to prove that the STLC 
is SN is through arguments based on {\bf reducibility models}.
The idea is to interpret each type $\typ$ as a set $\semtyp{\typ}$
of strongly normalizing terms,
and to prove that each term $\ltm$ of type $\typ$
is an element of $\semtyp{\typ}$.
Many variants of these ideas can be found in the literature,
including Girard's reducibility candidates~\cite{thesisgirard}
and Tait's saturated sets~\cite{Tait75}.
These techniques provide relatively succint proofs and
they generalize well to extensions of the STLC,
\eg to dependent type theory~\cite{Coquand19} or
classical calculi~\cite{DownenJA20}.
On the other hand, the abstract nature of reducibility arguments
does not provide a ``tangible'' insight
on why a $\beta$-reduction step brings a term closer to normal form.
More specifically, reducibility arguments do not construct explicit
{\bf decreasing measures}. By decreasing measure we mean a function
``$\#$'' mapping each $\lambda$-term to a well-founded ordering $(X,>)$
such that $\ltm \tobeta \ltmtwo$ implies $\#(M) > \#(N)$.
\medskip

%%%%%%%%%%%%%%%%%%%%%%%%%%%%%%%%%%%%%%%%

Another way to prove strong normalization is based on {\bf redex degrees}.
A {\em redex} in the STLC is an applied abstraction,
\ie a term of the form $(\lam{\var}{\ltm})\,\ltmtwo$.
The {\em degree} of a redex is defined as the {\em height}
of the type of its abstraction.
A crucial observation,
that can be attributed to an unpublished note of Turing
(as reported by Gandy~\cite{gandy80turing}; see also~\cite{BarendregtM13}),
is that {\em contracting a redex cannot create a redex of higher or equal degree}.
Recall that a redex $S$ is {\em created} by the contraction of a redex $R$
if $S$ has no {\em ancestor} before $R$.
Indeed, as shown by Lévy~\cite{Tesis:Levy:1978}, in the $\lambda$-calculus,
redexes can be created in exactly one of the three ways below:
\[
  \begin{array}{rrcl}
    \Item{1} &
      (\ulam{\var}{\var})\,(\lam{\vartwo}{\ltm})\,\ltmtwo
      & \tobeta &
      (\ulam{\vartwo}{\ltm})\,\ltmtwo
  \\
    \Item{2} & 
      (\ulam{\var}{\lam{\vartwo}{\ltm}})\,\ltmtwo\,\ltmthree
      & \tobeta &
      (\ulam{\vartwo}{\ltm\sub{\var}{\ltmtwo}})\,\ltmthree
  \\
    \Item{3} & 
      (\ulam{\var}{\hdots \var\,\ltm \hdots})\,(\lam{\vartwo}{\ltmtwo})
      & \tobeta &
      \hdots
      (\ulam{\vartwo}{\ltmtwo})\,\ltm\sub{\var}{\lam{\vartwo}{\ltmtwo}}
      \hdots
  \end{array}
\]
where we underline the $\lambda$ of the contracted redex on the left, 
and the $\lambda$ of the created redex on the right.
In each of these cases, it can be seen that the degree of the created
redex is strictly lower than the degree of the contracted redex.
For instance, in creation case~\Item{1},
the type of the contracted redex
is of the form $(\typ\to\typtwo)\to(\typ\to\typtwo)$,
while the type of the created redex is $\typ\to\typtwo$,
so the height strictly decreases.

With this fact in mind, for each term $\ltm$
one can define
what we call {\bf Turing's measure},
\ie the multiset $\turingme{\ltm}$ of the degrees
of all the redexes of $\ltm$.
One may hope that any reduction step $\ltm \tobeta \ltmtwo$
decreases the measure, \ie $\turingme{\ltm} \mgt \turingme{\ltmtwo}$,
where ``$\mgt$'' is the usual
well-founded multiset ordering
induced by the ordering $(\Nat,>)$ of its elements~\cite{dershowitz1979proving}.
Unfortunately, this is not the case:
even though contracting a redex can only create redexes of strictly lower
degree, it can still make an arbitrary number of {\em copies}
of redexes of arbitrarily large degrees.
% The following example illustrates that contracting a redex of degree $1$ may
% duplicate a redex $I\,\ltm^\typtwo$ of arbitrarily high degree,
% where $I := \lam{\var}{\var}$.
% We use superscripts to indicate the types of some of the subterms:
% \[
%    (\lam{\var^\btyp}{\vartwo^{\btyp\to\btyp\to\btyp}\,\var^\btyp\,\var^\btyp})
%      (\varthree^{\typtwo\to\btyp}\,(I\,\ltm^\typtwo))
%  \,\,\tobeta\,\,
%    \vartwo^{\btyp\to\btyp\to\btyp}
%      \,(\varthree^{\typtwo\to\btyp}\,(I\,\ltm^\typtwo))
%      \,(\varthree^{\typtwo\to\btyp}\,(I\,\ltm^\typtwo))
% \]
% Note that, indeed, the outermost abstraction is of type $\btyp\to\btyp$
% but it makes two copies of a redex $I\,\ltm$ of arbitrary type $\typtwo$.

In his notes, Turing observed that one can follow a reduction strategy
that always selects the {\em rightmost} redex of highest degree.
This strategy ensures that the contracted redex does not copy redexes
of higher or
equal degree, which makes the $\turingme{-}$ measure strictly decrease,
thus proving that the $\lambda$-calculus is WN.
An even simpler measure that also decreases, if one follows this strategy, is
$\turingmep{\ltm} = (D,n)$, where $D$ is the maximum degree of the redexes in $\ltm$
and $n$ is the number of redexes of degree $D$ in $\ltm$.
Similar ideas were exploited by Prawitz~\cite{prawitz1965natural}
and Gentzen (as reported by von Plato~\cite{plato08bsl})
to normalize proofs in natural deduction.
After WN has been established, an indirect proof of SN can be obtained
by translating each typable $\lambda$-term $\ltm$ to
a typable term $\ltm'$ of the {\em $\lambda{I}$-calculus};
see for instance~\cite[Section 3.5]{sorensen2006lectures}.
% Recall that the $\lambda{I}$-calculus is the restriction of
% the $\lambda$-calculus in which every abstraction $\lam{\var}{\ltm}$
% contains at least one free occurrence of $\var$ in the body $\ltm$.
% One key aspect of the $\lambda{I}$-calculus is
% that it is {\em non-erasing},
% which in particular means that a term is WN if and only if it is SN~\cite[Section II.5]{Tesis:Klop}.

In summary, redex degrees can be used to define
concrete measures such as $\turingme{\ltm}$
and $\turingmep{\ltm}$,
that are computable in linear time
and decrease when following a particular reduction strategy.
As already mentioned,
these measures do not necessarily decrease when contracting arbitrary $\beta$-redexes.
\medskip

%%%%%%%%%%%%%%%%%%%%%%%%%%%%%%%%%%%%%%%%

A third way to prove SN relies
on an interpretation that maps terms to {\bf increasing functionals}.
This approach was pioneered by Gandy~\cite{gandy80sn}
and refined by de Vrijer~\cite{dv1987exactly}.
Each type $\typ$ is mapped to a partially ordered set $\semtyp{\typ}$.
Specifically, base types are mapped to $(\Nat,\leq)$,
and $\semtyp{\typ \to \typtwo}$ is defined as the set of
strictly increasing functions $\semtyp{\typ} \to \semtyp{\typtwo}$,
partially ordered by the point-wise order.
Each term $\ltm$ of type $\typ$ is interpreted as
an element $\sem{\ltm} \in \semtyp{\typ}$.
Moreover, an element $f \in \semtyp{\typ}$
can be projected to a natural number $f\star \in \Nat$
in such a way that $\ltm \tobeta \ltmtwo$
implies $\sem{\ltm}\star > \sem{\ltmtwo}\star$.
This indeed provides a decreasing measure.
One of the downsides of this measure
is that computing $\sem{\ltm}\star$
is essentially as difficult as evaluating $\ltm$,
because $\sem{\ltm}$ is defined as a higher-order functional
with a similar structure as the $\lambda$-term $\ltm$ itself.

\medskip
%%%%%%%%%%%%%%%%%%%%%%%%%%%%%%%%%%%%%%%%%%%%%%%%%%%%%%%%%%%%%%%%%%%%%%%%%%%%%%%%
%%%%%%%%%%%%%%%%%%%%%%%%%%%%%%%%%%%%%%%%%%%%%%%%%%%%%%%%%%%%%%%%%%%%%%%%%%%%%%%%

In this work we propose {\bf two decreasing measures for the STLC},
that we dub the $\meassym$-measure and the $\amesym$-measure,
and we prove that they are decreasing.
An ideal decreasing measure should fulfill multiple (partly subjective)
requirements:
\Item{1.}~
  the measure should be easy to calculate,
  in terms of computational complexity;
\Item{2.}~
  its codomain (a well-founded ordering)
  should be simple, in terms of its ordinal type;
\Item{3.}~
  it should give us insight on why $\beta$-reduction terminates;
\Item{4.}~
  it should be easy to prove that the measure is decreasing.
A measure that excels simultaneously at all these requirements is elusive,
and perhaps unattainable.
The proposed measures have different strengths and weaknesses.

\subparagraph{Contributions and structure of this document}
The $\meassym$-measure and the $\amesym$-measure are defined by means
of on an auxiliary calculus that we dub the $\lambdaG$-calculus.
The remainder of the paper is structured as follows.
\medskip

%%%%%%%%%%%%%%%%%%%%%%%%%%%%%%%%%%%%%%%%

In~{\bf \rsec{lambdaG}}
we {\bf define the $\lambdaG$-calculus}.
It is an extension of the STLC with terms\footnote{Note that terms of the
$\lambdaG$-calculus are ranged over by $\tm,\tmtwo,\hdots$
(rather than $\ltm,\ltmtwo,\hdots$).}
of the form $\bin{\tm}{\tmtwo}$,
called {\em wrappers}.
A wrapper $\bin{\tm}{\tmtwo}$ should be understood as
essentially the term $\tm$
in which $\tmtwo$ is a {\em memorized term}, that is, leftover garbage that can
be reduced but cannot interact with the context in any way.
The type of $\bin{\tm}{\tmtwo}$ is the same as the type of $\tm$,
disregarding the type of $\tmtwo$.

The $\beta$-reduction rule is modified so that
contracting a redex $(\lam{\var}{\tm})\,\tmtwo$,
besides substituting the free occurrences of $\var$ by $\tmtwo$ in $\tm$,
produces a wrapper that contains a copy of the argument~$\tmtwo$.
The reduction rule is
$
  (\lam{\var}{\tm})\garb{\tmthree_1}\hdots\garb{\tmthree_n}\,\tmtwo
  \,\,\tog\,\,
  \tm\sub{\var}{\tmtwo}\garb{\tmtwo}\garb{\tmthree_1}\hdots\garb{\tmthree_n}
$.
Note that we allow the presence of an arbitrary number of memorized terms
mediating between the abstraction and the application.
This is to avoid memorized terms {\em blocking} redexes.
For example, if $I = \lam{\var}{\var}$:
\[
  (\lam{\var}{\var(\var\vartwo)})I
  \tog (I(I\vartwo))\garb{I}
  \tog (I\vartwo)\garb{I\vartwo}\garb{I}
  \tog (I\vartwo)\garb{\vartwo\garb{\vartwo}}\garb{I}
  \tog \vartwo\garb{\vartwo}\garb{\vartwo\garb{\vartwo}}\garb{I}
\]
Then we study some syntactic properties of $\lambdaG$.
In particular, we define a relation $\tm \shone \tmtwo$
of {\em forgetful reduction}, meaning that $\tmtwo$ is obtained from $\tm$
by erasing one memorized subterm.
For example,
$\var\,\garb{\var\garb{\vartwo}}\garb{\vartwo\garb{\varthree}}
\shone \var\,\garb{\vartwo\garb{\varthree}}$.
Forgetful reduction is used as a technical tool to prove that the measures
are decreasing in the following sections.

\medskip

%%%%%%%%%%%%%%%%%%%%%%%%%%%%%%%%%%%%%%%%

In~{\bf \rsec{meas_measure}},
we {\bf propose the $\meassym$-measure}~(\rdef{meassym_measure}),
and we prove that it is decreasing.
To define the $\meassym$-measure,
we resort to an operation $\simp{d}{\tm}$
that simultaneously
contracts all the redexes of degree $d$ in a term of the $\lambdaG$-calculus,
that is, the result of the {\em complete development}
of all the redexes of degree~$d$.
The degree of a redex
$(\lam{\var}{\tm})\garb{\tmthree_1}\hdots\garb{\tmthree_n}\,\tmtwo$
is defined similarly as for the STLC,
as the height of the type of the abstraction.
To calculate the $\meassym$-measure
of a $\lambda$-term $\ltm$,
let $D$ be the maximum degree of the redexes in $\ltm$,
and define $\meas{\ltm}$
as the number of wrappers in $\simp{1}{\simp{2}{\hdots\simp{D}{\ltm}}}$.
For example, if
$\ltm = (\lam{\var}{\var\,(\var\,\vartwo)})\,(\lam{\varthree}{\varfour})$,
it turns out that
$\simp{1}{\simp{2}{\ltm}} =
  \bin{
    \bin{
      \varfour
    }{
      \bin{
        \varfour
      }{
        \vartwo
      }
    }
  }{
    \lam{\varthree}{\varfour}
  }$
which has three wrappers,
so $\meas{\ltm} = 3$.
The $\meassym$-measure maps each typable $\lambda$-term
to a natural number.
The main result of~\rsec{meas_measure}
is~\rthm{z_measure_decreases},
stating that {\bf $\meassym$ is decreasing},
\ie that $\ltm \tobeta \ltmtwo$ implies $\meas{\ltm} > \meas{\ltmtwo}$.

\medskip

%%%%%%%%%%%%%%%%%%%%%%%%%%%%%%%%%%%%%%%%

In {\bf \rsec{reduction_by_degrees}} we study {\bf reduction by degrees},
a restricted notion of reduction in the $\lambdaG$-calculus,
written $\tm \tod{d} \tmtwo$,
meaning that $\tm$ reduces to $\tmtwo$ by contracting a redex of degree $d$.
This section contains technical commutation, termination, and postponement results.

\medskip

%%%%%%%%%%%%%%%%%%%%%%%%%%%%%%%%%%%%%%%%

In~{\bf \rsec{a_measure}},
we {\bf propose the $\amesym$-measure}, and we prove that it is decreasing.
To define the $\amesym$-measure,
we define two auxiliary measures $\ame{D}{\tm}$ and $\bme{D}{\tm}$,
indexed by a natural number $D \in \Natz$, mutually recursively:
\begin{itemize}
\item
  $\ame{D}{\tm}$
  is the multiset of pairs $(d,\bme{d}{\tm})$,
  for each redex occurrence of degree $d \leq D$ in $\tm$;
\item
  $\bme{D}{\tm}$
  is the multiset of elements $\ame{D-1}{\tm'}$,
  for each reduction sequence $\tm \rtod{D} \tm'$.
\end{itemize}
The measure $\ame{D}{\tm}$ is defined for every $D \geq 0$,
while $\bme{D}{\tm}$ is defined only for $D \geq 1$.
Multisets are ordered according to the usual multiset ordering,
and pairs according to the lexicographic ordering.
To calculate the $\amesym$-measure of a $\lambda$-term $\ltm$,
let $D$ be the maximum degree of the redexes in $\ltm$,
and define $\ame{}{\ltm} \eqdef \ame{D}{\ltm}$.
The measure $\ame{}{\ltm}$ yields a structure of
nested multisets of nesting depth at most $2D$.
The main theorem of \rsec{meas_measure}
is~\rthm{ame_decreasing},
stating that {\bf $\amesym$ is decreasing},
\ie that $\ltm \tobeta \ltmtwo$ implies $\ame{}{\ltm} > \ame{}{\ltmtwo}$.

\medskip
%%%%%%%%%%%%%%%%%%%%%%%%%%%%%%%%%%%%%%%%

Finally, in~{\bf \rsec{conclusion}}, we conclude.

\section{The $\lambdaG$-calculus}
\lsec{lambdaG}

As mentioned in the introduction, the $\lambdaG$-calculus is an extension
of the STLC in which the $\beta$-reduction rule keeps an extra memorized
copy of the argument in a ``wrapper'' $\bin{\tm}{\tmtwo}$,
in such a way that contracting a redex
like $(\lam{\var}{\tm})\,\tmtwo$ does not erase $\tmtwo$,
even if $\var$ does not occur free in $\tm$.
In this section we define the $\lambdaG$-calculus and we prove some
of the properties that are needed in the following sections
to prove that the $\meassym$-measure
and the $\amesym$-measure are decreasing.
In particular, we discuss {\em subject reduction}~(\rprop{subject_reduction})
and {\em confluence}~(\rprop{confluence});
we define an operation of {\bf simplification}~(\rdef{simplification})
which turns out to calculate
the normal form of a term~(\rprop{simplification_is_normalization});
and we define the relation called {\bf forgetful reduction}~(\rdef{shrinking}),
which is shown to commute with reduction~(\rprop{shrinking_simulation}).

\medskip

First we fix the notation and nomenclature.
{\em Types} of the STLC are either
base types ($\btyp,\btyptwo,\hdots$) or arrow types ($\typ \to \typtwo$).
{\em Terms} are either variables~($\var^\typ,\vartwo^\typ,\hdots$),
abstractions~($\lam{\var^\typ}{\ltm}$),
or applications~($\ltm\,\ltmtwo$), with the usual typing rules.
Terms are defined up to $\alpha$-renaming of bound variables.
We adopt an {\em à la} Church presentation of the STLC,
but we omit most type decorations on variables as long as there is little
danger of confusion.
The $\beta$-reduction rule is
$(\lam{\var}{\ltm})\,\ltmtwo \tobeta \ltm\sub{\var}{\ltmtwo}$
where $\ltm\sub{\var}{\ltmtwo}$ is the capture-avoiding substitution
of the free occurrences of $\var$ in $\ltm$ by $\ltmtwo$.

\subparagraph{The $\lambdaG$-calculus: syntax and reduction}
The set of {\em $\lambdaG$-terms} ---or just {\em terms}---
is given by
$
  \tm,\tmtwo,\hdots
      ::= \var^\typ
     \mid \lam{\var^\typ}{\tm}
     \mid \tm\,\tmtwo
     \mid \bin{\tm}{\tmtwo}
$.
The four kinds of terms are respectively called
{\em variables}, {\em abstractions}, {\em applications},
and {\em wrappers}.
In a wrapper $\bin{\tm}{\tmtwo}$,
the subterm $\tm$ is called the {\em body}
and $\tmtwo$ is called the {\em memorized term}.
As in the STLC, we usually omit type annotations
and terms are regarded up to $\alpha$-renaming.
A {\em context} is a term $\gctx$ with a single free
occurrence of a distinguished variable~$\ctxhole$,
and $\of{\gctx}{\tm}$ is the variable-capturing substitution
of the occurrence of $\ctxhole$ in $\gctx$ by $\tm$.

{\em Typing judgments} are of the form $\judg{\tctx}{\tm}{\typ}$
where $\tctx$ is a partial function mapping variables to types.
Derivable typing judgments are defined by the following rules:
\[
  {\small
  \indrule{}{
         \vphantom{\judg{\tctx}{}{}}
  }{
    \judg{\tctx,\var:\typ}{\var^\typ}{\typ}
         \vphantom{\bin{}{}}
  }
  \HS
  \indrule{}{
    \judg{\tctx,\var:\typ}{\tm}{\typtwo}
  }{
    \judg{\tctx}{\lam{\var^\typ}{\tm}}{\typ\to\typtwo}
         \vphantom{\bin{}{}}
  }
  \HS
  \indrule{}{
    \judg{\tctx}{\tm}{\typ\to\typtwo}
    \HS
    \judg{\tctx}{\tmtwo}{\typ}
  }{
    \judg{\tctx}{\tm\,\tmtwo}{\typtwo}
         \vphantom{\bin{}{}}
  }
  \HS
  \indrule{}{
    \judg{\tctx}{\tm}{\typ}
    \HS
    \judg{\tctx}{\tmtwo}{\typtwo}
  }{
    \judg{\tctx}{\bin{\tm}{\tmtwo}}{\typ}
  }
  }
\]
A term $\tm$ is {\em typable} if $\judg{\tctx}{\tm}{\typ}$ holds
for some $\tctx$ and some $\typ$. Unless otherwise specified, when we
speak of ``terms'' we mean ``typable terms''.
It is straightforward to show that a typable term has a unique type.
We write $\typeof{\tm}$ for the type of $\tm$.

A {\em memory}, written $\sctx$, is a list of memorized terms,
given by the grammar
$\sctx ::= \ctxhole \mid \bin{\sctx}{\tm}$.
If $\tm$ is a term and $\sctx$ is a memory, we write
$\tm\sctx$ for the term that results from appending all the memorized terms
in $\sctx$ to $\tm$, that is,
$(\tm)(\ctxhole\garb{\tmtwo_1}\hdots\garb{\tmtwo_n}) =
\tm\garb{\tmtwo_1}\hdots\garb{\tmtwo_n}$.
We write $\tm\sub{\var}{\tmtwo}$
for the operation of capture-avoiding substitution of the free occurrences
of $\var$ in $\tm$ by $\tmtwo$.
The $\lambdaG$-calculus is the rewriting system whose objects
are typable $\lambdaG$-terms, endowed with the following 
notion of reduction, closed by compatibility under arbitrary contexts:
\begin{definition}[Reduction in the $\lambdaG$-calculus]
$
  (\lam{\var}{\tm})\sctx\,\tmtwo
  \,\,\tog\,\,
  \bin{\tm\sub{\var}{\tmtwo}}{\tmtwo}\sctx
$
\end{definition}
Abstractions followed by lists of memorized terms,
\ie terms of the form $(\lam{\var}{\tm})\sctx$,
are called {\em $\symg$-abstractions}.
Note that all abstractions are also $\symg$-abstractions, as $\sctx$ may
be empty.
A {\em redex} is an expression matching the left-hand side of the
$\tog$-reduction rule, which must be an {\em applied $\symg$-abstraction},
\ie a term of the form
$(\lam{\var}{\tm})\sctx\,\tmtwo$.
The {\em height} of a type is given by
$\height{\btyp} \eqdef 0$
and
$
  \height{\typ\to\typtwo} \eqdef 1+\max(\height{\typ},\height{\typtwo})
$.
The {\em degree} of a $\symg$-abstraction $(\lam{\var}{\tm})\sctx$
is defined as the height of its type;
note that this number is always strictly positive, since the type
must be of the form $\typ\to\typtwo$.
Moreover, this type is unique, so the operation is well-defined.
The {\em degree} of a redex
$(\lam{\var}{\tm})\sctx\,\tmtwo$
is defined as the degree of the $\symg$-abstraction
$(\lam{\var}{\tm})\sctx$.
The {\em max-degree} of a term $\tm$
is written $\maxdeg{\tm}$
and it is defined as the maximum degree of the redexes in $\tm$,
or $0$ if $\tm$ has no redexes.
The {\em weight} $\weight{\tm}$ of a $\lambdaG$-term $\tm$
is the number of wrappers in $\tm$.

\begin{example}
\lexample{reduction_in_lambdaG}
Let $0$ be a base type
and let
$\tm := (\lam{\var^{\bbtyp\to\bbtyp}}{
  \lam{\vartwo^\bbtyp}{
    \bin{\vartwo^\bbtyp}{
      \var^{\bbtyp\to\bbtyp}\,(\var^{\bbtyp\to\bbtyp}\,\varthree^\bbtyp)
    }
   }})\,I\,\varfour^\bbtyp$,
where $I := \lam{\var^\bbtyp}{\var^\bbtyp}$.
One possible way to reduce $\tm$ is:
\[
  \begin{array}{r@{\,\,\,\,}l@{\,\,\,\,}l@{\,\,\,\,}l@{\,\,\,\,}l}
    (\lam{\var}{\lam{\vartwo}{\bin{\vartwo}{\var\,(\var\,\varthree)}}})\,I\,\varfour
  & \tog &
    (\lam{\vartwo}{\bin{\vartwo}{I\,(I\,\varthree)}})\garb{I}\,\varfour
  & \tog &
    \bin{\varfour}{I\,(I\,\varthree)}\garb{\varfour}\garb{I}
  \\
  & \tog &
    \bin{\varfour}{I\,(\varthree\garb{\varthree})}\,\garb{\varfour}\,\garb{I}
  & \tog &
    \bin{\varfour}{
      \varthree\garb{\varthree}\garb{\varthree\garb{\varthree}}
    }\garb{\varfour}\garb{I} = \tmtwo
  \end{array}
\]
The degrees of the redexes contracted in each step are $2$, $1$, $1$, and $1$,
in that order. Note that $\maxdeg{\tm} = 2$ and that the weight of the
resulting term is $\weight{\tmtwo} = 6$.
\end{example}

Two basic properties of the $\lambdaG$-calculus are
subject reduction and confluence.
These are immediate consequences of the fact that the $\lambdaG$-calculus
can be understood as an {\em orthogonal HRS} in the sense of
Nipkow~\cite{nipkow1991higher},
\ie a left-linear higher-order rewriting system without critical pairs.

\begin{proposition}[Subject reduction]
\lprop{subject_reduction}
Let $\judg{\tctx}{\tm}{\typ}$ and $\tm \tog \tmtwo$.
Then $\judg{\tctx}{\tmtwo}{\typ}$.
\end{proposition}

\begin{proposition}[Confluence]
\lprop{confluence}
If $\tm_1 \tog^* \tm_2$ and $\tm_1 \tog^* \tm_3$,
there exists a term $\tm_4$ such that
$\tm_2 \tog^* \tm_4$ and $\tm_3 \tog^* \tm_4$.
\end{proposition}

\subparagraph{Full simplification}
Next, we define an operation written $\simpfull{\tm}$
and called {\em full simplification}.

Let $d \geq 1$ be a natural number.
The {\em simplification of degree $d$}, written $\simp{d}{\tm}$,
is the result of simultaneously contracting
all the redexes of degree $d$ in $\tm$,
that is,
the result of the {\em complete development} of all redexes of degree $d$.
Formally, for each $\lambdaG$-term $\tm$
we define $\simp{d}{\tm}$,
and, for each memory $\sctx$, we define $\simp{d}{\sctx}$
as follows:
\begin{definition}[Simplification]
\ldef{simplification}
\[
  \begin{array}{r@{\,}c@{\,}l}
    \simp{d}{\var}
    & \eqdef &
    \var
  \\
    \simp{d}{\lam{\var}{\tm}}
    & \eqdef &
    \lam{\var}{\simp{d}{\tm}}
  \\
    \simp{d}{\tm\,\tmtwo}
    & \eqdef &
    \begin{cases}
        \simp{d}{\tm'}
          \sub{\var}{\simp{d}{\tmtwo}}
          \garb{\simp{d}{\tmtwo}}
          \simp{d}{\sctx}
      &
          \text{if $\tm = (\lam{\var}{\tm'})\sctx$ and it is of degree $d$}
    \\
      \simp{d}{\tm}\,\simp{d}{\tmtwo}
      &
        \text{otherwise}
    \end{cases}
  \\
    \simp{d}{\bin{\tm}{\tmtwo}}
    & \eqdef &
    \bin{\simp{d}{\tm}}{\simp{d}{\tmtwo}}
  \end{array}
\]
\end{definition}
where if $\sctx$ is a memory, $\simp{d}{\sctx}$
is defined by $\simp{d}{\ctxhole} \eqdef \ctxhole$
and $\simp{d}{\sctx\garb{\tm}} \eqdef \simp{d}{\sctx}\garb{\simp{d}{\tm}}$.
Furthermore, if $\tm$ is a $\lambdaG$-term of max-degree $D$, we define
the {\em full simplification} of $\tm$ as the term that results from
iteratively taking the simplification of degree $i$ from $D$ down to $1$.
More precisely,
$
  \simpfull{\tm} \eqdef \simp{1}{\hdots\simp{D-1}{\simp{D}{\tm}}}
$.

\begin{example}
\lexample{full_simplification}
Consider the $\lambda$-term
$\ltm = (\lam{\var^{\bbtyp\to\bbtyp}}{\var^{\bbtyp\to\bbtyp}(\var^{\bbtyp\to\bbtyp}\,\vartwo^{\bbtyp})})(\lam{\varthree^{\bbtyp}}{\varfour^{\bbtyp}})$.
It can be regarded also as a $\lambdaG$-term, and we have:
\[
 \begin{array}{rcl}
    \simp{2}{\ltm}
    & = &
      \bin{
        (
          (\lam{\varthree^\bbtyp}{\varfour^\bbtyp})
        \,
          ((\lam{\varthree^\bbtyp}{\varfour^\bbtyp})\,\vartwo^\bbtyp)
        )
      }{
        \lam{\varthree^\bbtyp}{\varfour^\bbtyp}
      }
  \\
    \simpfull{\ltm} = \simp{1}{\simp{2}{\ltm}}
    & = &
      \bin{
        \bin{
          \varfour^\bbtyp
        }{
          \bin{
            \varfour^\bbtyp
          }{
            \vartwo^\bbtyp
          }
        }
      }{
        \lam{\varthree^\bbtyp}{\varfour^\bbtyp}
      }
  \end{array}
\]
Note that $\ltm$ has only one redex,
whose abstraction is of type $(\bbtyp\to\bbtyp)\to\bbtyp$ 
and hence of degree $2$,
and that $\simp{2}{\ltm}$
has two redexes, whose abstractions are of type $\bbtyp\to\bbtyp$ 
and hence of degree $1$.
Moreover, consider the $\lambda$-term
$\ltmtwo = (\lam{\varthree^{\bbtyp}}{\varfour^{\bbtyp}})\,((\lam{\varthree^{\bbtyp}}{\varfour^{\bbtyp}})\,\vartwo^{\bbtyp})$. Then
$
  \simpfull{\ltmtwo} = \simp{1}{\ltmtwo}
  =
    \bin{
      \varfour
    }{
      \bin{
        \varfour
      }{
        \vartwo
      }
    }
$.
Note that $\ltmtwo$
has two redexes whose abstraction is of type $\bbtyp\to\bbtyp$ 
and hence of degree $1$.
As an additional note, in the $\lambda$-calculus
there is a reduction step $\ltm \tobeta \ltmtwo$,
and we have that
$\weight{\simpfull{\ltm}}
 = 3
 > 2
 = \weight{\simpfull{\ltmtwo}}$.
So this example illustrates that the $\meassym$-measure
(as defined in~\rdef{meassym_measure}) is decreasing
(as we will show in~\rthm{z_measure_decreases}).
\end{example}

As it turns out, {\bf full simplification corresponds
to reduction to normal form}. More precisely, we have the
following result, which entails in particular that the $\lambdaG$-calculus
is weakly normalizing:

\begin{proposition}
\lprop{simplification_is_normalization}
$\tm \tog^* \simpfull{\tm}$,
and moreover $\simpfull{\tm}$ is a $\tog$-normal form.
\end{proposition}
\begin{proof}
To show that $\tm \tog^* \simpfull{\tm}$, it suffices to prove a lemma 
stating that $\tm \tog^* \simp{d}{\tm}$ for all $d \geq 1$.
This implies that
$\tm
  \tog^* \simp{D}{\tm}
  \tog^* \simp{D-1}{\simp{D}{\tm}}
  \hdots
  \tog^* \simp{1}{\hdots\simp{D-1}{\simp{D}{\tm}}} = \simpfull{\tm}$,
where $D$ is the max-degree of $\tm$.
The lemma itself is straightforward by induction on $\tm$.
\\\indent
To show that $\simpfull{\tm}$ is a $\tog$-normal form,
the key property is that, after performing a simplification of order $d$,
no redexes of order $d$ remain.
The reason is that contracting a redex of order $d$ can only
create redexes of lower degree.
More precisely, we prove a key lemma stating that
if $d \geq 1$ and $\maxdeg{\tm} \leq d$,
then $\maxdeg{\simp{d}{\tm}} < d$.
If we let $\maxdeg{\tm} \leq D$,
we can iterate this lemma, to obtain
that $\maxdeg{\simp{D}{\tm}} < D$,
and $\maxdeg{\simp{D-1}{\simp{D}{\tm}}} < D - 1$,
$\hdots$,
and finally
$\maxdeg{\simp{1}{\hdots\simp{D-1}{\simp{D}{\tm}}}} < 1$.
This means that $\simpfull{\tm} = \simp{1}{\hdots\simp{D-1}{\simp{D}{\tm}}}$
does not contain redexes, since there are no redexes of degree $0$,
so $\simpfull{\tm}$ must be a $\tog$-normal form.
\SeeAppendix{See~\rprop{appendix:simplification_is_normalization} in the appendix for detailed proofs.}
\end{proof}

\subparagraph{Forgetful reduction}
To conclude this section, we introduce the relation of {\em forgetful reduction}
$\tm \shone^+ \tmtwo$, and we prove that it commutes with reduction.

\begin{definition}
\ldef{shrinking}
A $\lambdaG$-term $\tm$
{\em reduces via a forgetful step} to $\tmtwo$,
written $\tm \shone \tmtwo$,
according to the following axiom,
closed by compatibility under arbitrary contexts:
\[
  \bin{\tm}{\tmtwo} \shone \tm
\]
We say that $\tm$ {\em reduces via forgetful reduction} to $\tmtwo$
if and only if $\tm \mathrel{\shone^+} \tmtwo$,
where $\shone^+$ denotes the transitive closure of $\shone$.
\end{definition}

\begin{example}
$(\lam{\var}{\var\garb{\vartwo\garb{\vartwo}}})\garb{\varthree\garb{\varthree}}
 \shone
 (\lam{\var}{\var\garb{\vartwo\garb{\vartwo}}})\garb{\varthree}
 \shone
 (\lam{\var}{\var})\garb{\varthree}
 \shone
 \lam{\var}{\var}$.
\end{example}

\begin{proposition}[Forgetful reduction commutes with reduction]
\lprop{shrinking_simulation}
If $\tm \shone^+ \tmtwo$ and $\tm \tog^* \tm'$,
there exists a term $\tmtwo'$ such that
$\tm' \shone^+ \tmtwo'$ and $\tmtwo \tog^* \tmtwo'$.
Furthermore,
if $\tm \shone^+ \tmtwo$ and $\tm$ is a $\tog$-normal form,
then $\tmtwo$ is also a normal form.
\end{proposition}
\begin{proof}
The result can be reduced to a local commutation result,
stating that if $\tm \shone \tmtwo$ and $\tm \tog \tm'$,
there exists a term $\tmtwo'$ such that
$\tm' \shone^+ \tmtwo'$ and $\tmtwo \tog^= \tmtwo'$,
where $\tog^=$ is the reflexive closure of $\tog$.
Local commutation can be proved by case analysis.
The interesting cases are when a shrinking step $\tmtwo \shone \tmtwo'$
lies inside the argument of a redex,
and when a reduction step $\tmfour \tog \tmfour'$ is inside erased garbage:
\[
  \xymatrix@C=.25cm@R=.5cm{
    \app{(\lam{\var}{\tm})\sctx}{\tmtwo}
    \ar[d]
    & \shone &
    \app{(\lam{\var}{\tm})\sctx}{\tmtwo'}
    \ar[d]
  \\
    \tm\sub{\var}{\tmtwo}\garb{\tmtwo}\sctx
    & \shone^+ &
    \tm\sub{\var}{\tmtwo'}\garb{\tmtwo'}\sctx
  }
  \HS\HS\HS\HS
  \xymatrix@C=.25cm@R=.5cm{
    \bin{\tmthree}{\tmfour}
    \ar[d]
    & \shone &
    \tmthree
    \ar@{=}[d]
  \\
    \bin{\tmthree}{\tmfour'}
    & \shone^+ &
    \tmthree
  }
\]
For the last part of the statement,
it suffices to show that
if $\tm \shone \tmtwo$ in one step and $\tm$ is a $\tog$-normal form,
then $\tmtwo$ is also a normal form,
which is straightforward by induction on $\tm$.
\SeeAppendix{See \rprop{appendix:shrinking_simulation} in the appendix
for detailed proofs.}
\end{proof}

Each step in the STLC has a {\em corresponding}
step in the $\lambdaG$-calculus, that contracts the redex in the
same position.
For instance the step $(\lam{\var}{\var\,\vartwo})\,I \tobeta I\,\vartwo$
in the STLC
has a corresponding step $(\lam{\var}{\var\,\vartwo})\,I \tog (I\,\vartwo)\garb{I}$
in the $\lambdaG$-calculus.
In this example, $(I\,\vartwo)\garb{I} \shone I\,\vartwo$.
The following easy lemma confirms that this is a general fact:

\begin{lemma}[Reduce/forget lemma]
\llem{reduce_shrink_lemma}
Let $\ltm \tobeta \ltmtwo$ be a $\beta$-step,
and let $\ltm \tog \tmtwo$ be the corresponding step in $\lambdaG$.
Then $\tmtwo \shone \ltmtwo$.
\end{lemma}

\section{The $\meassym$-measure}
\lsec{meas_measure}

In this section, we {\bf define the $\meassym$-measure}~(\rdef{meassym_measure})
and we {\bf prove that it is decreasing}~(\rthm{z_measure_decreases}).
Let us try to convey some ideas that led to the definition of the
$\meassym$-measure.
Recall that an abstract rewriting system $(A,\to)$
is {\em weakly Church--Rosser} (WCR) if
$\leftarrow\rightarrow \,\subseteq\, \rightarrow^*\leftarrow^*$,
{\em Church--Rosser} (CR) if
$\leftarrow^*\rightarrow^* \,\subseteq\, \rightarrow^*\leftarrow^*$,
and {\em increasing} (Inc)
if there exists a function $\incf{\cdot} : A \to \Nat$ such that
$a \to b$ implies $\incf{a} < \incf{b}$.
Let us also recall
{\em Klop--Nederpelt's lemma}~\cite[Theorem~1.2.3~(iii)]{Terese},
which states that Inc $\land$ WCR $\land$ WN $\implies$ SN $\land$ CR.

Let $(A,\to)$ be increasing and WCR.
Given a reduction $a \to^* b$, where $b$ is a normal form,
we can find a {\em decreasing} measure for
the set of objects reachable from $a$, that is,
the set $\set{c \in A \ST a \to^* c}$.
In fact, by Klop--Nederpelt's lemma,
we know that for every $c \in A$ such that $a \to^* c$
we have that $c \to^* b$,
which implies that $\incf{c} \leq \incf{b}$,
and hence we can define $\#(c) := \incf{b} - \incf{c}$.
It is easy to see that $\#(-)$ is a decreasing measure,
since $c \to c'$ implies that $\incf{c} < \incf{c'}$
so $\#(c) := \incf{b} - \incf{c} > \incf{b} - \incf{c'} = \#(c')$.
Furthermore, the value of $\#(c)$ does not depend on the choice
of $a$, by uniqueness of normal forms.

The idea behind the $\meassym$-measure is that
the construction of a {\em decreasing} measure
can be based on an {\em increasing} measure,
according to the previous observation.
It is not possible to build an increasing measure directly
for the STLC; \eg the following 
infinite sequence of expansions
$\tm \leftarrow I\,\tm \leftarrow I\,(I\,\tm) \leftarrow \hdots$
would induce an infinite decreasing chain of natural numbers
$\incf{\tm} > \incf{I\,\tm} > \incf{I\,(I\,\tm)} > \hdots$.

One could try to define an increasing measure
in a variant of the STLC
such as Endrullis \etal's clocked $\lambda$-calculus~\cite{EndrullisHKP17},
in which
the $\beta$-rule becomes
$(\lam{\var}{\tm})\,\tmtwo \to \counter(\tm\sub{\var}{\tmtwo})$,
that is, contracting a $\beta$-redex produces a counter ``$\counter$'' that
keeps track of the number of contracted redexes.
One could then count the number of $\tau$'s:
for example, in the reduction sequence
$(\lam{\var}{\var\,(\var\,\vartwo)})\,I
 \to \counter(I\,(I\,\vartwo))
 \to \counter\counter(I\,\vartwo)
 \to \counter\counter\counter\vartwo$
the number of counters strictly increases with each step.
Unfortunately, this does not define an increasing measure, due to {\em erasure}.
For example, $(\lam{\var}{\vartwo})\,\tm \to \counter\vartwo$
erases all the counters in $\tm$.

This is the motivation behind the definition of the $\lambdaG$-calculus,
which avoids erasure by always keeping an extra copy of the argument in a wrapper.
The $\lambdaG$-calculus is indeed increasing: in a step $\tm \tog \tmtwo$
one has that $\weight{\tm} < \weight{\tmtwo}$,
where we recall that $\weight{\tm}$ denotes the {\em weight},
\ie the number of wrappers in~$\tm$.
For example, the step
$(\lam{\var}{\vartwo})\,(\bin{\varthree}{\varthree})
 \tog \bin{\vartwo}{\bin{\varthree}{\varthree}}$
increases the number of wrappers.
The decreasing measure $\meas{\ltm}$ is defined essentially
by reducing $\ltm$ to normal form in the $\lambdaG$-calculus
and counting the number of wrappers in the result:

\begin{definition}[The $\meassym$-measure]
\ldef{meassym_measure}
For each typable $\lambda$-term $\ltm$,
define $\meas{\ltm} \eqdef \weight{\simpfull{\ltm}}$.
\end{definition}
As we show below, $\simpfull{\ltm}$ turns out to be
exactly the normal form of $\ltm$ in the $\lambdaG$-calculus.
We insist in writing $\simpfull{\ltm}$
to emphasize that the {\em definition} of the $\meassym$-measure
does not require to prove that the $\lambdaG$-calculus is weakly normalizing.
Indeed, the simplification $\simp{d}{\tm}$
can be defined by structural induction on $\tm$,
and the full simplification
$\simpfull{\tm} = \simp{1}{\simp{2}{\hdots\simp{D}{\tm}}}$
can be calculated in exactly $D$ iterations.
On the other hand, the {\em proof} that the $\meassym$-measure is
decreasing does rely on the fact that
$\simpfull{\ltm}$ is the normal form of $\ltm$.
\medskip

In the remainder of this section, we prove that the $\meassym$-measure
is indeed decreasing.
The following lemma states that forgetful reduction decreases weight,
and it is straightforward to prove:

\begin{lemma}
\llem{shrinking_decrases_weight}
If $\tm \shone^+ \tmtwo$ then $\weight{\tm} > \weight{\tmtwo}$.
\end{lemma}
%\begin{proof}
%It suffices to see that
%if $\tm \shone \tmtwo$ in one step then $\weight{\tm} > \weight{\tmtwo}$.
%This is immediate because forgetful reduction erases a memorized term,
%decreasing the number of wrappers in at least one,
%and $\weight{\tm}$ counts the number of wrappers in $\tm$.
%\end{proof}
The proof that the $\meassym$-measure decreases relies on the two following
properties that relate full simplification $\simpfull{-}$
respectively with reduction ($\tog$) and forgetful reduction ($\shone^+$):
\begin{lemma}
\llem{simpfull_properties}
\Item{1.}
  If $\tm \tog \tmtwo$ then $\simpfull{\tm} = \simpfull{\tmtwo}$.
\Item{2.}
  If $\tm \rhd^+ \tmtwo$ then $\simpfull{\tm} \rhd^+ \simpfull{\tmtwo}$.
\end{lemma}
\begin{proof}
For the first item,
note that by \rprop{simplification_is_normalization},
we know that
$\tm \tog^* \simpfull{\tm}$
and that $\tm \tog \tmtwo \tog^* \simpfull{\tmtwo}$,
where moreover $\simpfull{\tm}$ and $\simpfull{\tmtwo}$
are $\tog$-normal forms.
By confluence~(\rprop{confluence}), this means that
$\simpfull{\tm} = \simpfull{\tmtwo}$.

For the second item,
note that by \rprop{simplification_is_normalization},
we know that
$\tm \tog^* \simpfull{\tm}$.
Since we also know $\tm \shone^+ \tmtwo$ by hypothesis,
and since forgetful reduction commutes with reduction~(\rprop{shrinking_simulation}),
there exists a term $\tmthree$
such that $\tmtwo \tog^* \tmthree$ and $\simpfull{\tm} \shone^+ \tmthree$.
By \rprop{simplification_is_normalization} we know that
$\simpfull{\tm}$ is in normal form, so by \rprop{shrinking_simulation}
$\tmthree$ must also be a normal form.
On the other hand,
by \rprop{simplification_is_normalization} we know that
$\tmtwo \tog^* \simpfull{\tmtwo}$,
where $\simpfull{\tmtwo}$ must also be a normal form.
In summary, we have that 
$\tmtwo \tog^* \tmthree$
and
$\tmtwo \tog^* \simpfull{\tmtwo}$,
where both $\tmthree$ and $\simpfull{\tmtwo}$
are normal forms.
By confluence~(\rprop{confluence}) $\tmthree = \simpfull{\tmtwo}$,
and from this we obtain that
$\simpfull{\tm} \shone^+ \tmthree = \simpfull{\tmtwo}$,
as required.
\end{proof}

\begin{theorem}
\lthm{z_measure_decreases}
Let $\ltm,\ltmtwo$ be typable $\lambda$-terms
such that $\ltm \tobeta \ltmtwo$.
Then $\meas{\ltm} > \meas{\ltmtwo}$.
\end{theorem}
\begin{proof}
Given the step $\ltm \tobeta \ltmtwo$,
consider the corresponding step $\ltm \tog \tmtwo$,
and note that $\tmtwo \shone^+ \ltmtwo$ by
the reduce/forget lemma~(\rlem{reduce_shrink_lemma}).
Since $\ltm \tog \tmtwo \shone^+ \ltmtwo$,
by \rlem{simpfull_properties},
we have that
$\simpfull{\ltm} = \simpfull{\tmtwo} \shone^+ \simpfull{\ltmtwo}$.
Finally, by \rlem{shrinking_decrases_weight},
$\meas{\ltm}
 = \weight{\simpfull{\ltm}}
 > \weight{\simpfull{\ltmtwo}}
 = \meas{\ltmtwo}$.
\end{proof}

The following is one example that the $\meassym$-measure decreases
(see~\rexample{full_simplification} for another example):
\begin{example}
Let
  $\ltm = (\lam{\var^\bbtyp}{
            \vartwo^{\bbtyp\to\bbtyp\to\bbtyp}\,\var^\bbtyp\,\var^\bbtyp
          })
          \,
         ((\lam{\var^{\bbtyp\to\bbtyp}}{\var^{\bbtyp\to\bbtyp}\,\varthree^\bbtyp})
          \,f^{\bbtyp\to\bbtyp})$,
consider the step
$
  \ltm =
  (\lam{\var}{
    \vartwo\,\var\,\var
  })((\lam{\var}{\var\,\varthree})\,f)
  \tobeta
  (\lam{\var}{
    \vartwo\,\var\,\var
  })\,(f\,\varthree) = \ltmtwo
$, and note that
$\meas{\ltm} = \weight{\simpfull{\ltm}} = 4 > 1 = \meas{\ltmtwo}$,
since:
\[
  \simpfull{\ltm} =
    \bin{
      (\vartwo
      \,\bin{(f\,\varthree)}{f}
      \,\bin{(f\,\varthree)}{f})
    }{
      \bin{(f\,\varthree)}{f}
    }
  \HS\HS
  \simpfull{\ltmtwo} =
    \bin{
      (\vartwo
      \,(f\,\varthree)
      \,(f\,\varthree))
    }{
      f\,\varthree
    }
\]
\end{example}

\section{Reduction by degrees}
\lsec{reduction_by_degrees}

This section is of purely technical nature.
The aim is to develop tools
that we use in the following section to reason about the $\amesym$-measure.
To do so, we need to introduce witnesses of steps and
reduction sequences, treating the $\lambdaG$-calculus as
an {\em abstract rewriting system} in the sense of~\cite[Def.~8.2.2]{Terese}
or as a {\em transition system} in the sense of~\cite[Def.~1]{Mellies05}.
{\em Objects} are $\lambdaG$-terms,
{\em steps} are 5-uples $\redex = (\gctx,\var,\tm,\sctx,\tmtwo)$
witnessing the reduction step
$\of{\gctx}{(\lam{\var}{\tm})\sctx\,\tmtwo} \tog
 \of{\gctx}{\tm\sub{\var}{\tmtwo}\garb{\tmtwo}\sctx}$
under a context $\gctx$, and
{\em reductions} ($\redseq,\redseqtwo,\hdots$)
are sequences of composable steps.
Similarly, {\em forgetful steps} are triples $\redex = (\gctx,\tm,\tmtwo)$
witnessing the forgetful reduction $\of{\gctx}{\tm\garb{\tmtwo}} \shone \of{\gctx}{\tm}$,
and {\em forgetful reductions} (also written $\redseq,\redseqtwo,\hdots$)
are sequences of composable forgetful steps.
We write $\src{\redseq}$ and $\tgt{\redseq}$ respectively for the
source and target terms of $\redseq$.

For each $d \in \Natz$,
we define {\bf reduction of degree $d$} as follows:

\begin{definition}
$\tm \tod{d} \tmtwo$
if and only if
$\tm \tog \tmtwo$ by contracting a redex of degree $d$.
\end{definition}
We write $\redex : \tm \tod{d} \tmtwo$
if $\redex$ is a step witnessing a reduction step of degree $d$,
and $\redseq : \tm \rtod{d} \tmtwo$
if $\redseq$ is a reduction witnessing a sequence of
reduction steps of degree $d$.

The following results require to explicitly manipulate steps and reductions.
\SeeAppendix{
We only give sketches of the proofs for lack of space.
See~\rsec{appendix:reduction_by_degrees} in the appendix for
detailed proofs.}

\begin{proposition}[Commutation of reduction by degrees]
\lprop{commutation_by_degrees}
For any two reductions
    $\redseq : \tm_1 \rtod{d} \tm_2$
and $\redseqtwo : \tm_1 \rtod{D} \tm_3$,
there exists a term $\tm_4$
and one can construct
reductions
    $\redseqtwo/\redseq : \tm_2 \rtod{D} \tm_4$
and $\redseq/\redseqtwo : \tm_3 \rtod{d} \tm_4$
such that,
furthermore, if $d \neq D$, then
\Item{1.}
  $\redseq/\redseqtwo$
  contains at least as many steps as $\redseq$;
and
\Item{2.}
  $\redseq/\redseqtwo$ determines $\redseq$,
  that is,
  $\redseq_1/\redseqtwo = \redseq_2/\redseqtwo$
  implies $\redseq_1 = \redseq_2$.
\end{proposition}
\begin{proof}
This is reduced to the fact that the $\lambdaG$-calculus can be understood
as an orthogonal higher-order rewriting system in the sense of Nipkow~\cite{nipkow1991higher}.
Indeed, $\redseq/\redseqtwo$ and $\redseqtwo/\redseq$ can be taken to be
the standard notion of projection based on residuals for orthogonal HRSs.
Note that
item~\Item{1.} holds because the $\lambdaG$-calculus is non-erasing
while
item~\Item{2.} is a consequence of the {\em unique ancestor} property,
\ie each redex {\em descends} from at most one redex.
\end{proof}

\begin{corollary}[Termination of reduction by degrees]
\lcoro{termination_of_tod}
The relation $\tod{d}$ is strongly normalizing.
\end{corollary}
\begin{proof}
This is a consequence of the fact that HRSs enjoy the
Finite Developments property~\cite[Theorem~11.5.11]{Terese},
observing that reduction of degree $d$ does not create redexes of degree $d$.
Alternatively, it can be easily shown that $\tm \rtod{d} \simp{d}{\tm}$
and $\simp{d}{\tm}$ is in $\tod{d}$-normal form,
so $\tod{d}$ is WN.
Moreover, one can observe that $\tod{d}$ is
{\em uniformly normalizing}~\cite{KhasidashviliOO01},
given that there is no erasure, which entails that $\tod{d}$ is SN.
%Moreover, $\tod{d}$ is WCR by~\rprop{commutation_by_degrees},
%and increasing because $\tm \tod{d} \tm'$
%implies $\weight{\tm} < \weight{\tm'}$.
%As noted at the beginning of~\rsec{meas_measure},
%this means that
%$\#(\tm) := \weight{\simpd{\tm}} - \weight{\tm}$ is decreasing,
%which concludes the proof.
\end{proof}

\begin{proposition}[Lifting property for lower steps]
\lprop{retraction_of_higher_degree_steps}
Let $d < D$ and $\tm \tod{d} \tmtwo \rtod{D} \tmtwo'$.
Then there exist terms $\tm',\tmtwo''$
such that $\tm \rtod{D} \tm'$
and $\tmtwo' \rtod{D} \tmtwo''$
and $\tm' \todplus{d} \tmtwo''$.
\end{proposition}
\begin{proof}
Note that $\tm \rtod{D} \simp{D}{\tm}$.
By~\rprop{commutation_by_degrees},
there exists a term $\tmthree$ such that
$\tmtwo \rtod{D} \tmthree$
and
$\simp{D}{\tm} \todplus{d} \tmthree$.
Again by~\rprop{commutation_by_degrees},
there exists $\tmtwo''$
such that $\tmthree \rtod{D} \tmtwo''$
and $\tmtwo' \rtod{D} \tmtwo''$.
Moreover, $\simp{D}{\tm}$ is in $\tod{D}$-normal form.
Since $\simp{D}{\tm} \rtod{d} \tmthree$ with $d < D$
and reduction does not create redexes of higher degree,
$\tmthree$ is also in $\tod{D}$-normal form,
so $\tmthree = \tmtwo''$, and we are done.
%\SeeAppendix{See \rprop{appendix:retraction_of_higher_degree_steps} in the appendix.}
\end{proof}

\begin{proposition}[Postponement of forgetful reduction]
\lprop{retraction_before_shrinking}
For any two reductions
    $\redseq : \tm \shone^* \tm'$
and $\redseqtwo : \tm' \rtod{d} \tmtwo'$,
there exists a term $\tmtwo$
and reductions $\protract{\redseq}{\redseqtwo} : \tmtwo \shone^* \tmtwo'$
and $\retract{\redseqtwo}{\redseq} : \tm \rtod{d} \tmtwo$.
Furthermore, $\retract{\redseqtwo}{\redseq}$ determines $\redseqtwo$,
that is,
$\retract{\redseqtwo_1}{\redseq} = \retract{\redseqtwo_2}{\redseq}$
implies $\redseqtwo_1 = \redseqtwo_2$.
\end{proposition}
\begin{proof}
This can be reduced to an analysis of the critical pairs
between the rewriting rules defining $\shone^{-1}$ and $\tog$.
Critical pairs are of the form
$(\lam{\var}{\tm})\sctx_1\garb{\tmtwo}\sctx_2\,\tmthree
 \shone
 (\lam{\var}{\tm})\sctx_1\sctx_2\,\tmthree
 \tog
 \tm\sub{\var}{\tmthree}\garb{\tmthree}\sctx_1\sctx_2$
and can be closed by
$(\lam{\var}{\tm})\sctx_1\garb{\tmtwo}\sctx_2\,\tmthree
 \tog
 \tm\sub{\var}{\tmthree}\garb{\tmthree}\sctx_1\garb{\tmtwo}\sctx_2
 \shone
 \tm\sub{\var}{\tmthree}\garb{\tmthree}\sctx_1\sctx_2$.
%\SeeAppendix{See \rprop{appendix:retraction_before_shrinking} in the appendix.}
\end{proof}

The following diagrams depict the statements of the three preceding propositions:
\[
  \xymatrix@C=2.5cm{
    \tm_1 \ar_{\redseqtwo}^{D}[d] \ar_{\redseq}^{d}@{->>}[r]
    \ar|{\text{\rprop{commutation_by_degrees}}}@{}[dr]
  &
    \tm_2 \ar_{\redseqtwo/\redseq}^{D}@{.>>}[d]
  \\
    \tm_3 \ar_{\redseq/\redseqtwo}^{d}@{.>>}[r]
  &
    \tm_4
  }
  \HS\HS
  \xymatrix{
    \ar|{\text{\rprop{retraction_of_higher_degree_steps}}}@{}[drr]
    \tm \ar^{d}[d] \ar^{D}@{.>>}[rr] & & \tm' \ar^{d}@{.>}^>{+}[d]
  \\
    \tmtwo \ar^{D}@{->>}[r] & \tmtwo' \ar^{D}@{.>>}[r] & \tmtwo''
  }
  \HS\HS
  \xymatrix@C=2cm{
    \ar|{\text{\rprop{retraction_before_shrinking}}}@{}[dr]
    \tm \ar_{\redseq}@{}|{\shone^*}[r]
        \ar^{d}_{\retract{\redseqtwo}{\redseq}}@{.>>}[d]
  &
    \tm' \ar^{d}_{\redseqtwo}@{->}^>{*}[d]
  \\
    \tmtwo \ar_{\protract{\redseq}{\redseqtwo}}@{}|{\shone^*}[r]
  &
    \tmtwo'
  }
\]

\section{The $\amesym$-measure}
\lsec{a_measure}

In this section, we {\bf define the $\amesym$-measure}~(\rdef{amesym_measure})
and
we {\bf prove that it is decreasing}~(\rthm{ame_decreasing}).
We start with some preliminary notions.

A partially ordered set $(X,>)$ is {\em well-founded}
if there are no infinite decreasing
chains. $\Multi{X}$ denotes the set of {\em finite multisets} over a set $X$,
which are functions
$\mset : X \to \Natz$ such that $\mset(x) > 0$ for finitely many values of $x \in X$.
We write $\mset+\msettwo$ for the sum of multisets,
and $x \in \mset$ if $\mset(x) > 0$.
We write $\ms{x_1,\hdots,x_n}$ for the multiset of elements
$x_1,\hdots,x_n$, taking multiplicities into account.
If $X$ is a finite set and $f : X \to Y$ is a function,
we use the ``multiset builder'' notation
$\msb{f(x)}{x \in X}$ to denote the multiset
$\sum_{x \in X} [f(x)]$.
If $(X,>)$ is a partially ordered set,
we define a binary relation $\mgt^1$ on multisets
by declaring that
$\mset + \ms{x} \mgtone \mset + \msettwo$
if $x > y$ for every $y \in \msettwo$.
The {\em multiset order} induced by $(X,>)$
is the strict order relation on multisets
defined by declaring that
$\mset \mgt \msettwo$ if and only if $\mset \mathrel{(\mgtone)^+} \msettwo$.
We recall the following widely known theorem
by Dershowitz and Manna~\cite{dershowitz1979proving}:

\begin{theorem}
\lthm{dershowitz_manna}
If $(X,>)$ is well-founded,
then $(\Multi{X},\mgt)$ is well-founded.
\end{theorem}
As usual, $\mset \mgeq \msettwo$
stands for $(\mset = \msettwo \lor \mset \mgt \msettwo)$,
and $\mset \mleq \msettwo$ stands for $\msettwo \mgeq \mset$.
We define an operation $k \mtimes \mset$
by the recursive equations $0 \mtimes \mset \eqdef \msempty$
and $(1 + k) \mtimes \mset \eqdef \mset + k \mtimes \mset$.
The relation $\mset \mgtmap \msettwo$,
called the {\em pointwise multiset order},
is defined to hold if
$\mset$ and $\msettwo$ can be written as of the forms
$\mset = \ms{x_1,\hdots,x_n}$
and
$\msettwo = \ms{y_1,\hdots,y_n}$
in such a way that $x_i > y_i$ for all $i\in1..n$.
Observe that
if $\mset \mgtmap \msettwo$
then for all $k \in \Natz$ we have that $\mset \mgeq k \mtimes \msettwo$.
Another easy-to-check property is that
if $\mset \mgtmap \msettwo$ and $\mset$ is non-empty
then $\mset \mgt \msettwo$.

\subparagraph{A first frustrated attempt}
As mentioned in the introduction,
Turing's measure, given by
$
  \turingme{\ltm} \eqdef
    \msb{d}{\text{$\redex$ is a redex occurrence of degree $d$ in $\ltm$}}
$,
decreases when contracting the rightmost redex of
highest degree.
Our goal is to mend the $\turingmesym$-measure in such a way
that contracting {\em any} redex decreases the measure.
The difficulty is that a redex of degree $d$ may copy redexes of
a higher or equal degree $d' \geq d$.
So one can wonder: whenever a redex $\redex$ of degree $d$
makes $n$ copies of a redex $\redextwo$ of degree $d' \geq d$,
in what sense can the copies of $\redextwo$ be considered ``smaller''
than $\redextwo$?
To address this, we generalize the $\turingmesym$-measure
to a family of measures
$
\turingmetwo{D}{\ltm} \eqdef
  \msb{(d, \turingmetwo{d-1}{\ltm})}{
    \text{$\redex$ is a redex occurrence of degree $d \leq D$ in $\ltm$}
  }
$
indexed by a degree $D \in \Natz$.
Note that $\turingmetwo{0}{\ltm}$ is the empty multiset because
there are no redexes of degree $0$.

Let us try to argue that if $d \leq D$ and
$\ltm \tobetad{d} \ltmtwo$
then $\turingmetwo{D}{\ltm} \mgt \turingmetwo{D}{\ltmtwo}$.
Here $\ltm \tobetad{d} \ltmtwo$ means that $\ltm \tobeta \ltmtwo$
by contracting a redex of degree $d$.
Suppose that the contraction
of the redex $\redex : \ltm \tobetad{d} \ltmtwo$
copies a redex $\redextwo$ of degree $d'$,
where we assume that $d < d' \leq D$,
producing $n$ copies $\redextwo_1,\hdots,\redextwo_n$.
Note that the contribution of $\redextwo$ to the multiset
is $(d',\turingmetwo{d'-1}{\ltm})$,
and the contribution of each $\redextwo_i$ is $(d',\turingmetwo{d'-1}{\ltmtwo})$.
By induction on $D$, we could inductively argue that
$\turingmetwo{d'-1}{\ltm} \mgt \turingmetwo{d'-1}{\ltmtwo}$,
since $d' - 1 < d' \leq D$.
So far the property would seem to hold.

The problem with this proposal is that
a redex $\redex$ of degree $d$ may still
make copies of redexes of degree {\em exactly} $d$,
whose contribution does not necessarily decrease\footnote{For example,
in $\ltm
 = (\lam{\var^\bbtyp}{\vartwo^{\bbtyp\to\bbtyp\to\bbtyp}\,\var^\bbtyp\,\var^\bbtyp})\,
   ((\lam{\varthree^\bbtyp}{\varthree^\bbtyp})\,\varfour^\bbtyp)
 \mathrel{\tobetad{1}}
   \vartwo^{\bbtyp\to\bbtyp\to\bbtyp}
   \,((\lam{\varthree^\bbtyp}{\varthree^\bbtyp})\,\varfour^\bbtyp)
   \,((\lam{\varthree^\bbtyp}{\varthree^\bbtyp})\,\varfour^\bbtyp)
 = \ltmtwo$
the measure does not decrease,
as $\turingmetwo{1}{\ltm} = \ms{(1,[]),(1,[])} = \turingmetwo{1}{\ltmtwo}$.
}.

\subparagraph{A second frustrated attempt}
The difficulty is to deal with the situation in which
a redex $\redex$ of degree $d$ makes $n$ copies of a redex $\redextwo$
of the same degree $d$.
A key observation is that a reduction sequence $\ltm \rtobetad{d} \ltmtwo$
must be a {\em development}\footnote{Recall that a development of a set of
redexes $X$ is a reduction sequence $\ltm \tobeta^* \ltmtwo$
in which each step contracts a {\em residual} of a redex in $X$.
The residuals of a redex $\redextwo : \tm \tobeta \tmtwo$
after the contraction of a redex $\redex : \tm \tobeta \tm'$ 
are, informally speaking, the ``copies'' left of $\redextwo$ in $\tm'$.
For formal definitions see~\cite[Section~11.2]{barendregt1984}.}
of the set of redexes of degree $d$.
This is because contracting a redex of degree $d$ can only create redexes of
degree strictly less than $d$,
so any redex of degree $d$ that remains after one $\tobetad{d}$-step
must be a {\em residual} of a
preexisting redex.
This motivates our second attempt to define a measure,
consisting of two families of measures $\turingmethreea{D}{-}$ and
$\turingmethreeb{D}{-}$,
indexed by $D \in \Natz$ and defined mutually recursively:
\begin{itemize}
\item[]
  $\turingmethreea{D}{\ltm} \eqdef
    \msb{(d,\turingmethreeb{d}{\ltm})}{
      \text{$\redex$ is a $\beta$-redex occurrence of degree $d \leq D$ in $\ltm$}
    }$
\item[]
  $\turingmethreeb{D}{\ltm} \eqdef
    \msb{\turingmethreea{D-1}{\ltm'}}{
      \redseq : \ltm \rtobetad{D} \ltm'
    }$
\end{itemize}
Note that there are no redexes of degree $0$,
so $\turingmethreea{D}{\ltm}$ may not depend on $\turingmethreeb{0}{\ltm}$.
In fact, $\turingmethreeb{D}{\ltm}$ is defined only for $D \geq 1$.
The recursive definition is well-founded because
$\turingmethreea{D}{\ltm}$ may depend on
$\turingmethreeb{1}{\ltm},\hdots,\turingmethreeb{D}{\ltm}$
which in turn may only depend on $\turingmethreea{d}{\ltm'}$
for $d < D$.
The multiplicity of $\turingmethreea{D-1}{\ltm'}$
in the multiset $\turingmethreeb{D}{\ltm}$
is given by the number of reduction sequences that contract only
redexes of degree $D$, that is,
the number of different paths $\ltm \rtod{D} \ltm'$.
One important point is that, for the measure $\turingmethreeb{D}{\tm}$
to be well defined,
one needs to argue that the number of paths $\ltm \rtod{D} \ltm'$ is finite.
Since $\ltm \rtod{D} \ltm'$ is a development,
this is a consequence of the {\em finite developments} (FD)
property for orthogonal HRSs~\cite[Theorem~11.5.11]{Terese}.\footnote{Note that
FD only ensures that developments are finite. To see that
the set $\set{\redseq \ST \ltm \rtod{D} \ltm'}$ is finite, one should resort to
K\"onig's lemma, together with the fact that the STLC is finitely branching.
For a constructive proof, one can use a computable decreasing measure,
such as in de~Vrijer's proof of FD~\cite{de1985direct}.}

Let us try to argue that if $d \leq D$ and
$\ltm \tobetad{d} \ltmtwo$
then $\turingmethreea{D}{\ltm} \mgt \turingmethreea{D}{\ltmtwo}$.
On the first hand,
if a redex $\redex : \ltm \tobetad{d} \ltmtwo$
of degree $d$ copies a redex $\redextwo$ of {\em exactly} the same
degree $d$ making $n$ copies $\redextwo_1,\hdots,\redextwo_n$,
the contribution of $\redextwo$ to the multiset is $(d,\turingmethreeb{d}{\ltm})$,
whereas each $\redextwo_i$ contributes $(d,\turingmethreeb{d}{\ltmtwo})$,
and we can argue that 
$\turingmethreeb{d}{\ltm} \mgt \turingmethreeb{d}{\ltmtwo}$,
because we can injectively map each reduction sequence
$\redseq : \ltmtwo \rtobetad{d} \ltmtwo'$
to the reduction sequence
$\redex\redseq : \ltm \tobetad{d} \ltmtwo \rtobetad{d} \ltmtwo'$,
where $\redex\redseq$ denotes the composition of $\redex$ and $\redseq$.
Furthermore, there is an empty reduction sequence $\ltm \rtobetad{d} \ltm$
contributing an element $\turingmethreea{d-1}{\ltm}$ to 
$\turingmethreeb{d}{\ltm}$ but not to $\turingmethreeb{d}{\ltmtwo}$.

On the other hand,
if the contraction of a redex $\redex : \ltm \tobetad{d} \ltmtwo$
of degree $d$ copies a redex $\redextwo$ of {\em strictly} greater
degree $d' > d$ making $n$ copies $\redextwo_1,\hdots,\redextwo_n$,
the weight of $\redextwo$ is $(d',\turingmethreeb{d'}{\ltm})$
and the weight of each $\redextwo_i$ is $(d',\turingmethreeb{d'}{\ltmtwo})$,
and we would need to show that
$\turingmethreeb{d'}{\ltm} \mgt \turingmethreeb{d'}{\ltmtwo}$.
One way to do so would be
to map each reduction sequence $\redseq : \ltmtwo \rtobetad{d} \ltmtwo'$
to a reduction sequence $\redseqtwo : \ltm \rtobetad{d} \ltm'$
such that $\turingmethreea{d'-1}{\ltm'} \mgt \turingmethreea{d'-1}{\ltmtwo'}$.
However, there does not seem to be a way to rule out
the possibility that $\redseqtwo$ might erase $\redex$ and that
$\ltm' = \ltmtwo'$, which would yield
$\turingmethreea{d'-1}{\ltm'} = \turingmethreea{d'-1}{\ltmtwo'}$,
rather than a strict inequality.
The root of the problem seems again to be {\em erasure}.

\subparagraph{Definition of the $\amesym$-measure}
The $\amesym$-measure is based on the ideas described above,
but considering reduction in the $\lambdaG$-calculus rather than in
the STLC, to ensure that there is no erasure.
Informally, the $\amesym$-measure is defined by means of the two
following equations.
These equations are exactly as the ones
defining $\turingmethreea{D}{-}$ and $\turingmethreeb{D}{-}$ above,
with the only difference that they deal
with $\lambdaG$-terms and $\tog$-reduction
rather than with pure $\lambda$-terms and $\tobeta$-reduction:
\begin{itemize}
\item[]
  $\ame{D}{\tm} \eqdef
    \msb{(d,\bme{d}{\tm})}{
      \text{$\redex$ is a $\symg$-redex occurrence of degree $d \leq D$ in $\tm$}
    }$
\item[]
  $\bme{D}{\tm} \eqdef
    \msb{\ame{D-1}{\tm'}}{
      \redseq : \tm \rtod{D} \tm'
    }$
\end{itemize}
To be able to reason about these measures inductively, it will be convenient
to define an auxiliary measure $\eme{d}{\tm_0}{\tm}$
as the multiset of elements of the form $(d,\bme{d}{\tm_0})$
for each $\symg$-redex occurrence of degree {\em exactly} $d$ in $\tm$.
This auxiliary measure takes two arguments $\tm_0$ and $\tm$,
and it is defined by structural recursion on the second argument ($\tm$),
while the first argument ($\tm_0$) is used to keep track of the original
term.
Note that, with this auxiliary definition,
we can write $\ame{D}{\tm}$ as the sum
$\ame{D}{\tm} = \eme{1}{\tm}{\tm} + \hdots + \eme{D}{\tm}{\tm}$.

To define the measure formally, we start by precisely defining the codomain of the measure.

\begin{definition}[Codomain of the $\amesym$-measure]
For each $d \geq 0$, we define a set $\ameset{d}$,
and for $d \geq 1$ we define a set $\bmeset{d}$,
mutually recursively:
\begin{itemize}
\item[]
  $
  \ameset{d}
  \eqdef
  \Multi{\set{(i,b) \ST 1 \leq i \leq d,\ b \in \bmeset{i}}}
  $
  \HS
  $
  \bmeset{d}
  \eqdef
  \Multi{\ameset{d - 1}}
  $
\end{itemize}
\end{definition}

The sets $\ameset{d}$ and $\bmeset{d}$
are partially ordered by the induced multiset ordering
on their elements.
Tuples $(i,b)$ are ordered with the lexicographic order,
that is, $(i,b) > (i',b')$ if and only if
$i > i' \lor (i = i' \land b \mgt b')$.
Note that $\ameset{0} = \set{\msempty}$
and that if $d \leq d'$ then $\ameset{d} \subseteq \ameset{d'}$
and $\bmeset{d} \subseteq \bmeset{d'}$.
Moreover, $(\ameset{d},\mgt)$ and $(\bmeset{d},\mgt)$ are well-founded
partial orders by~\rthm{dershowitz_manna}.

Given typable $\lambdaG$-terms $\tm_0$, $\tm$,
and $d \in \Natz$, we define
$\eme{d}{\tm_0}{\tm} \in \ameset{d}$
and $\ame{d}{\tm} \in \ameset{d}$,
and if $d > 0$ we define $\bme{d}{\tm} \in \bmeset{d}$,
by induction on $d$ as follows.
Note that $\eme{d}{\tm_0}{\tm}$ is defined by a nested
induction on $\tm$,
and it is also defined on memories ($\eme{d}{\tm_0}{\sctx}$):

\begin{definition}[The measures $\eme{d}{-}{-}$, $\ame{d}{-}$, and $\bme{d}{-}$]
\[
  \begin{array}{r@{\,}c@{\,}l}
    \eme{d}{\tm_0}{\var}
    & \eqdef &
    \msempty
  \\
    \eme{d}{\tm_0}{\lam{\var}{\tmtwo}}
    & \eqdef &
    \eme{d}{\tm_0}{\tmtwo}
  \\
    \eme{d}{\tm_0}{\tmtwo\,\tmthree}
    & \eqdef &
    \begin{cases}
    \eme{d}{\tm_0}{\tmtwo'}
    + \eme{d}{\tm_0}{\sctx}
    + \eme{d}{\tm_0}{\tmthree}
    + \ms{(d,\bme{d}{\tm_0})}
    \\
    \hfill\text{if $\tmtwo = (\lam{\var}{\tmtwo'})\sctx$ and it is of degree $d$}
    \\
    \eme{d}{\tm_0}{\tmtwo} + \eme{d}{\tm_0}{\tmthree}
    \hfill\text{otherwise}
    \end{cases}
  \\
    \eme{d}{\tm_0}{\bin{\tmtwo}{\tmthree}}
    & \eqdef &
    \eme{d}{\tm_0}{\tmtwo} + \eme{d}{\tm_0}{\tmthree}
  \\
  \\
    \eme{d}{\tm_0}{\ctxhole}
    & \eqdef &
    \msempty
  \\
    \eme{d}{\tm_0}{\bin{\sctx}{\tm}}
    & \eqdef &
    \eme{d}{\tm_0}{\sctx} + \eme{d}{\tm_0}{\tm}
  \end{array}
\]
\[
  \begin{array}{rcll}
    \ame{d}{\tm}
    & \eqdef &
    \sum_{i=1}^{d} \eme{i}{\tm}{\tm}
  \\
    \bme{d}{\tm}
    & \eqdef &
    \msb{
      \ame{d-1}{\tm'}
    }{\redseq : \tm \rtod{d} \tm'}
  \end{array}
\]
\end{definition}
\medskip
Moreover, the {\bf $\amesym$-measure} itself is defined for $\lambda$-terms
as follows:

\begin{definition}
\ldef{amesym_measure}
If $\ltm$ is a typable $\lambda$-term,
$\amefull{\ltm} \eqdef \ame{D}{\ltm}$
where $D := \maxdeg{\ltm}$.
\end{definition}
When we write $\ame{D}{\ltm}$, we implicitly regard
$\ltm$ as a $\lambdaG$-term without any memorized terms.

From a higher-level perspective,
the $\eme{d}{\tm_0}{\tm}$ measure defined above
is the multiset of pairs of the form $(d,\bme{d}{\tm_0})$
for each redex of degree $d$ in $\tm$.
Similarly, $\ame{D}{\tm}$
is the multiset of pairs of the form
$(d,\bme{d}{\tm})$
for each redex of degree $d \leq D$ in $\tm$.
In particular, $\eme{0}{\tm_0}{\tm}$ and $\ame{0}{\tm}$
are empty multisets, because there are no redexes of degree $0$.
Two easy remarks are that
$D \leq D'$ implies $\ame{D}{\tm} \mleq \ame{D'}{\tm}$,
and that
$\eme{d}{\tm_0}{\tm\sctx} = \eme{d}{\tm_0}{\tm} + \eme{d}{\tm_0}{\sctx}$.

\begin{remark}
As mentioned in the preceding discussion, one important point is that
for $\bme{d}{-}$ to be well-defined we need to argue that the set
$\set{\redseq \ST \exists \tm'.\,\, \redseq : \tm \rtod{d} \tm'}$ is finite.
This is a consequence of~\rcoro{termination_of_tod}.
\end{remark}

\begin{example}
Let
$\Delta := \lam{\var^{\bbtyp\to\bbtyp}}{
             \var^{\bbtyp\to\bbtyp}(\var^{\bbtyp\to\bbtyp}\varthree^\bbtyp)}$
and
$W := \lam{\vartwo^\bbtyp}{\varfour^\bbtyp}$
and consider the diagram:
\[
  \xymatrix@R=-.2cm@C=.6cm{
  &
  &
    \tm_2 = \varfour\garb{W\varthree}\garb{W}
    \ar@/^.2cm/^-{1}[rd]
  \\
    \tm_0 = \Delta\,W
    \ar^-{2}[r]
    &
    \tm_1 = (W(W\varthree))\garb{W}
    \ar@/^.2cm/^-{1}[ur]
    \ar@/_.2cm/^-{1}[dr]
    &
    &
    \tm_4 = \varfour\garb{w\garb{\varthree}}\garb{W}
  \\
    &&
    \tm_3 = (W(\varfour\garb{\varthree}))\garb{W}
    \ar@/_.2cm/^-{1}[ur]
  }
\]
Then
$\ame{0}{\tm_1} = \ame{0}{\tm_2} = \ame{0}{\tm_3} = \ame{0}{\tm_4}
 = \ame{1}{\tm_4} = \ame{2}{\tm_4} = \msempty$, and:
\[
\!\!\!\!
\begin{array}{ll}
  \ame{2}{\tm_0} = \ms{(2,\bme{2}{\tm_0})}
  &
  \bme{2}{\tm_0} = \ms{\ame{1}{\tm_0}, \ame{1}{\tm_1}}
\\
  \ame{2}{\tm_1} = \ame{1}{\tm_1} = \ms{(1,\bme{1}{\tm_1}),(1,\bme{1}{\tm_1})}
  &
  \bme{1}{\tm_1} = \ms{\ame{0}{\tm_1},\ame{0}{\tm_2},\ame{0}{\tm_3},\ame{0}{\tm_4}}
\\
  \ame{2}{\tm_2} = \ame{1}{\tm_2} = \ms{(1,\bme{1}{\tm_2})}
  &
  \bme{1}{\tm_2} = \ms{\ame{0}{\tm_2},\ame{0}{\tm_4}}
\\
  \ame{2}{\tm_3} = \ame{1}{\tm_3} = \ms{(1,\bme{1}{\tm_3})}
  &
  \bme{1}{\tm_3} = \ms{\ame{0}{\tm_3},\ame{0}{\tm_4}}
\end{array}
\]
In particular,
$\ame{2}{\tm_0}
 \mgt \ame{2}{\tm_1}
 \mgt \ame{2}{\tm_2}
 \mgt \ame{2}{\tm_4}$
and
$\ame{2}{\tm_1} \mgt \ame{2}{\tm_3} \mgt \ame{2}{\tm_4}$.
\end{example}

\subparagraph{The $\amesym$-measure is decreasing}
Lastly, we show the main theorem of this section, stating that
if $\ltm \tobeta \ltmtwo$ then $\amefull{\ltm} \mgt \amefull{\ltmtwo}$.
This theorem is based on three technical results, that we call
{\em high/increase}, {\em low/decrease}, and {\em forget/decrease}:
\begin{enumerate}
\item
  {\bf High/increase}~(\rprop{upper_reduction}) establishes
  ---perhaps confusingly---
  that $\ame{d}{-}$ (non-strictly) {\bf increases} if one contracts a redex of
  higher degree $D > d$. More precisely, if $0 \leq d < D$ and $\tm \tod{D} \tm'$
  then $\ame{d}{\tm} \mleq \ame{d}{\tm'}$.
  Note that $\ame{d}{\tm}$ only looks at redexes of degree $i \leq d$,
  and contracting a redex of degree $D > d$ {\em cannot erase} a redex of
  any degree $i \leq d$, because the $\lambdaG$-calculus is non-erasing.
  Contracting a redex of degree $D$ can, at most, replicate redexes of degree $i$.
  This property is needed for a technical reason to prove the low/decrease
  property, and it relies crucially on the commutation result of the previous
  section~(\rprop{commutation_by_degrees}).
\item
  {\bf Low/decrease}~(\rprop{lower_reduction}) establishes
  that $\ame{D}{-}$ {\bf strictly} decreases if one contracts a redex of
  lower degree $d < D$. More precisely, if $1 \leq d \leq D$ and $\tm \tod{d} \tm'$
  then $\ame{D}{\tm} \mgt \ame{D}{\tm'}$.
  This is the core of the argument, and the most technically difficult part
  to prove. It relies crucially on the lifting property of the previous
  section~(\rprop{retraction_of_higher_degree_steps}).
\item
  {\bf Forget/decrease}~(\rprop{shrinking_ame}) establishes
  that forgetful reduction (non-strictly) decreases the measure.
  More precisely, if $\tm \shone \tm'$ then $\ame{d}{\tm} \mgeq \ame{d}{\tm'}$.
  This property is used as a final step in the main theorem,
  and it relies crucially on postponement of forgetful reduction,
  as studied in the previous section~(\rprop{retraction_before_shrinking}).
\end{enumerate}

Below we sketch the proofs of these three properties.
\SeeAppendix{See \rprop{appendix:upper_reduction}, \rprop{appendix:lower_reduction},
and \rprop{appendix:shrinking_ame} in the appendix for detailed proofs.}
Let us first mention a straightforward lemma.
\begin{lemma}[Measure of a substitution]
\llem{upper_substitution_lemma}
\llem{lower_substitution_lemma}
\Item{1.}
  $\eme{d}{\tm_0}{\tm} \mleq \eme{d}{\tm_0}{\tm\sub{\var}{\tmtwo}}$.
\Item{2.}
  If $\tmtwo$ is not a $\symg$-abstraction of degree $d$,
  then
  $\eme{d}{\tm_0}{\tm\sub{\var}{\tmtwo}}
   = \eme{d}{\tm_0}{\tm} + k \mtimes \eme{d}{\tm_0}{\tmtwo}$
  for some $k \in \Natz$.
\end{lemma}
\begin{proof}
By induction on $\tm$.
\SeeAppendix{See \rlem{appendix:upper_substitution_lemma}
and \rlem{appendix:lower_substitution_lemma} in the appendix for
details.}
\end{proof}

\begin{proposition}[High/increase]
\lprop{upper_reduction}
Let $D \in \Natz$. Then the following hold:
\begin{enumerate}
\item
  If $1 \leq d < D$
  and $\tm \tod{D} \tm'$
  then $\bme{d}{\tm} \mleq \bme{d}{\tm'}$.
\item
  If $0 \leq d < D$
  and $\tm_0 \tod{D} \tm'_0$
  then $\eme{d}{\tm_0}{\tm} \mleq \eme{d}{\tm'_0}{\tm}$.
\item
  If $0 \leq d < D$
  and $\tm_0 \tod{D} \tm'_0$
  and $\tm \tod{D} \tm'$
  then $\eme{d}{\tm_0}{\tm} \mleq \eme{d}{\tm'_0}{\tm'}$.
\item
  If $0 \leq d < D$
  and $\tm \tod{D} \tm'$
  then $\ame{d}{\tm} \mleq \ame{d}{\tm'}$.
\end{enumerate}
\end{proposition}
\begin{proof}
The four items are proved simultaneously by induction on $d$,
where item~\Item{1} resorts to the \ih,
and the following items may resort to the previous items
without decreasing~$d$. Items~\Item{2} and~\Item{3}
proceed by a nested induction on $\tm$. Most cases are straightforward.

One interesting situation occurs in item~\Item{3} when
$\tm = \app{(\lam{\var}{\tmtwo})\sctx}{\tmthree}$
is the redex of degree $D$ contracted by the step $\tm \tod{D} \tm'$.
Then we resort to the first part of~\rlem{upper_substitution_lemma}.

Another interesting part of the proof is item~\Item{1}.
Let $1 \leq d < D$ and $\tm \tod{D} \tm'$
and let us show that $\bme{d}{\tm} \mleq \bme{d}{\tm'}$.
Indeed, let
$X := \set{\redseq \ST (\exists \tmtwo)\ \redseq : \tm \rtod{d} \tmtwo}$
and
$Y := \set{\redseqtwo \ST (\exists \tmtwo')\ \redseqtwo : \tm' \rtod{d} \tmtwo'}$,
%Note that, by definition,
%$\bme{d}{\tm} = \msb{\ame{d-1}{\tgt{\redseq}}}{\redseq \in X}$
%and
%$\bme{d}{\tm'} = \msb{\ame{d-1}{\tgt{\redseqtwo}}}{\redseqtwo \in Y}$.
and let $\redex : \tm \tod{D} \tm'$.
Using~\rprop{commutation_by_degrees},
we can define an injective function $\varphi : X \to Y$
by $\varphi(\redseq) := \redseq/\redex$.
Note that $\ame{d-1}{\tgt{\redseq}} \mleq \ame{d-1}{\tgt{\varphi(\redseq)}}$
holds for every $\redseq \in X$ using item~\Item{4} of the \ih
(noting that $1 \leq d - 1 < D$ holds because $1 \leq d < D$),
resorting to the \ih as many times as the length of the reduction
$\tmtwo \rtod{D} \tmtwo'_\redseq$.
To conclude the proof, let $Z = Y \setminus \varphi(X)$. Then:
\begin{itemize}
\item[]
  $
    \bme{d}{\tm}
  =
    \msb{\ame{d-1}{\tgt{\redseq}}}{\redseq \in X}
  \mleq^{\text{($\star$)}}
    \msb{\ame{d-1}{\tgt{\varphi(\redseq)}}}{\redseq \in X}
  =^{\text{($\star\star$)}}
    \msb{\ame{d-1}{\tgt{\redseqtwo}}}{\redseqtwo \in \varphi(X)}
  $
\item[]
  $
  \hphantom{\bme{d}{\tm}}
  \mleq
    \msb{\ame{d-1}{\tgt{\redseqtwo}}}{\redseqtwo \in \varphi(X)}
    +
    \msb{\ame{d-1}{\tgt{\redseqtwo}}}{\redseqtwo \in Z}
  =
    \msb{\ame{d-1}{\tgt{\redseqtwo}}}{\redseqtwo \in Y}
  =
    \bme{d}{\tm'}
  $
\end{itemize}
To justify the step marked with ($\star$),
note that
$\msb{\ame{d-1}{\tgt{\redseq}}}{\redseq \in X}
 = \sum_{\redseq \in X} \ms{\ame{d-1}{\tgt{\redseq}}}
 \mleq \sum_{\redseq \in X} \ms{\ame{d-1}{\tgt{\varphi(\redseq)}}}
 = \msb{\ame{d-1}{\tgt{\varphi(\redseq)}}}{\redseq \in X}$
because $\ame{d-1}{\tgt{\redseq}} \mleq \ame{d-1}{\tgt{\varphi(\redseq)}}$,
as we have already claimed.
To justify the step marked with ($\star\star$),
note that $\varphi$ is injective.
\end{proof}

\begin{proposition}[Low/decrease]
\lprop{lower_reduction}
Let $D \in \Natz$. Then the following hold:
\begin{enumerate}
\item
  \label{lower_reduction:bme_decrease}
  If $1 \leq d \leq j \leq D$
  and $\tm \tod{d} \tm'$
  then $\bme{j}{\tm} \mgt \bme{j}{\tm'}$.
\item
  \label{lower_reduction:eme_left_decrease}
  If $1 \leq d \leq j \leq D$
  and $\tm_0 \tod{d} \tm'_0$
  then $\eme{j}{\tm_0}{\tm} \mgtmap \eme{j}{\tm'_0}{\tm}$.
\item
  \label{lower_reduction:eme_right_equal_decrease}
  If $1 \leq d \leq D$
  and $\tm_0 \tod{d} \tm'_0$
  and $\tm \tod{d} \tm'$,
  then for all $\mset \in \ameset{d-1}$
  we have
  $\eme{d}{\tm_0}{\tm} \mgt \eme{d}{\tm'_0}{\tm'} + \mset$.
\item
  \label{lower_reduction:eme_right_nonequal_decrease}
  If $1 \leq d < j \leq D$
  and $\tm_0 \tod{d} \tm'_0$
  and $\tm \tod{d} \tm'$
  then $\eme{j}{\tm_0}{\tm} \mgeq \eme{j}{\tm'_0}{\tm'}$.
\item
  \label{lower_reduction:ame_decrease}
  If $1 \leq d \leq D$
  and $\tm \tod{d} \tm'$
  then $\ame{D}{\tm} \mgt \ame{D}{\tm'}$.
\end{enumerate}
\end{proposition}
\begin{proof}
The five items are proved simultaneously by induction on $D$,
where item~\Item{1} resorts to the \ih,
and the following items may resort to the previous items
without decreasing~$d$. Items~\Item{2}--\Item{4}
proceed by a nested induction on $\tm$.
We mention some of the interesting parts of the proof.

For item~\Item{1},
let $1 \leq d \leq j \leq D$
and $\tm \tod{d} \tm'$
and let us show that $\bme{j}{\tm} \mgt \bme{j}{\tm'}$.
Let $X := \set{\redseq \ST (\exists{\tmtwo})\ \redseq : \tm \rtod{j} \tmtwo}$
and
$Y := \set{\redseqtwo \ST (\exists{\tmtwo'})\ \redseqtwo : \tm' \rtod{j} \tmtwo'}$,
and consider two subcases:
\begin{itemize}
\item
  If $d = j$,
  let $\redex : \tm \tod{d} \tm'$,
  define an injective function $\varphi : Y \to X$
  by $\varphi(\redseqtwo) = \redex\,\redseqtwo$,
  let $Z = X \setminus \varphi(Y)$, and note that:
  \begin{itemize}
  \item[]
    $\bme{j}{\tm}
     =
        \msb{\ame{j-1}{\tgt{\redseq}}}{\redseq \in \varphi(Y)}
      + \msb{\ame{j-1}{\tgt{\redseq}}}{\redseq \in Z}
     $
  \item[]
    $
     \hphantom{\bme{j}{\tm}}
     =
        \msb{\ame{j-1}{\tgt{\redex\redseqtwo}}}{\redseqtwo \in Y}
      + \msb{\ame{j-1}{\tgt{\redseq}}}{\redseq \in Z}
     $
     \text{ since $\varphi$ is injective}
   \item[]
     $
     \hphantom{\bme{j}{\tm}}
     =
        \msb{\ame{j-1}{\tgt{\redseqtwo}}}{\redseqtwo \in Y}
      + \msb{\ame{j-1}{\tgt{\redseq}}}{\redseq \in Z}
     =
        \bme{j}{\tm'}
      + \msb{\ame{j-1}{\tgt{\redseq}}}{\redseq \in Z}
     $
  \end{itemize}
  To conclude that $\bme{j}{\tm} \mgt \bme{j}{\tm'}$,
  note that $Z$ is non-empty because it contains
  the empty reduction $\emptyseq : \tm \rtod{d} \tm$.
\item
  If $d < j$,
  we construct a function $\varphi : Y \to X$ as follows.
  By~\rprop{retraction_of_higher_degree_steps},
  for each reduction $\redseqtwo : \tm' \rtod{j} \tmtwo'$
  there exist $\tmtwo_\redseqtwo, \tmthree_\redseqtwo$,
  and reductions $\varphi(\redseqtwo) : \tm \rtod{j} \tmtwo_\redseqtwo$
  and $\tmtwo' \rtod{j} \tmthree_\redseqtwo$
  and $\tmtwo_\redseqtwo \todplus{d} \tmthree_\redseqtwo$.
  Note that for every $\redseqtwo \in Y$
  we have
  $\ame{j-1}{\tgt{\varphi(\redseqtwo)}}
   =
     \ame{j-1}{\tmtwo_\redseqtwo}
   \mgt^{\dagger}
     \ame{j-1}{\tmthree_\redseqtwo}
   \mgeq^{\ddagger}
     \ame{j-1}{\tmtwo'}
   =
     \ame{j-1}{\tgt{\redseqtwo}}$
  where $\dagger$ holds by item \ref{lower_reduction:ame_decrease} of the \ih
  observing that $1 \leq d \leq j - 1 < D$ because $d < j \leq D$,
  and $\ddagger$ holds by high/increase~(\rprop{upper_reduction})
  observing that $0 \leq j - 1 < j$.
  To conclude the proof,
  let $Z = X \setminus \varphi(Y)$, and note that:
  \begin{itemize}
  \item[]
    $
    \bme{j}{\tm}
    =
      \msb{\ame{j-1}{\tgt{\redseq}}}{\redseq \in \varphi(Y)}
    + \msb{\ame{j-1}{\tgt{\redseq}}}{\redseq \in Z}
    $
  \item[]
    $
    \hphantom{\bme{j}{\tm}}
    =
      \msb{\ame{j-1}{\tgt{\varphi(\redseqtwo)}}}{\redseqtwo \in Y}
    + \msb{\ame{j-1}{\tgt{\redseq}}}{\redseq \in Z}
    $
  \item[]
    $
    \hphantom{\bme{j}{\tm}}
      \mgeq
      \msb{\ame{j-1}{\tgt{\varphi(\redseqtwo)}}}{\redseqtwo \in Y}
      \mgt^{(\star)}
      \msb{\ame{j-1}{\tgt{\redseqtwo}}}{\redseqtwo \in Y}
      =
      \bme{j}{\tm'}
    $
  \end{itemize}
  For the step marked with ($\star$),
  note that
  $\msb{\ame{j-1}{\tgt{\varphi(\redseqtwo)}}}{\redseqtwo \in Y}
   \mgtmap
   \msb{\ame{j-1}{\tgt{\redseqtwo}}}{\redseqtwo \in Y}$
  because
  $\ame{j-1}{\tgt{\varphi(\redseqtwo)}} \mgt \ame{j-1}{\tgt{\redseqtwo}}$
  holds by the claim above
  where, moreover, $Y$ is non-empty because it contains
  the empty reduction $\emptyseq : \tm' \rtod{j} \tm'$.
\end{itemize}

Another interesting situation occurs in item~\Item{3},
when $\tm = \app{(\lam{\var}{\tmtwo})\sctx}{\tmthree}$
is the redex of degree $d$ contracted by the step $\tm \tod{d} \tm'$.
The step is of the form
$\tm
 = (\lam{\var}{\tmtwo})\sctx\,\tmthree
 \tod{d} \tmtwo\sub{\var}{\tmthree}\garb{\tmthree}\sctx
 = \tm'$.
Note that $\tmthree$ is not an abstraction of degree $d$,
because it is the argument of an abstraction of degree $d$.
So by~\rlem{lower_substitution_lemma} there exists $k \in \Natz$
such that
$\eme{d}{\tm'_0}{\tmtwo\sub{\var}{\tmthree}}
 = \eme{d}{\tm'_0}{\tmtwo} + k \mtimes \eme{d}{\tm'_0}{\tmthree}$.
The crucial observation is that
$\eme{d}{\tm_0}{\tmthree}
 \mgeq
 (1 + k) \mtimes \eme{d}{\tm'_0}{\tmthree}$,
which is because
by item~\ref{lower_reduction:eme_left_decrease}
we have that $\eme{d}{\tm_0}{\tmthree} \mgtmap \eme{d}{\tm'_0}{\tmthree}$.

Finally, for item~\Item{5},
let $1 \leq d \leq D$ and $\tm \tod{d} \tm'$
and let us show that $\ame{D}{\tm} \mgt \ame{D}{\tm'}$.
Indeed:
\begin{itemize}
\item[]
  $
    \ame{D}{\tm}
    =
      \sum_{i=1}^{D} \eme{i}{\tm}{\tm}
    \mgeq
        \eme{d}{\tm}{\tm}
      + \sum_{j=d+1}^{D} \eme{j}{\tm}{\tm}
  $
\item[]
  $
    \hphantom{\ame{D}{\tm}}
    \mgt
      \ame{d-1}{\tm'}
    + \eme{d}{\tm'}{\tm'}
    + \sum_{j=d+1}^{D} \eme{j}{\tm}{\tm}
  $
  by item~\Item{\ref{lower_reduction:eme_right_equal_decrease}},
  taking $\mset := \ame{d-1}{\tm'}$
\item[]
  $
    \hphantom{\ame{D}{\tm}}
    \mgeq
      \ame{d-1}{\tm'}
    + \eme{d}{\tm'}{\tm'}
    + \sum_{j=d+1}^{D} \eme{j}{\tm'}{\tm'}
    = \ame{D}{\tm'}$
  by item~\Item{\ref{lower_reduction:eme_right_nonequal_decrease}}.
  \qedhere
\end{itemize}
\end{proof}

\begin{proposition}[Forget/decrease]
\lprop{shrinking_ame}
Let $d \in \Natz$. Then the following hold:
\begin{enumerate}
\item
  \label{shrinking_ame:bme_decrease}
  If $\tm \shone \tm'$
  then $\bme{d}{\tm} \mgeq \bme{d}{\tm'}$.
\item
  \label{shrinking_ame:eme_left_decrease}
  If $\tm_0 \shone \tm'_0$
  then $\eme{d}{\tm_0}{\tm} \mgeq \eme{d}{\tm'_0}{\tm}$.
\item
  \label{shrinking_ame:eme_right_decrease}
  If $\tm_0 \shone \tm'_0$
  and $\tm \shone \tm'$
  then $\eme{d}{\tm_0}{\tm} \mgeq \eme{d}{\tm'_0}{\tm'}$.
\item
  \label{shrinking_ame:ame_decrease}
  If $\tm \shone \tm'$
  then $\ame{d}{\tm} \mgeq \ame{d}{\tm'}$.
\end{enumerate}
\end{proposition}
\begin{proof}
The four items are proved simultaneously by induction on $D$,
where item~\Item{1} resorts to the \ih,
and the following items may resort to the previous items
without decreasing~$d$. Items~\Item{2} and \Item{3}
proceed by a nested induction on $\tm$.

The interesting part is item~\Item{1},
so let $\tm \shone \tm'$ and let us show
that $\bme{d}{\tm} \mgeq \bme{d}{\tm'}$.
Let $X := \set{\redseq \ST (\exists{\tmtwo})\ \redseq : \tm \rtod{d} \tmtwo}$
and $Y := \set{\redseqtwo \ST (\exists{\tmtwo'})\ \redseqtwo : \tm' \rtod{d} \tmtwo'}$.
Define an injective function $\varphi : Y \to X$
by $\varphi(\redseqtwo) := \retract{\redseqtwo}{\redex}$,
resorting to~\rprop{retraction_before_shrinking},
where $\retract{\redseqtwo}{\redex} : \tm \rtod{d} \tmtwo_\redseqtwo$.
and $\tmtwo_\redseqtwo \shone^* \tmtwo'$.
Note that for every $\redseqtwo \in Y$ we have
$\ame{d-1}{\tgt{\varphi(\redseqtwo)}}
 = \ame{d-1}{\tmtwo_\redseqtwo}
 \mgeq^\dagger \ame{d-1}{\tmtwo'}
 = \ame{d-1}{\tgt{\redseqtwo}}$,
where $\dagger$ holds by item~\Item{\ref{shrinking_ame:ame_decrease}}
of the \ih, observing that $d - 1 < d$.
To conclude the proof, let $Z = X \setminus \varphi(Y)$, and note that:
\begin{itemize}
\item[]
  $\bme{d}{\tm}
   =
   \msb{\ame{d-1}{\tgt{\redseq}}}{\redseq \in \varphi(Y)}
   + \msb{\ame{d-1}{\tgt{\redseq}}}{\redseq \in Z}
   \mgeq
   \msb{\ame{d-1}{\tgt{\redseq}}}{\redseq \in \varphi(Y)}
  $
\item[]
  $
   \hphantom{\bme{d}{\tm}}
    =^{(\star)}
    \msb{\ame{d-1}{\tgt{\varphi(\redseqtwo)}}}{\redseqtwo \in Y}
    \mgeq^{(\star\star)}
      \msb{\ame{d-1}{\tgt{\redseqtwo}}}{\redseqtwo \in Y}
    = \bme{d}{\tm'}
  $
\end{itemize}
For the step marked with ($\star$), note that $\varphi$ is injective.
For the step marked with ($\star\star$),
note that
$\msb{\ame{d-1}{\tgt{\varphi(\redseqtwo)}}}{\redseqtwo \in Y}
 = \sum_{\redseqtwo \in Y} \ms{\ame{d-1}{\tgt{\varphi(\redseqtwo)}}}
 \mgeq \sum_{\redseqtwo \in Y} \ms{\ame{d-1}{\tgt{\redseqtwo}}}
 = \msb{\ame{d-1}{\tgt{\redseqtwo}}}{\redseqtwo \in Y}$
because 
$\ame{d-1}{\tgt{\varphi(\redseqtwo)}} \mgeq \ame{d-1}{\tgt{\redseqtwo}}$,
as we have already justified.
\end{proof}

Finally, we prove the main theorem in this section:
\begin{theorem}
\lthm{ame_decreasing}
Let $\ltm,\ltmtwo$ be typable $\lambda$-terms
such that $\ltm \tobeta \ltmtwo$.
Then $\amefull{\ltm} > \amefull{\ltmtwo}$.
\end{theorem}
\begin{proof}
Let $D = \maxdeg{\ltm}$ and $D' = \maxdeg{\ltmtwo}$.
Let $\ltm \tog \tmtwo$ be the step corresponding to $\ltm \tobeta \ltmtwo$.
By \rlem{reduce_shrink_lemma} note that $\tmtwo \shone \ltmtwo$.
Then:
\[
  \amefull{\ltm}
  = \ame{D}{\ltm}
  \mgt^{\text{\rprop{lower_reduction}}}
    \ame{D}{\tmtwo}
  \mgeq^{\text{\rprop{shrinking_ame}}}
    \ame{D}{\ltmtwo}
  \mgeq
    \ame{D'}{\ltmtwo}
  =
    \amefull{\ltmtwo}
\]
The last inequality holds because $D \geq D'$
since, as is well-known, contraction of a $\beta$-redex in the
simply typed $\lambda$-calculus
cannot create a redex of higher degree.
\end{proof}

\section{Conclusion}
\lsec{conclusion}

We have defined two decreasing measures for the STLC,
the $\meassym$-measure~(\rdef{meassym_measure})
and the $\amesym$-measure~(\rdef{amesym_measure}).
These measures are
decreasing (\rthm{z_measure_decreases} and \rthm{ame_decreasing} respectively)
and, to the best of our knowledge, they provide two new proofs of strong
normalization for the STLC.
Both measures are defined constructively and by purely syntactic methods,
using the $\lambdaG$-calculus as an auxiliary tool.

The problem of finding a ``straightforward'' decreasing measure for $\beta$-reduction
in the simply typed $\lambda$-calculus is posed as Problem~\#26
in the TLCA list of open problems~\cite{tlcaopen},
and as Problem~\#19
in the RTA list of open problems~\cite{rtaopen}.

One strength of the $\meassym$-measure is that its codomain is simple:
each term is mapped to a natural number.
One weakness is that the definition of the $\meassym$-measure relies on
reduction in the $\lambdaG$-calculus,
and computing the $\meassym$-measure
is at least as costly as evaluating the $\lambda$-term itself.
Measures based on Gandy's~\cite{gandy80sn,dv1987exactly} have similar
characteristics.
One question is whether the values of the $\meassym$-measure and
measures based on Gandy's
can be related. It is not immediate to establish a precise
correspondence.

On the other hand, one strength of the $\amesym$-measure
is that it shows how to extend
Turing's measure $\turingme{-}$ so that it decreases when contracting
{\em any} redex.
The proof is based on a delicate analysis of how contracting
a redex of degree $d$ may create and copy redexes of degree $d'$,
depending on whether $d < d'$, or $d = d'$, or $d > d'$.
We hope that this may provide novel insights on why the STLC is SN.
The codomain of the $\amesym$-measure is not so simple,
as the $\amesym$-measure maps each term to a structure of nested multisets.
Yet, it is ``reasonably simple'':
the fact that the partial orders $\ameset{d}$ and $\bmeset{d}$ 
are well-founded only relies on the ordinary multiset and lexicographic orderings.
The $\amesym$-measure is costly to compute; in particular
$\bme{d}{\tm}$ is defined as a sum over all reductions
$\redseq : \tm \rtod{d} \tm'$, which may produce a combinatorial explosion.
Another weakness is that our proofs make use of relatively heavy rewriting
machinery, as we have to keep explicit track of witnesses
(\eg in~\rsec{reduction_by_degrees}).

Besides the techniques mentioned in the introduction,
other proofs of SN of the STLC can be found in the literature.
For example, David~\cite{David01} gives a purely syntactic proof of
SN relying on the standardization theorem;
Loader~\cite{Loader98}, as well as Joachimski and Matthes~\cite{JoachimskiM03},
give combinatorial proofs of SN based on inductive predicates
characterizing strongly normalizing terms.
As far as we know, the only proofs that explicitly construct
decreasing measures are those based on Gandy's.

The idea of keeping ``leftover garbage''
can be traced back to at least
the works of Nederpelt~\cite{NederpeltPhdThesis} and
Klop~\cite{Tesis:Klop},
who studied non-erasing variants of (possibly) erasing rewriting systems,
in order to relate weak and strong normalization.
Many variations of these ideas have been explored in the past,
such as in de Groote's notion of $\beta_S$ reduction~\cite{Groote93}
or Neergaard and Sørensen calculus with memory~\cite{NeergaardS02}.
Instead of using the $\lambdaG$-calculus, it is possible that other
non-erasing systems may be used.
For instance, Gandy~\cite{gandy80sn} translates $\lambda$-terms to the
terms of $\lambda{I}$-calculus to avoid erasing arguments.

The definition of reduction in the $\lambdaG$-calculus, which allows
arbitrary memory in between the abstraction and the application,
is inspired by Accattoli and Kesner's work on calculi with
explicit substitutions ``at a distance''~\cite{AccattoliK10}.
This mechanism can be traced back, again, to at least the work
of Nederpelt~\cite{NederpeltPhdThesis}.

The definition of the $\lambdaG$-calculus as a means to obtain an
{\em increasing} measure was inspired by the fact that,
in explicit substitution calculi without erasure,
labeled reduction (in the sense of Lévy labels~\cite{Tesis:Levy:1978})
increases the sum of the sizes of all the labels in the term~\cite{BarenbaumB17}.

\bibliography{biblio}

\newpage
\appendix

\section{Technical appendix}

\subsection{Proofs of \rsec{lambdaG} --- The $\lambdaG$-calculus}
\lsec{appendix:lambdaG}

In this section we give detailed proofs of the results
about the $\lambdaG$-calculus stated in \rsec{lambdaG}.

\begin{remark}
\lremark{bin_is_symg_abstraction}
$\bin{\tm}{\tmtwo}$ is a $\symg$-abstraction
if and only if
$\tm$ is a $\symg$-abstraction.
\end{remark}

\begin{lemma}[Substitution lemma]
\llem{typing_substitution_lemma}
Let $\judg{\tctx,\var:\typ}{\tm}{\typtwo}$
and $\judg{\tctx}{\tmtwo}{\typ}$.
Then $\judg{\tctx}{\tm\sub{\var}{\tmtwo}}{\typtwo}$.
\end{lemma}
\begin{proof}
Straightforward by induction on $\tm$.
\end{proof}

\begin{proposition}[Subject reduction]
\lprop{appendix:subject_reduction}
Let $\judg{\tctx}{\tm}{\typ}$ and $\tm \tog \tmtwo$.
Then $\judg{\tctx}{\tmtwo}{\typ}$.
\end{proposition}
\begin{proof}
Straightforward by induction on the derivation
of the step $\tm \tog \tmtwo$, resorting to \rlem{typing_substitution_lemma}
for the base case, when there is a $\symg$-reduction step at
the root.
\end{proof}

\subsubsection{Confluence of the $\lambdaG$-calculus}

\begin{proposition}[Confluence]
\lprop{appendix:confluence}
The $\lambdaG$-calculus is confluent.
That is,
if $\tm_1 \tog^* \tm_2$ and $\tm_1 \tog^* \tm_3$,
there exists a term $\tm_4$ such that
$\tm_2 \tog^* \tm_4$ and $\tm_3 \tog^* \tm_4$.
\end{proposition}
\begin{proof}
The proof can be done following standard techniques.
For example, following Tait and Martin-L\"of's technique,
we may define a notion of simultaneous reduction $\ptog$
that allows to contract many redexes simultaneously,
\ie allowing the complete development of any set of redexes on
the starting term.
Then it suffices to show that 
$\tog \mathrel{\subseteq} \ptog \mathrel{\subseteq} \tog^*$
and that $\ptog$ enjoys the diamond property,
\ie that if $\tm_1 \ptog \tm_2$ and $\tm_1 \ptog \tm_3$
there exists a term $\tm_4$
such that $\tm_2 \ptog \tm_4$ and $\tm_3 \ptog \tm_4$.
The key lemma is:
\[
  \tm \ptog \tm'
  \text{ and }
  \tmtwo \ptog \tmtwo'
  \text{ implies }
  \tm\sub{\var}{\tmtwo} \ptog \tm'\sub{\var}{\tmtwo'}
\]
The key diagrams in the proof that $\ptog$ enjoys the diamond
property are:
\[
  \xymatrix{
    \app{(\lam{\var}{\tm_1})\sctx_1\,\tmtwo_2}
    \ar@{=>}[r]
    \ar@{=>}[d]
  &
    \app{(\lam{\var}{\tm_2})\sctx_2\,\tmtwo_2}
    \ar@{=>}[d]
  \\
    \tm_3\sub{\var}{\tmtwo_3}\garb{\tmtwo_3}\sctx_3
    \ar@{=>}[r]
  &
    \tm_4\sub{\var}{\tmtwo_4}\garb{\tmtwo_4}\sctx_4
  }
  \HS\HS
  \xymatrix{
    \app{(\lam{\var}{\tm_1})\sctx_1\,\tmtwo_2}
    \ar@{=>}[r]
    \ar@{=>}[d]
  &
    \tm_2\sub{\var}{\tmtwo_2}\garb{\tmtwo_2}\sctx_2
    \ar@{=>}[d]
  \\
    \tm_3\sub{\var}{\tmtwo_3}\garb{\tmtwo_3}\sctx_3
    \ar@{=>}[r]
  &
    \tm_4\sub{\var}{\tmtwo_4}\garb{\tmtwo_4}\sctx_4
  }
\]
\end{proof}

\subsubsection{Simplification of a $\lambdaG$-term}

\begin{definition}[Generalization of notions to memories]
We generalize some of the notions to memories as follows:
\begin{enumerate}
\item
  The reduction relation $\tog$ is extended to operate on memories
  with the two following inductively defined rules:
  \begin{enumerate}
  \item If $\tm \tog \tm'$ then $\sctx\garb{\tm} \tog \sctx\garb{\tm'}$.
  \item If $\sctx \tog \sctx'$ then $\sctx\garb{\tm} \tog \sctx'\garb{\tm}$.
  \end{enumerate}
\item
  The max-degree is extended to memories as follows:
  $\maxdeg{\ctxhole} = 0$
  and
  $\maxdeg{\sctx\garb{\tm}} = \max(\maxdeg{\sctx},\maxdeg{\tm})$.
%%%% Ya está en la def. original
%\item
%  The simplification of degree $k$ is extended to memories as follows:
%  $\simpk{\ctxhole} = \ctxhole$
%  and
%  $\simpk{\sctx\garb{\tm}} = \simpk{\sctx}\garb{\simpk{\tm}}$.
\end{enumerate}
\end{definition}

\begin{lemma}[Terms reduce to its simplification]
\llem{tm_reduces_to_simpk}
For every term $\tm$ and for all $k \geq 1$
we have that $\tm \tog^* \simpk{\tm}$.
\end{lemma}
\begin{proof}
To prove it by induction,
we generalize the statement to memories,
\ie $\sctx \tog^* \simpk{\sctx}$.
We proceed by simultaneous induction on $\tm$ and $\sctx$:
\begin{enumerate}
\item
  $\tm = \var$:
  Immediate,
  as
  $\var
   \tog^* \var
   = \simpk{\var}$ in zero steps.
\item
  $\tm = \lam{\var}{\tmtwo}$:
  Then
  $\lam{\tmtwo}
   \tog^* \lam{\var}{\simpk{\tmtwo}}
   = \simpk{\lam{\var}{\tmtwo}}$
  by \ih.
\item
  $\tm = (\lam{\var}{\tmtwo})\sctx\,\tmthree$
    where $(\lam{\var}{\tmtwo})\sctx$
    is a $\symg$-abstraction of degree $k$:
  By \ih
  $(\lam{\var}{\tmtwo})\sctx\,\tmthree
  \tog^*
      (\lam{\var}{\simpk{\tmtwo}})\simpk{\sctx}\,\simpk{\tmthree}
  \tog
    \simpk{\tmtwo}
    \sub{\var}{\simpk{\tmthree}}
    \garb{\simpk{\tmthree}}
    \simpk{\sctx}
  = \simpk{(\lam{\var}{\tmtwo})\sctx\,\tmthree}
  = \simpk{\tm}$.
\item
  $\tm = \tmtwo\,\tmthree$
    where $\tmtwo$ is not a $\symg$-abstraction of degree $k$:
  By \ih
  $\app{\tmtwo}{\tmthree}
   \tog^* \app{\simpk{\tmtwo}}{\simpk{\tmthree}}
   = \simpk{\app{\tmtwo}{\tmthree}}
   = \simpk{\tm}$.
\item
  $\tm = \bin{\tmtwo}{\tmthree}$:
  By \ih
  $\bin{\tmtwo}{\tmthree}
   \tog^* \bin{\simpk{\tmtwo}}{\simpk{\tmthree}}
   = \simpk{\bin{\tmtwo}{\tmthree}}
   = \simpk{\tm}$.
\item
  $\sctx = \ctxhole$:
  Immediate, as
  $\ctxhole
   \tog^* \ctxhole
   = \simpk{\ctxhole}$ in zero steps.
\item
  $\sctx = \sctx'\garb{\tm}$:
  By \ih
  $\sctx'\garb{\tm}
   \tog^* \simpk{\sctx'}\garb{\simpk{\tm}}
   = \simpk{\sctx'\garb{\tm}}$.
\end{enumerate}
\end{proof}

\begin{lemma}[Substitution of terms of lower type does not create abstractions]
\llem{substitution_does_not_create_abstractions}
If $\height{\typeof{\tm}} > \height{\typeof{\tmtwo}}$
and $\tm$ not a $\symg$-abstraction,
then $\tm\sub{\var}{\tmtwo}$ is not a $\symg$-abstraction.
\end{lemma}
\begin{proof}
By induction on $\tm$:
\begin{enumerate}
\item
  $\tm = \vartwo$:
  We claim that $\vartwo \neq \var$.
  Indeed, note that the type of $\vartwo$ is $\typeof{\tm}$ 
  but the type of $\var$ is $\typeof{\tmtwo}$.
  By contradiction, suppose that $\var = \vartwo$.
  Then $\typeof{\tm} = \typeof{\tmtwo}$
  and in particular
  $\height{\typeof{\tm}}
   > \height{\typeof{\tmtwo}}
   = \height{\typeof{\tm}}$,
  which is impossible.
  Then we have that $\vartwo \neq \var$,
  so $\tm\sub{\var}{\tmtwo} = \vartwo\sub{\var}{\tmtwo} = \vartwo$,
  which is not a $\symg$-abstraction.
\item
  $\tm = \lam{\var}{\tm'}$:
  Impossible, since $\tm$ is not a $\symg$-abstraction by hypothesis.
\item
  $\tm = \app{\tm_1}{\tm_2}$:
  Then
  $\tm\sub{\var}{\tmtwo}
   = \app{\tm_1\sub{\var}{\tmtwo}}{\tm_2\sub{\var}{\tmtwo}}$
  is trivially not an abstraction.
\item
  $\tm = \bin{\tm_1}{\tm_2}$:
  Note that $\tm_1$ is not a $\symg$-abstraction
  by \rremark{bin_is_symg_abstraction}.
  Furthermore, note that $\typeof{\tm_1} = \typeof{\tm}$
  so in particular
  $\height{\typeof{\tm_1}}
   = \height{\typeof{\tm}}
   > \height{\typeof{\tmtwo}}$.
  We are under the conditions to apply the \ih on $\tm_1$,
  hence $\tm_1\sub{\var}{\tmtwo}$ is not a $\symg$-abstraction.
  To conclude, note that
  $\tm\sub{\var}{\tmtwo}
   = \bin{\tm_1\sub{\var}{\tmtwo}}{\tm_2\sub{\var}{\tmtwo}}$
  cannot be a $\symg$-abstraction by \rremark{bin_is_symg_abstraction}.
\end{enumerate}
\end{proof}

\begin{lemma}[Simplification does not create abstractions]
\llem{simplification_does_not_create_abstractions}
If $\height{\typeof{\tm}} \geq k$
and $\maxdeg{\tm} \leq k$
and $\tm$ is not a $\symg$-abstraction,
then $\simpk{\tm}$ is not a $\symg$-abstraction.
\end{lemma}
\begin{proof}
By induction on $\tm$:
\begin{enumerate}
\item
  $\tm = \var$:
  Then $\simpk{\var} = \var$ is not a $\symg$-abstraction.
\item
  $\tm = \lam{\var}{\tmtwo}$:
  Impossible, since $\tm$ is not a $\symg$-abstraction by hypothesis.
\item
  $\tm = \app{(\lam{\var}{\tmtwo})\sctx}{\tmthree}$
    where $(\lam{\var}{\tmtwo})\sctx$
    is a $\symg$-abstraction of degree $k$:
  We claim that this case is impossible.
  Writing the types explicitly, we have that
  $\typeof{\tm} = \typtwo$ with
  $\height{\typtwo} = \height{\typeof{\tm}} \geq k$ by hypothesis.
  Then the type of the function must be of the form
  $\typeof{\lam{\var}{\tmtwo}} = \typ\to\typtwo$.
  But note that
  $\height{\typeof{\lam{\var}{\tmtwo}}}
  = \height{\typ\to\typtwo}
  > \height{\typtwo}
  \geq k$.
  This means that $(\lam{\var}{\tmtwo})\sctx$ cannot be of degree $k$,
  contradicting the hypothesis of this case.
\item
  $\tm = \app{\tmtwo}{\tmthree}$
    where $\tmtwo$ is not a $\symg$-abstraction of degree $k$:
  Then
  $\simpk{\tm}
   = \app{\tmtwo}{\tmthree}
   = \app{\simpk{\tmtwo}}{\simpk{\tmthree}}$
  is not a $\symg$-abstraction.
\item
  $\tm = \bin{\tmtwo}{\tmthree}$:
  Note that $\typeof{\tmtwo} = \typeof{\tm}$,
  so in particular $\height{\typeof{\tmtwo}} = \height{\typeof{\tm}} \geq k$.
  Moreover, $\maxdeg{\tmtwo} \leq \maxdeg{\tm} \leq k$
  and $\tmtwo$ is not a $\symg$-abstraction
  by \rremark{bin_is_symg_abstraction}.
  We are under the conditions to apply the \ih on $\tmtwo$,
  hence $\simpk{\tmtwo}$ is not a $\symg$-abstraction.
  To conclude, note that
  $\simpk{\tm} = \bin{\simpk{\tmtwo}}{\simpk{\tmthree}}$
  cannot be a $\symg$-abstraction by \rremark{bin_is_symg_abstraction}.
\end{enumerate}
\end{proof}

\begin{lemma}[Properties of the max-degree]
\llem{maxdeg_properties}
\quad
\begin{enumerate}
\item
  \label{maxdeg_properties:sctx}
  $\maxdeg{\tm\sctx} = \max(\maxdeg{\tm},\maxdeg{\sctx})$
\item
  \label{maxdeg_properties:substitution}
  If $\maxdeg{\tm} < k$
  and $\maxdeg{\tmtwo} < k$
  and $\height{\typeof{\tmtwo}} < k$
  then $\maxdeg{\tm\sub{\var}{\tmtwo}} < k$.
\end{enumerate}
\end{lemma}
\begin{proof}
Item \ref{maxdeg_properties:sctx}
is straightforward by induction on $\sctx$,
since $\maxdeg{\bin{\tm}{\tmtwo}} = \max(\maxdeg{\tm},\maxdeg{\tmtwo})$.
For item \ref{maxdeg_properties:substitution},
we proceed by induction on $\tm$:
\begin{enumerate}
\item
  $\tm = \vartwo$:
  We consider two subcases, depending on whether $\vartwo = \var$ or not.
  If $\vartwo = \var$,
  then
  $\maxdeg{\var\sub{\var}{\tmtwo}}
   = \maxdeg{\tmtwo} < k$.
  If $\vartwo \neq \var$,
  then
  $\maxdeg{\vartwo\sub{\var}{\tmtwo}}
   = \maxdeg{\vartwo}
   = \maxdeg{\tm} < k$.
\item
  $\tm = \lam{\vartwo}{\tm'}$:
  By $\alpha$-conversion we assume that $\var \neq \vartwo$.
  Note that
  $\maxdeg{\tm'}
   = \maxdeg{\lam{\vartwo}{\tm'}}
   = \maxdeg{\tm} < k$ by hypothesis.
  Then
  $\maxdeg{(\lam{\vartwo}{\tm'})\sub{\var}{\tmtwo}}
   = \maxdeg{\lam{\vartwo}{\tm'\sub{\var}{\tmtwo}}}
   = \maxdeg{\tm'\sub{\var}{\tmtwo}}
   < k$ by \ih.
\item
  $\tm = \app{\tm_1}{\tm_2}$:
  Note that
  $\maxdeg{\tm_1}
   \leq \maxdeg{\app{\tm_1}{\tm_2}}
   = \maxdeg{\tm} < k$
  and, similarly, $\maxdeg{\tm_2} < k$.
  This means that we can apply the \ih
  to obtain that
  $\maxdeg{\tm_1\sub{\var}{\tmtwo}} < k$
  and $\maxdeg{\tm_2\sub{\var}{\tmtwo}} < k$.
  We proceed by case analysis, depending on whether
  $\height{\typeof{\tm_1}} < k$ or $\height{\typeof{\tm_1}} \geq k$:
  \begin{enumerate}
  \item
    If $\height{\typeof{\tm_1}} < k$
    then, by the substitution lemma~(\rlem{typing_substitution_lemma}),
    the terms $\tm_1\sub{\var}{\tmtwo}$ and $\tm_1$ have the same type.
    In particular,
    $\app{\tm_1\sub{\var}{\tmtwo}}{\tm_2\sub{\var}{\tmtwo}}$
    cannot be a redex of degree $k$ or greater,
    since $\height{\typeof{\tm_1\sub{\var}{\tmtwo}}} < k$.
    As a consequence,
    if $\app{\tm_1\sub{\var}{\tmtwo}}{\tm_2\sub{\var}{\tmtwo}}$
    is a redex, its degree is at most $k - 1$. 
    Hence
    $\maxdeg{(\app{\tm_1}{\tm_2})\sub{\var}{\tmtwo}}
     = \maxdeg{\app{\tm_1\sub{\var}{\tmtwo}}{\tm_2\sub{\var}{\tmtwo}}}
     \leq \max(k - 1,
               \maxdeg{\tm_1\sub{\var}{\tmtwo}},
               \maxdeg{\tm_2\sub{\var}{\tmtwo}})
     < k$.
  \item
    If $\height{\typeof{\tm_1}} \geq k$,
    note that $\tm_1$ cannot be a $\symg$-abstraction,
    because then $\tm = \app{\tm_1}{\tm_2}$ would be a redex of
    degree $k$ or greater,
    but by hypothesis we know that $\maxdeg{\tm} < k$.
    Note that we are under the conditions of
    \rlem{substitution_does_not_create_abstractions},
    so we know that $\tm_1\sub{\var}{\tmtwo}$ is not a $\symg$-abstraction.
    In particular,
    $\app{\tm_1\sub{\var}{\tmtwo}}{\tm_2\sub{\var}{\tmtwo}}$
    cannot be a redex.
    Hence
    $\maxdeg{(\app{\tm_1}{\tm_2})\sub{\var}{\tmtwo}}
     = \maxdeg{\app{\tm_1\sub{\var}{\tmtwo}}{\tm_2\sub{\var}{\tmtwo}}}
     = \max(\maxdeg{\tm_1\sub{\var}{\tmtwo}},
            \maxdeg{\tm_2\sub{\var}{\tmtwo}})
     < k$.
  \end{enumerate}
\item
  $\tm = \bin{\tm_1}{\tm_2}$:
  Note that
  $\maxdeg{\tm_1} \leq \maxdeg{\bin{\tm_1}{\tm_2}} = \maxdeg{\tm} < k$
  and, similarly,
  $\maxdeg{\tm_2} < k$.
  This means that we can apply the \ih to obtain that
  $\maxdeg{\tm_1\sub{\var}{\tmtwo}} < k$
  and
  $\maxdeg{\tm_2\sub{\var}{\tmtwo}} < k$.
  Hence
  $\maxdeg{(\bin{\tm_1}{\tm_2})\sub{\var}{\tmtwo}}
   = \maxdeg{\bin{\tm_1\sub{\var}{\tmtwo}}{\tm_2\sub{\var}{\tmtwo}}}
   = \max(
       \maxdeg{\tm_1\sub{\var}{\tmtwo}},
       \maxdeg{\tm_2\sub{\var}{\tmtwo}}
     )
   < k$.
\end{enumerate}
\end{proof}

\begin{lemma}[Simplification decreases the max-degree]
\llem{simpk_maxdeg_decrease}
Suppose that $k \geq 1$.
If $\maxdeg{\tm} \leq k$
then $\maxdeg{\simpk{\tm}} < k$.
\end{lemma}
\begin{proof}
Let $k \geq 1$ be such that $\maxdeg{\tm} \leq k$.
We argue that $\maxdeg{\simpk{\tm}} < k$, that is,
all the redexes in $\tm$ have degree less than $k$.
To prove it by induction,
we generalize the statement to memories,
proving also that $\maxdeg{\simpk{\sctx}} < k$.
We prove the statement simultaneously
by induction on $\tm$ and $\sctx$:
\begin{enumerate}
\item
  $\tm = \var$:
  Then
  $\simpk{\var}
   = \var$
  has no redexes, so $\maxdeg{\var} = 0 < 1 \leq k$.
\item
  $\tm = \lam{\var}{\tmtwo}$:
  Note that $\maxdeg{\tmtwo} = \maxdeg{\lam{\var}{\tmtwo}} \leq k$ 
  so by \ih 
  $\maxdeg{\simpk{\tmtwo}} < k$.
  Moreover,
  $\maxdeg{\simpk{\lam{\var}{\tmtwo}}}
   = \maxdeg{\lam{\var}{\simpk{\tmtwo}}}
   = \maxdeg{\simpk{\tmtwo}} < k$.
\item
  $\tm = \app{(\lam{\var}{\tmtwo})\sctx}{\tmthree}$
    where $(\lam{\var}{\tmtwo})\sctx$
    is a $\symg$-abstraction of degree $k$:
  Note that
  $\maxdeg{\tmtwo} \leq \maxdeg{\tm}$
  because any redex in the subterm $\tmtwo$
  is also a redex in the whole term $\tm$,
  so in particular $\maxdeg{\tmtwo} \leq k$
  and we may apply the \ih on $\tmtwo$
  to conclude that $\maxdeg{\simpk{\tmtwo}} < k$.
  Similarly, by \ih,
  we have that $\maxdeg{\simpk{\sctx}} < k$
  and $\maxdeg{\simpk{\tmthree}} < k$.

  Since $\tm$ is typable,
  its type is of the form $\typeof{\tm} = \typtwo$
  with
  $\typeof{(\lam{\var}{\tmtwo})\sctx} = \typ\to\typtwo$
  and $\typeof{\tmthree} = \typ$.
  Note that
  $\height{\typeof{\tmthree}} = \height{\typ} < \height{\typ\to\typtwo} = k$
  since $(\lam{\var}{\tmtwo})\sctx$ is of degree $k$ by hypothesis
  of this case.

  To conclude this case, note that:
  \[
    \begin{array}{rcll}
    &&
      \maxdeg{\simpk{\tm}}
    \\
    & = &
      \maxdeg{
        \simpk{\tmtwo}
        \sub{\var}{\simpk{\tmthree}}
        \garb{\simpk{\tmthree}}
        \simpk{\sctx}
      }
      \\&&\HS\text{by definition}
    \\
    & \leq &
      \max(
        \maxdeg{
          \simpk{\tmtwo}
          \sub{\var}{\simpk{\tmthree}}
        }
      , \maxdeg{\simpk{\tmthree}}
      , \maxdeg{\simpk{\sctx}}
      )
    \\&&\HS\text{by \rlem{maxdeg_properties}~(\ref{maxdeg_properties:sctx})}
    \\
    & < &
      \max(k,k,k)
    \\&&\HS\text{by \rlem{maxdeg_properties}~(\ref{maxdeg_properties:substitution}) and the \ih}
    \\
    & = &
      k
    \end{array}
  \]
  For the last inequality,
  we use the fact that $\height{\typeof{\tmthree}} < k$.
\item
  $\tm = \app{\tmtwo}{\tmthree}$
    where $\tmtwo$ is not a $\symg$-abstraction of degree $k$:
  Note that $\maxdeg{\tmtwo} \leq \maxdeg{\tm}$
  because any redex in the subterm $\tmtwo$
  is also a redex in the whole term $\tm$,
  so in particular $\maxdeg{\tmtwo} \leq k$
  and we may apply the \ih on $\tmtwo$
  to conclude that $\maxdeg{\simpk{\tmtwo}} < k$.
  Similarly, by \ih, we have that $\maxdeg{\simpk{\tmthree}} < k$.

  We proceed by case analysis, depending on whether
  $\height{\typeof{\tmtwo}} < k$ or $\height{\typeof{\tmtwo}} \geq k$:
  \begin{enumerate}
  \item
    If $\height{\typeof{\tmtwo}} < k$,
    then by \rlem{tm_reduces_to_simpk}
    we know that $\tmtwo \tog^* \simpk{\tmtwo}$
    and by subject reduction~(\rprop{subject_reduction})
    we have that $\typeof{\tmtwo} = \typeof{\simpk{\tmtwo}}$.
    In particular, $\app{\simpk{\tmtwo}}{\simpk{\tmthree}}$
    cannot be a redex of degree $k$ or greater,
    because $\height{\typeof{\simpk{\tmtwo}}} = \height{\typeof{\tmtwo}} < k$.
    That is, if $\app{\simpk{\tmtwo}}{\simpk{\tmthree}}$ is a redex,
    its degree is at most $k - 1$.
    Hence we have that
    $\maxdeg{\simpk{\tm}}
     = \maxdeg{\app{\simpk{\tmtwo}}{\simpk{\tmthree}}}
     \leq \max(k - 1, \maxdeg{\simpk{\tmtwo}}, \maxdeg{\simpk{\tmthree}})
     < k$.
  \item
    If $\height{\typeof{\tmtwo}} \geq k$,
    note that $\tmtwo$ cannot be a $\symg$-abstraction.
    Indeed, we know by hypothesis of this case
    that $\tmtwo$ is not an abstraction of degree $k$.
    Furthermore, $\tmtwo$ cannot be an abstraction of degree $k' > k$,
    because then $\tm = \app{\tmtwo}{\tmthree}$
    would be a redex of degree $k' > k$,
    but then we would have that
    $k < k' \leq \maxdeg{\tm} \leq k$,
    which is a contradiction.
    Since $\tmtwo$ is not a $\symg$-abstraction,
    $\maxdeg{\tmtwo} \leq k$, and $\height{\typeof{\tmtwo}} \geq k$,
    we are under the conditions to apply
    \rlem{simplification_does_not_create_abstractions}
    to conclude that $\simpk{\tmtwo}$ is not a $\symg$-abstraction.
    This means that
    $\app{\simpk{\tmtwo}}{\simpk{\tmthree}}$
    cannot be a redex.
    Hence we have that
    $\maxdeg{\simpk{\tm}}
     = \maxdeg{\app{\simpk{\tmtwo}}{\simpk{\tmthree}}}
     = \max(\maxdeg{\simpk{\tmtwo}},\maxdeg{\simpk{\tmthree}})
     < k$.
  \end{enumerate}
\item
  \label{simpk_maxdeg_decrease:bin}
  $\tm = \bin{\tmtwo}{\tmthree}$:
  Note that $\maxdeg{\tmtwo} \leq \maxdeg{\tm}$,
  so in particular $\maxdeg{\tmtwo} \leq k$
  and we may apply the \ih on $\tmtwo$
  to conclude that $\maxdeg{\simpk{\tmtwo}} < k$.
  Similarly, by \ih, we have that $\maxdeg{\simpk{\tmthree}} < k$.
  Hence we have that
  $\maxdeg{\simpk{\tm}}
   = \maxdeg{\bin{\simpk{\tmtwo}}{\simpk{\tmthree}}}
   = \max(\maxdeg{\simpk{\tmtwo}},\maxdeg{\simpk{\tmthree}})
   < k$.
\item
  $\sctx = \ctxhole$:
  Immediate, as $\maxdeg{\ctxhole} = 0 < 1 \leq k$.
\item
  $\sctx = \sctx'\garb{\tm}$:
  Similar to case \ref{simpk_maxdeg_decrease:bin} of this lemma.
\end{enumerate}
\end{proof}

\begin{proposition}[Simplification is normalization]
\lprop{appendix:simplification_is_normalization}
$\tm \tog^* \simpfull{\tm}$
and $\simpfull{\tm}$ is a $\tog$-normal form.
\end{proposition}
\begin{proof}
Let $k$ be the max-degree of $\tm$.
For each $0 \leq i \leq k$ we define $\simp{> i}{\tm}$
as follows, by induction on $k - i$:
\[
  \begin{array}{rcll}
    \simp{> k}{\tm} & \eqdef & \tm \\
    \simp{> i}{\tm} & \eqdef & \simp{i+1}{\simp{>i+1}{\tm}}
                             & \text{for each $0 \leq i < k$}\\
  \end{array}
\]
That is,
$\simp{> i}{\tm} \eqdef
 \simp{i+1}{\hdots\simp{k-1}{\simp{k}{\tm}}}$.
Note that
$\simp{> k}{\tm} = \tm$
and
$\simp{> 0}{\tm} = \simpfull{\tm}$.
Let us prove each of the two parts of the statement:
\begin{enumerate}
\item
  To show that $\tm \tog^* \simpfull{\tm}$,
  note that
  for each $1 \leq i \leq k$
  we have that
  $\simp{> i}{\tm} \tog^* \simp{i}{\simp{> i}{\tm}} = \simp{> i - 1}{\tm}$
  by \rlem{tm_reduces_to_simpk}.
  Hence:
  \[
    \tm = \simp{> k}{\tm}
    \tog^* \simp{> k-1}{\tm}
    \hdots
    \tog^* \simp{> i}{\tm}
    \tog^* \simp{> i-1}{\tm}
    \hdots
    \tog^* \simp{> 0}{\tm}
    = \simpfull{\tm}
  \]
\item
  To show that $\simpfull{\tm}$ is a $\tog$-normal form,
  we claim that for each $0 \leq i \leq k$
  we have that $\maxdeg{\simp{> i}{\tm}} \leq i$.
  We proceed by induction on $k - i$.
  In the base case, we have that $i = k$,
  so $\maxdeg{\simp{> k}{\tm}} = \maxdeg{\tm} = k$
  since $k$ is the max-degree of $\tm$.
  For the induction step,
  let $k - i > 0$, so $0 \leq i < k$.
  By \ih we have that $\maxdeg{\simp{> i + 1}{\tm}} \leq i + 1$.
  Then
  $\maxdeg{\simp{> i}{\tm}}
   = \maxdeg{\simp{i + 1}{\simp{> i + 1}{\tm}}}
   < i + 1$
  by \rlem{simpk_maxdeg_decrease}.
  This means that $\maxdeg{\simp{> i}{\tm}} \leq i$, as required.
\end{enumerate}
\end{proof}

\subsubsection{Forgetful reduction}

The forgetful reduction relation is generalized to operate on substitution contexts
so that, for example,
$(\ctxhole\garb{\var}\garb{\vartwo})
 \shone
 (\ctxhole\garb{\vartwo})$.

\begin{lemma}[Properties of forgetful reduction]
\llem{properties_of_shrinking}
\quad
\begin{enumerate}
\item
  \label{properties_of_shrinking:sctx_left}
  If $\tm \shone \tm'$ then $\tm\sctx \shone \tm'\sctx$.
\item
  \label{properties_of_shrinking:sctx_right}
  If $\sctx \shone \sctx'$ then $\tm\sctx \shone \tm\sctx'$.
\item
  \label{properties_of_shrinking:sub_left}
  If $\tm \shone \tm'$
  then $\tm\sub{\var}{\tmtwo} \shone \tm'\sub{\var}{\tmtwo}$.
\item
  \label{properties_of_shrinking:sub_right}
  If $\tmtwo \shone \tmtwo'$
  then $\tm\sub{\var}{\tmtwo} \shone^* \tm\sub{\var}{\tmtwo'}$
  (in zero or more steps).
\end{enumerate}
\end{lemma}
\begin{proof}
\quad
\begin{itemize}
\item
  Items \ref{properties_of_shrinking:sctx_left}
  and \ref{properties_of_shrinking:sctx_right}
  are straightforward by induction on $\sctx$.
\item
  Items \ref{properties_of_shrinking:sub_left}
  and \ref{properties_of_shrinking:sub_right}
  are straightforward by induction on $\tm$.
  For item~\ref{properties_of_shrinking:sub_right},
  note that when $\tm = \vartwo$ with $\vartwo \neq \var$,
  we have that
  $\vartwo\sub{\var}{\tmtwo} = \vartwo
   \shone^* \vartwo = \vartwo\sub{\var}{\tmtwo'}$
  in exactly zero steps.
  Note also that more that one step of $\shone$
  may be required when $\tm$ is an application or a wrapper.
\end{itemize}
\end{proof}

\begin{lemma}[Local commutation of reduction and forgetful reduction]
\llem{tog_shone_commutation}
If $\tm \shone \tmtwo$ and $\tm \tog \tm'$,
there exists a term $\tmtwo'$
such that $\tm' \shone^+ \tmtwo'$ and $\tmtwo \tog^= \tmtwo'$,
where $\shone^+$ is the transitive closure of $\shone$,
and $\tog^=$ is the reflexive closure of $\tog$.
Graphically:
\[
  \xymatrix@C=.1cm@R=.5cm{
    \tm \ar[d] & \shone & \tmtwo \ar^{=}@{.>}[d] \\
    \tm' & \shone^+ & \tmtwo' \\
  }
\]
\end{lemma}
\begin{proof}
By induction on $\tm$:
\begin{enumerate}
\item
  $\tm = \var$:
  Note that this case is impossible,
  since there are no steps $\var \tog \tm'$.
\item
  $\tm = \lam{\var}{\tm_1}$:
  Since $\lam{\var}{\tm_1} \tog \tm'$,
  we know that $\tm'$ must be of the form $\tm' = \lam{\var}{\tm'_1}$
  with $\tm_1 \tog \tm'_1$.
  Note that the $\shone$ step is internal,
  that is,
  $\lam{\var}{\tm_1} \shone \lam{\var}{\tmtwo_1} = \tmtwo$
  with $\tm_1 \shone \tmtwo_1$.
  By \ih there exists $\tmtwo'_1$
  such that $\tm'_1 \shone^+ \tmtwo'_1$
  and $\tmtwo_1 \tog^= \tmtwo'_1$.
  Taking $\tmtwo' := \lam{\var}{\tmtwo'_1}$ we have:
  \[
    \xymatrix@C=.25cm@R=.5cm{
      \lam{\var}{\tm_1}
      \ar[d]
      & \shone &
      \lam{\var}{\tmtwo_1}
      \ar^{=}[d]
    \\
      \lam{\var}{\tm'_1}
      & \shone^+ &
      \lam{\var}{\tmtwo'_1}
    }
  \]
\item
  $\tm = \app{\tm_1}{\tm_2}$:
  We consider three subcases, depending on whether the step
  $\app{\tm_1}{\tm_2} \tog \tm'$ is a $\tog$ step at the root,
  internal to $\tm_1$, or internal to $\tm_2$:
  \begin{enumerate}
  \item
    If the $\tog$ step is at the root,
    then $\tm_1$ is a $\symg$-abstraction
    of the form $\tm_1 = (\lam{\var}{\tm_{11}})\sctx$
    and the step is of the form
    $\tm = \app{(\lam{\var}{\tm_{11}})\sctx}{\tm_2}
      \tog \tm_{11}\sub{\var}{\tm_2}\garb{\tm_2}\sctx = \tm'$.
    Moreover, since
    $\tm = \app{(\lam{\var}{\tm_{11}})\sctx}{\tm_2} \shone \tm'$,
    we consider three further subcases, depending on whether the
    step $\tm \shone \tmtwo$
    is internal to $\tm_{11}$, internal to $\sctx$,
    or internal to $\tm_{2}$:
    \begin{enumerate}
    \item
      If the $\shone$ step is internal to $\tm_{11}$,
      then $\tmtwo = \app{(\lam{\var}{\tmtwo_{11}})\sctx}{\tm_2}$
      with $\tm_{11} \shone \tmtwo_{11}$.
      Taking
      $\tmtwo' := \tmtwo_{11}\sub{\var}{\tm_2}\garb{\tm_2}\sctx$
      we have:
      \[
        \xymatrix@C=.25cm@R=.5cm{
          \app{(\lam{\var}{\tm_{11}})\sctx}{\tm_2}
          \ar[d]
          & \shone &
          \app{(\lam{\var}{\tmtwo_{11}})\sctx}{\tm_2}
          \ar[d]
        \\
          \tm_{11}\sub{\var}{\tm_2}\garb{\tm_2}\sctx
          & \shone &
          \tmtwo_{11}\sub{\var}{\tm_2}\garb{\tm_2}\sctx
        }
      \]
      For the $\shone$ step at the bottom,
      by
      \rlem{properties_of_shrinking}~(\ref{properties_of_shrinking:sctx_left})
      it suffices to show that
      $\tm_{11}\sub{\var}{\tm_2}
       \shone
       \tmtwo_{11}\sub{\var}{\tm_2}$.
      This is a consequence of
      \rlem{properties_of_shrinking}~(\ref{properties_of_shrinking:sub_left}).
    \item
      If the $\shone$ step is internal to $\sctx$,
      then $\tmtwo = \app{(\lam{\var}{\tm_{11}})\sctx'}{\tm_2}$
      with $\sctx \shone \sctx'$.
      Taking
      $\tmtwo' := \tm_{11}\sub{\var}{\tm_2}\garb{\tm_2}\sctx'$
      we have:
      \[
        \xymatrix@C=.25cm@R=.5cm{
          \app{(\lam{\var}{\tm_{11}})\sctx}{\tm_2}
          \ar[d]
          & \shone &
          \app{(\lam{\var}{\tm_{11}})\sctx'}{\tm_2}
          \ar[d]
        \\
          \tm_{11}\sub{\var}{\tm_2}\garb{\tm_2}\sctx
          & \shone &
          \tm_{11}\sub{\var}{\tm_2}\garb{\tm_2}\sctx'
        }
      \]
      The $\shone$ step at the bottom holds
      by
      \rlem{properties_of_shrinking}~(\ref{properties_of_shrinking:sctx_right}).
    \item
      %%%% This is the most interesting case.
      If the $\shone$ step is internal to $\tm_2$,
      then $\tmtwo = \app{(\lam{\var}{\tm_{11}})\sctx}{\tmtwo'_2}$
      with $\tm_2 \shone \tmtwo_2$.
      Taking
      $\tmtwo' := \tm_{11}\sub{\var}{\tm_2}\garb{\tmtwo'_2}\sctx$
      we have:
      \[
        \xymatrix@C=.25cm@R=.5cm{
          \app{(\lam{\var}{\tm_{11}})\sctx}{\tm_2}
          \ar[d]
          & \shone &
          \app{(\lam{\var}{\tm_{11}})\sctx}{\tmtwo_2}
          \ar[d]
        \\
          \tm_{11}\sub{\var}{\tm_2}\garb{\tm_2}\sctx
          & \shone^+ &
          \tm_{11}\sub{\var}{\tmtwo_2}\garb{\tmtwo_2}\sctx
        }
      \]
      For the bottom of the diagram, note that:
      $\tm_{11}\sub{\var}{\tm_2} \shone^* \tm_{11}\sub{\var}{\tmtwo_2}$
      by \rlem{properties_of_shrinking}~(\ref{properties_of_shrinking:sub_right}).
      Hence
      $\tm_{11}\sub{\var}{\tm_2}\garb{\tm_2}
       \shone^* \tm_{11}\sub{\var}{\tmtwo_2}\garb{\tm_2}
       \shone \tm_{11}\sub{\var}{\tmtwo_2}\garb{\tmtwo_2}$.
      Resorting to
      \rlem{properties_of_shrinking}~(\ref{properties_of_shrinking:sctx_left})
      we conclude.
    \end{enumerate}
  \item
    If the $\tog$ step is internal to $\tm_1$,
    the step is of the form
    $\app{\tm_1}{\tm_2} \tog \app{\tm'_1}{\tm_2}$
    with $\tm_1 \tog \tm'_1$.
    We consider two further subcases, depending on whether
    the $\shone$ step is
    internal to $\tm_1$ or internal to $\tm_2$:
    \begin{enumerate}
    \item
      \label{tog_shone_commutation:app_left_left}
      If the $\shone$ step is internal to $\tm_1$,
      then $\tmtwo = \app{\tmtwo_1}{\tm_2}$
      with $\tm_1 \shone \tmtwo_1$.
      By \ih there exists $\tmtwo'_1$
      such that $\tm'_1 \shone^+ \tmtwo'_1$ and $\tmtwo_1 \tog^= \tmtwo'_1$.
      Taking $\tmtwo' := \app{\tmtwo'_1}{\tm_2}$ we have:
      \[
        \xymatrix@C=.25cm@R=.5cm{
          \app{\tm_1}{\tm_2}
          \ar[d]
          & \shone &
          \ar^{=}[d]
          \app{\tmtwo_1}{\tm_2}
        \\
          \app{\tm'_1}{\tm_2}
          & \shone^+ &
          \app{\tmtwo'_1}{\tm_2}
        }
      \]
    \item
      \label{tog_shone_commutation:app_left_right}
      If the $\shone$ step is internal to $\tm_2$,
      then $\tmtwo = \app{\tm_1}{\tmtwo_2}$
      with $\tm_2 \shone \tmtwo_2$.
      Taking $\tmtwo' := \app{\tm'_1}{\tmtwo_2}$ we have:
      \[
        \xymatrix@C=.25cm@R=.5cm{
          \app{\tm_1}{\tm_2}
          \ar[d]
          & \shone &
          \ar[d]
          \app{\tm_1}{\tmtwo_2}
        \\
          \app{\tm'_1}{\tm_2}
          & \shone &
          \app{\tm'_1}{\tmtwo_2}
        }
      \]
    \end{enumerate}
  \item
    If the $\tog$ step is internal to $\tm_2$,
    the proof is similar to the previous case.
  \end{enumerate}
\item
  $\tm = \bin{\tm_1}{\tm_2}$:
  We consider two subcases, depending on whether the step
  $\bin{\tm_1}{\tm_2} \tog \tm'$
  is internal to $\tm_1$ or internal to $\tm_2$:
  \begin{enumerate}
  \item
    If the $\tog$ step is internal to $\tm_1$,
    then $\bin{\tm_1}{\tm_2} \tog \bin{\tm'_1}{\tm_2} = \tm'$
    with $\tm_1 \tog \tm'_1$.
    We consider three further subcases, depending on whether the
    step $\bin{\tm_1}{\tm_2} \shone \tmtwo$ is at the root of the wrapper,
    internal to $\tm_1$, or internal to $\tm_2$:
    \begin{enumerate}
    \item
      If the $\shone$ step is at the root of the wrapper,
      then $\bin{\tm_1}{\tm_2} \shone \tm_1 = \tmtwo$.
      Taking $\tmtwo' := \tm'_1$ we have:
      \[
        \xymatrix@C=.25cm@R=.5cm{
          \bin{\tm_1}{\tm_2}
          \ar[d]
          & \shone &
          \ar[d]
          \tm_1
        \\
          \bin{\tm'_1}{\tm_2}
          & \shone &
          \tm'_1
        }
      \]
    \item
      If the $\shone$ step is internal to $\tm_1$,
      then $\tmtwo = \bin{\tmtwo_1}{\tm_2}$ with $\tm_1 \shone \tmtwo_1$,
      and we conclude by \ih
      similarly as for case \ref{tog_shone_commutation:app_left_left}.
    \item
      If the $\shone$ step is internal to $\tm_2$,
      then $\tmtwo = \bin{\tm_1}{\tmtwo_2}$ with $\tm_2 \shone \tmtwo_2$,
      and we conclude taking $\tmtwo' := \bin{\tm'_1}{\tmtwo_2}$
      similarly as for case \ref{tog_shone_commutation:app_left_right}.
    \end{enumerate}
  \item
    If the $\tog$ step is internal to $\tm_2$,
    then $\bin{\tm_1}{\tm_2} \tog \bin{\tm_1}{\tm'_2} = \tm'$
    with $\tm_2 \tog \tm'_2$.
    We consider three further subcases, depending on whether the step
    $\bin{\tm_1}{\tm_2} \shone \tmtwo$ is at the root of the wrapper,
    internal to $\tm_1$, or internal to $\tm_2$:
    \begin{enumerate}
    \item
      If the $\shone$ step is at the root of the wrapper,
      then $\bin{\tm_1}{\tm_2} \shone \tm_1 = \tmtwo$.
      Taking $\tmtwo' := \tm_1$ we have:
      \[
        \xymatrix@C=.25cm@R=.5cm{
          \bin{\tm_1}{\tm_2}
          \ar[d]
          & \shone &
          \ar@{=}[d]
          \tm_1
        \\
          \bin{\tm_1}{\tm'_2}
          & \shone &
          \tm_1
        }
      \]
    \item
      If the $\shone$ step is internal to $\tm_1$,
      then $\bin{\tm_1}{\tm_2} \shone \bin{\tmtwo_1}{\tm_2} = \tmtwo$
      with $\tm_1 \shone \tmtwo_1$,
      and we conclude taking $\tmtwo' := \bin{\tmtwo_1}{\tm'_2}$
      similarly as for case \ref{tog_shone_commutation:app_left_right}.
    \item
      If the $\shone$ step is internal to $\tm_2$,
      then $\bin{\tm_1}{\tm_2} \shone \bin{\tm_1}{\tmtwo_2} = \tmtwo$
      with $\tm_2 \shone \tmtwo_2$,
      and we conclude by \ih
      similarly as for case \ref{tog_shone_commutation:app_left_left}.
    \end{enumerate}
  \end{enumerate}
\end{enumerate}
\end{proof}

\begin{proposition}[Forgetful reduction commutes with reduction]
\lprop{appendix:shrinking_simulation}
If $\tm \shone^+ \tmtwo$ and $\tm \tog^* \tm'$,
there exists a term $\tmtwo'$ such that
$\tm' \shone^+ \tmtwo'$ and $\tmtwo \tog^* \tmtwo'$.
Graphically:
\[
  \xymatrix@C=.1cm@R=.5cm{
    \tm \ar^{*}[d] & \shone^+ & \tmtwo \ar^{*}@{.>}[d] \\
    \tm' & \shone^+ & \tmtwo' \\
  }
\]
Furthermore,
if $\tm \shone^+ \tmtwo$ and $\tm$ is a $\tog$-normal form,
then $\tmtwo$ is also a normal form.
\end{proposition}
\begin{proof}
First we claim that
if $\tm \shone^+ \tmtwo$ and $\tm \tog \tm'$,
there exists a term $\tmtwo'$ such that
$\tm' \shone^+ \tmtwo'$ and $\tmtwo \tog^= \tmtwo'$.
This can be seen by induction on the number of $\shone$ steps
in a reduction sequence $\tm \shone^+ \tmtwo$, resorting to
the local commutation lemma~(\rlem{tog_shone_commutation}).

The main statement of the proposition can be seen by induction
on the number of steps in a reduction sequence $\tm \tog^* \tm'$,
resorting to the claim.

For the ``furthermore'' part in the statement,
it suffices to show that
if $\tm \shone \tmtwo$ in one step and $\tm$ is a $\tog$-normal form,
then $\tmtwo$ is also a normal form.
This is straightforward by induction on $\tm$.
\end{proof}

\begin{lemma}[Reduce/forget lemma]
\llem{appendix:reduce_shrink_lemma}
Let $\ltm \tobeta \ltmtwo$ be a $\beta$-step
and let $\ltm \tog \tmtwo$ be the corresponding step in $\lambdaG$.
Then $\tmtwo \shone \ltmtwo$.
\end{lemma}
\begin{proof}
We proceed by induction on $\ltm$:
\begin{enumerate}
\item
  $\ltm = \var$:
  Impossible, as there are no steps $\var \tobeta \ltmtwo$.
\item
  $\ltm = \lam{\var}{\ltm_1}$:
  Then the step must be of the form
  $\ltm
   = \lam{\var}{\ltm_1}
   \tobeta \lam{\var}{\ltmtwo_1}
   = \ltmtwo$
  with $\ltm_1 \tobeta \ltmtwo_1$,
  and the corresponding step must be of the form
  $\ltm
   = \lam{\var}{\ltm_1}
   \tog \lam{\var}{\tmtwo_1}
   = \tmtwo$
  where $\ltm_1 \tog \tmtwo_1$
  is the step corresponding to $\ltm_1 \tobeta \ltmtwo_1$.
  By \ih $\tmtwo_1 \shone \ltmtwo_1$,
  so
  $\tmtwo
   = \lam{\var}{\tmtwo_1}
   \shone \lam{\var}{\ltmtwo_1}
   = \ltmtwo$.
\item
  $\ltm = \ltm_1\,\ltm_2$:
  We consider three subcases, depending on whether the step is at the root,
  internal to $\ltm_1$, or internal to $\ltm_2$:
  \begin{enumerate}
  \item
    If the step is at the root,
    the step must be of the form
    $\ltm
     = (\lam{\var}{\ltm_{11}})\,\ltm_2
     \tobeta \ltm_{11}\sub{\var}{\ltm_2}$
    with $\ltm_1 = \lam{\var}{\ltm_{11}}$,
    and the corresponding step is
    $\ltm
     = \app{(\lam{\var}{\ltm_{11}})}{\ltm_2}
     \tog
     \ltm_{11}\sub{\var}{\ltm_2}\garb{\ltm_2}
     = \tmtwo$.
    Then:
    \[
      \begin{array}{rcl}
        \tmtwo
        & = &
        \ltm_{11}\sub{\var}{\ltm_2}\garb{\ltm_2}
      \\
        & \shone &
        \ltm_{11}\sub{\var}{\ltm_2}
      \\
        & = &
        \ltmtwo
      \end{array}
    \]
  \item
    If the step is internal to $\ltm_1$,
    the step must be of the form
    $\app{\ltm_1}{\ltm_2} \tobeta \app{\ltmtwo_1}{\ltm_2} = \ltmtwo$
    with $\ltm_1 \tobeta \ltmtwo_1$,
    and the corresponding step is
    $\app{\ltm_1}{\ltm_2} \tog \app{\tmtwo_1}{\ltm_2} = \tmtwo$
    where $\ltm_1 \tog \tmtwo_1$
    is the step corresponding to $\ltm_1 \tog \ltmtwo_1$.
    By \ih we have that $\tmtwo_1 \shone \ltmtwo_1$,
    so
    $\tmtwo
     = \app{\tmtwo_1}{\ltm_2}
     \shone \app{\ltmtwo_1}{\ltm_2}
     = \ltmtwo$.
  \item
    If the step is internal to $\ltm_2$,
    the proof is similar to the previous case.
  \end{enumerate}
\end{enumerate}
\end{proof}

\subsection{Proofs of \rsec{reduction_by_degrees} --- Reduction by degrees}
\lsec{appendix:reduction_by_degrees}

In this section we give detailed proofs of the results
about reduction by degrees stated in \rsec{lambdaG}.

\begin{remark}
$\tm\sctx$ is in $\tod{d}$-normal form
if and only if
$\tm$ and $\sctx$ are in $\tod{d}$-normal form.
\end{remark}

\begin{definition}[Steps and reduction sequences]
A {\em step of degree $d$} ---or just {\em step} if clear from the context-----
is formally a $5$-uple $\redex = (\gctx,\var^\typ,\tm,\sctx,\tmtwo)$
where $\gctx$ is an arbitrary context
and $\lam{\var^\typ}{\tm}$ is an abstraction of degree $d$.
The {\em source} of $\redex$ is
$\src{\redex} \eqdef \of{\gctx}{(\lam{\var}{\tm})\sctx\,\tmtwo}$
and its {\em target} is
$\tgt{\redex} \eqdef \of{\gctx}{\tm\sub{\var}{\tmtwo}\garb{\tmtwo}\sctx}$.
We write $\redex : \tm \tod{d} \tmtwo$ to mean that $\redex$ is a step
of degree $d$ with source $\tm$ and target $\tmtwo$.

A {\em forgetful step} ---or just {\em step} if clear from the context-----
is formally a triple $\redex = (\gctx,\tm,\tmtwo)$
where $\gctx$ is an arbitrary context and $\tm,\tmtwo$ are terms.
The {\em source} of $\redex$ is
$\src{\redex} \eqdef \of{\gctx}{\bin{\tm}{\tmtwo}}$
and its {\em target} is
$\tgt{\redex} \eqdef \of{\gctx}{\tm}$.
We write $\redex : \tm \shone \tmtwo$ to mean that $\redex$ is a
forgetful step of degree with source $\tm$ and target $\tmtwo$.

Steps of degree $d$ are generalized to {\em reduction sequences of degree $d$}
(and, respectively, {\em forgetful reduction sequences}),
which are sequences of composable steps of the corresponding kind.
Formally, a reduction sequence is a pair
$\redseq = ((\tm_0,\hdots,\tm_n),(\redex_1,\hdots,\redex_n))$
where $(\tm_0,\hdots,\tm_n)$ is a sequence of $n+1$ terms 
and $(\redex_1,\hdots,\redex_n)$ is a sequence of $n$ steps
$\src{\redex_i} = \tm_{i-1}$ and $\tgt{\redex_i} = \tm_{i}$
for all $i\in1..n$.
The notions of source and target are extended to reduction sequences
by declaring $\src{\redseq} = \tm_0$ and $\tgt{\redseq} = \tm_n$.
We write $\redseq : \tm \rtod{d} \tmtwo$ to mean that $\redseq$ is
a reduction sequence of degree $d$ with source $\tm$ and target $\tmtwo$.
Similarly,
we write $\redseq : \tm \shone^* \tmtwo$ to mean that $\redseq$ is
a forgetful reduction sequence with source $\tm$ and target $\tmtwo$

A step $\redex$ can be implicitly treated as the one-step reduction sequence
$((\src{\redex},\tgt{\redex}),\redex)$.
If $\tgt{\redseq} = \src{\redseqtwo}$,
we write $\redseq\,\redseqtwo$ for their composition,
defined as expected.
\end{definition}

\begin{definition}[Simultaneous reduction of degree $d$]
We define a relation $\tm \ptod{d} \tm'$,
meaning that there is a {\em multi-step} of degree $d$
from $\tm$ to $\tm'$,
inductively by the following rules:
\[
  \indrule{\rulePVar}{
    \emptyPremise
  }{
    \var \ptod{d} \var
  }
  \indrule{\rulePAbs}{
    \tm \ptod{d} \tm'
  }{
    \lam{\var}{\tm} \ptod{d} \lam{\var}{\tm'}
  }
  \indrule{\rulePAppOne}{
    \tm \ptod{d} \tm'
    \HS
    \tmtwo \ptod{d} \tmtwo'
  }{
    \tm\,\tmtwo \ptod{d} \tm'\,\tmtwo'
  }
\]
\[
  \indrule{\rulePAppTwo}{
    \tm \ptod{d} \tm'
    \HS
    \sctx \ptod{d} \sctx'
    \HS
    \tmtwo \ptod{d} \tmtwo'
    \HS
    \text{$\lam{\var}{\tm}$ is of degree $d$}
  }{
    (\lam{\var}{\tm})\sctx\,\tmtwo
    \ptod{d}
    \tm'\sub{\var}{\tmtwo'}\garb{\tmtwo'}\sctx'
  }
  \indrule{\rulePBin}{
    \tm \ptod{d} \tm'
    \HS
    \tmtwo \ptod{d} \tmtwo'
  }{
    \bin{\tm}{\tmtwo} \ptod{d} \bin{\tm'}{\tmtwo'}
  }
\]
\[
  \indrule{\rulePCtxHole}{
    \emptyPremise
  }{
    \ctxhole
    \ptod{d}
    \ctxhole
  }
  \indrule{\rulePCtxBin}{
    \sctx \ptod{d} \sctx'
    \HS
    \tm \ptod{d} \tm'
  }{
    \bin{\sctx}{\tm}
    \ptod{d}
    \bin{\sctx'}{\tm'}
  }
\]
If $\mstep$ is the derivation witnessing
a multi-step $\tm \ptodplus{d} \tm'$
We say that $\mstep$ is empty if it does not use the rule $\rulePAppTwo$.
We write $\mstep : \tm \ptodplus{d} \tm'$ if $\mstep$
uses the rule $\rulePAppTwo$ at least once.
\end{definition}

\begin{remark}[Simultaneous reduction of terms with memory]
\lremark{simultaneous_reduction_of_terms_with_garbage}
\quad
\begin{enumerate}
\item
  $\tm\sctx \ptod{d} \tmtwo$
  if and only if $\tmtwo$ is of the form $\tm'\sctx'$
  where $\tm \ptod{d} \tm'$
  and $\sctx \ptod{d} \sctx'$.
\item
  Furthermore, the set of derivations
  $\mstep : \tm\sctx \ptod{d} \tm'\sctx'$
  is in bijective correspondence
  with the set of pairs of derivations
  $\mstep_1 : \tm \ptod{d} \tm'$
  and
  $\mstep_2 : \sctx \ptod{d} \sctx'$.
\end{enumerate}
\end{remark}

\begin{lemma}[Properties of simultaneous reduction by degrees]
\llem{properties_of_ptod}
\quad
\begin{enumerate}
\item
  For each step $\redex : \tm \tod{d} \tm'$
  there is a multi-step $\steptomstep{\redex} : \tm \ptod{d} \tm'$.
\item
  For each multi-step $\mstep : \tm \ptod{d} \tm'$
  there is a reduction sequence $\msteptoderiv{\mstep} : \tm \rtod{d} \tm'$.
  Moreover, if $\mstep$ is non-empty,
  then $\msteptoderiv{\mstep}$ contains at least one step.
\item
  Reflexivity: $\tm \ptod{d} \tm$.
\item
  Substitution:
  If $\tm \ptod{d} \tm'$ and $\tmtwo \ptod{d} \tmtwo'$
  then $\tm\sub{\var}{\tmtwo} \ptod{d} \tm'\sub{\var}{\tmtwo'}$.
\end{enumerate}
\end{lemma}
\begin{proof}
All items are straightforward by induction.
\end{proof}

\begin{lemma}[Commutation of simultaneous reduction by degrees]
\llem{commutation_of_simultaneous_reduction_by_degrees}
Let $d, D \in \Natz$.
Given a step $\redex : \tm_1 \tod{d} \tm_2$
and a multi-step $\msteptwo : \tm_1 \ptod{D} \tm_3$,
there exists a term $\tm_4$,
a multi-step $\msteptwo/\redex : \tm_2 \ptod{D} \tm_4$ 
and a multi-step $\redex/\msteptwo : \tm_3 \ptod{d} \tm_4$.
Graphically:
\[
  \xymatrix@C=.5cm@R=.5cm{
    \tm_1 \ar^{d}[r] \ar@{=>}_{D}[d] & \tm_2 \ar@{=>}^{D}[d] \\
    \tm_3 \ar@{=>}_{d}[r] & \tm_4 \\
  }
\]
Furthermore:
\begin{enumerate}
\item
  If $d \neq D$ then $\redex/\msteptwo$ is non-empty,
  \ie $\redex/\msteptwo : \tm_3 \ptodplus{d} \tm_4$.
\item
  If $d \neq D$,
  the first step of $\msteptoderiv{\redex/\msteptwo}$
  determines the step $\redex$.
  More precisely,
  suppose that $\msteptoderiv{\redex_1/\msteptwo}$
  and $\msteptoderiv{\redex_2/\msteptwo}$
  start with the same step. Then $\redex_1 = \redex_2$.
\end{enumerate}
\end{lemma}
\begin{proof}
We prove a more general version of the statement including memories,
\ie we prove that for
$\redex : \anon_1 \tod{d} \anon_2$ and $\msteptwo : \anon_1 \ptod{D} \anon_3$,
there exist $\anon_4$
and $\msteptwo/\redex : \anon_2 \ptod{D} \anon_4$ 
and $\redex/\msteptwo : \anon_3 \ptod{d} \anon_4$,
where $\anon_1, \anon_2, \anon_3, \anon_4$
stand for either terms or memories.
We proceed by induction on $\anon_1$:
\begin{enumerate}
\item
  $\tm_1 = \var$:
  Impossible, as there are no reduction steps $\var \tod{d} \tm_2$.
\item
  $\tm_1 = \lam{\var}{\tmtwo_1}$:
  Then
  $\redex : \tm_1 = \lam{\var}{\tmtwo_1} \tod{d} \lam{\var}{\tmtwo_2} = \tm_2$
  with $\tmtwo_1 \tod{d} \tmtwo_2$
  and
  $\redextwo$ must be derived from the $\rulePAbs$ rule,
  so
  $\redextwo : \tm_1 = \lam{\var}{\tmtwo_1} \ptod{d} \lam{\var}{\tmtwo_3} = \tm_3$
  with $\tmtwo_1 \ptod{d} \tmtwo_3$.
  By \ih we have the diagram on the left, and we can construct the
  one on the right:
  \[
    \xymatrix@C=.5cm@R=.5cm{
      \tmtwo_1
        \ar^{d}[r] \ar@{=>}_{D}[d]
    &
      \tmtwo_2
        \ar@{=>}^{D}[d]
    \\
      \tmtwo_3
        \ar@{=>}_{d}[r]
    &
      \tmtwo_4
    }
    \HS
    \xymatrix@C=.5cm@R=.5cm{
      \lam{\var}{\tmtwo_1}
        \ar^{d}[r] \ar@{=>}_{D}[d]
    &
      \lam{\var}{\tmtwo_2}
        \ar@{=>}^{D}[d]
    \\
      \lam{\var}{\tmtwo_3}
        \ar@{=>}_{d}[r]
    &
      \lam{\var}{\tmtwo_4}
    }
  \]
  Furthermore, if $d \neq D$, using the \ih
  it is easy to show that $\redex/\msteptwo$ is non-empty
  and that $\redex/\msteptwo$ determines $\redex$.
\item
  $\tm_1 = \tmtwo_1\,\tmthree_1$:
  We consider three subcases, depending on whether
  $\redex$ is at the root, internal to $\tmtwo_1$, or internal to $\tmthree_1$:
  \begin{enumerate}
  \item
    If $\redex$ is at the root:
    Then $\tmtwo_1$ is a $\symg$-abstraction of degree $d$,
    \ie of the form $\tmtwo_1 = (\lam{\var}{\tmfour_1})\sctx_1$,
    and $\redex : \tm_1 = (\lam{\var}{\tmfour_1})\sctx_1\,\tmthree_1
                  \tod{d}
                  \tmfour_1\sub{\var}{\tmthree_1}\garb{\tmthree_1}\sctx_1
                  = \tm_2$.
    We consider two further subcases, depending on whether
    $\redextwo$ is derived using the $\rulePAppOne$ or the $\rulePAppTwo$ rule:
    \begin{enumerate}
    \item
      \label{commutation_of_simultaneous_reduction_by_degrees:app:root:app1}
      If $\redextwo$ is derived using the $\rulePAppOne$ rule:
      Then by~\rremark{simultaneous_reduction_of_terms_with_garbage}
      we have that
      $\redextwo : \tm_1 = (\lam{\var}{\tmfour_1})\sctx_1\,\tmthree_1
                   \ptod{D}
                   (\lam{\var}{\tmfour_3})\sctx_3\,\tmthree_3 = \tm_3$
      where $\tmfour_1 \ptod{D} \tmfour_3$
      and $\sctx_1 \ptod{D} \sctx_3$
      and $\tmthree_1 \ptod{D} \tmthree_3$.
      By~\rlem{properties_of_ptod} we can construct the following diagram,
      using reflexivity and $\rulePAppTwo$ on the bottom:
      \[
        \xymatrix@C=.5cm@R=.5cm{
          (\lam{\var}{\tmfour_1})\sctx_1\,\tmthree_1
            \ar^-{d}[r] \ar@{=>}_{D}[d]
        &
          \tmfour_1\sub{\var}{\tmthree_1}\garb{\tmthree_1}\sctx_1
            \ar@{=>}^{D}[d]
        \\
          (\lam{\var}{\tmfour_3})\sctx_3\,\tmthree_3
            \ar@{=>}_-{d}[r]
        &
          \tmfour_3\sub{\var}{\tmthree_3}\garb{\tmthree_3}\sctx_3
        }
      \]
    \item
      \label{commutation_of_simultaneous_reduction_by_degrees:app:root:app2}
      If $\redextwo$ is derived using the $\rulePAppTwo$ rule:
      Then note that $d = D$
      and we have that
      $\redextwo : \tm_1 = (\lam{\var}{\tmfour_1})\sctx_1\,\tmthree_1
                   \ptod{D}
                   \tmfour_3\sub{\var}{\tmthree_3}\garb{\tmthree_3}\sctx_3 = \tm_3$
      where $\tmfour_1 \ptod{D} \tmfour_3$
      and $\sctx_1 \ptod{D} \sctx_3$
      and $\tmthree_1 \ptod{D} \tmthree_3$.
      By~\rlem{properties_of_ptod} we can construct the following diagram,
      using reflexivity on the bottom:
      \[
        \xymatrix@C=.5cm@R=.5cm{
          (\lam{\var}{\tmfour_1})\sctx_1\,\tmthree_1
            \ar^-{d}[r] \ar@{=>}_{D}[d]
        &
          \tmfour_1\sub{\var}{\tmthree_1}\garb{\tmthree_1}\sctx_1
            \ar@{=>}^{D}[d]
        \\
          \tmfour_3\sub{\var}{\tmthree_3}\garb{\tmthree_3}\sctx_3
            \ar@{=>}_-{d}[r]
        &
          \tmfour_3\sub{\var}{\tmthree_3}\garb{\tmthree_3}\sctx_3
        }
      \]
    \end{enumerate}
  \item
    If $\redex$ is internal to $\tmtwo_1$:
    Then $\redex : \tm_1 = \tmtwo_1\,\tmthree_1 \tod{d} \tmtwo_2\,\tmthree_1$
    with $\tmtwo_1 \tod{d} \tmtwo_2$.
    We consider two further subcases, depending on whether
    $\redextwo$ is derived using the $\rulePAppOne$ or the $\rulePAppTwo$ rule:
    \begin{enumerate}
    \item
      If $\redextwo$ is derived using the $\rulePAppOne$ rule:
      Then $\redextwo : \tm_1 = \tmtwo_1\,\tmthree_1
                        \ptod{d} \tmtwo_3\,\tmthree_3 = \tm_3$
      with $\tmtwo_1 \ptod{d} \tmtwo_3$
      and $\tmthree_1 \ptod{d} \tmthree_3$.
      By \ih we have the diagram on the left, and we can construct the
      one on the right:
      \[
        \xymatrix@C=.5cm@R=.5cm{
          \tmtwo_1
            \ar^{d}[r] \ar@{=>}_{D}[d]
        &
          \tmtwo_2
            \ar@{=>}^{D}[d]
        \\
          \tmtwo_3
            \ar@{=>}_{d}[r]
        &
          \tmtwo_4
        }
        \HS
        \xymatrix@C=.5cm@R=.5cm{
          \tmtwo_1\,\tmthree_1
            \ar^{d}[r] \ar@{=>}_{D}[d]
        &
          \tmtwo_2\,\tmthree_1
            \ar@{=>}^{D}[d]
        \\
          \tmtwo_3\,\tmthree_3
            \ar@{=>}_{d}[r]
        &
          \tmtwo_4\,\tmthree_3
        }
      \]
    \item
      If $\redextwo$ is derived using the $\rulePAppTwo$ rule:
      Then $\tmtwo_1$ is a $\symg$-abstraction of degree $D$,
      \ie of the form $\tmtwo_1 = (\lam{\var}{\tmfour_1})\sctx_1$,
      and by~\rremark{simultaneous_reduction_of_terms_with_garbage}
      we have that
      $\redextwo : \tm_1 = (\lam{\var}{\tmfour_1})\sctx_1\,\tmthree_1
                   \ptod{D}
                   \tmfour_3\sub{\var}{\tmthree_3}\garb{\tmthree_3}\sctx_3$
      where $\tmfour_1 \ptod{D} \tmfour_3$
      and $\sctx_1 \ptod{D} \sctx_3$
      and $\tmthree_1 \ptod{D} \tmthree_3$.
      Moreover, since we know
      $\tmtwo_1 = (\lam{\var}{\tmfour_1})\sctx_1 \tod{d} \tmtwo_2$
      we consider two further subcases, depending on whether the step
      $\tmtwo_1 \tod{d} \tmtwo_2$
      is internal to $\tmfour_1$ or internal to $\sctx_1$.
      These subcases are similar; we only give the proof
      for the case in which the step is internal to $\tmfour_1$.
      In such case
      $\tmtwo_1 = (\lam{\var}{\tmfour_1})\sctx_1
          \tod{d} (\lam{\var}{\tmfour_2})\sctx_1$
      with $\tmfour_1 \tod{d} \tmfour_2$.
      By \ih we have the diagram on the left, and we can construct the
      one on the right,
      using~\rlem{properties_of_ptod}.
      On the right of the diagram, use $\rulePAppTwo$.
      On the bottom of the diagram,
      note that $\tmthree_3 \ptod{d} \tmthree_3$ by reflexivity
      so $\tmfour_3\sub{\var}{\tmthree_3} \ptod{d} \tmfour_4\sub{\var}{\tmthree_3}$:
      \[
        \xymatrix@C=.5cm@R=.5cm{
          \tmfour_1
            \ar^{d}[r] \ar@{=>}_{D}[d]
        &
          \tmfour_2
            \ar@{=>}^{D}[d]
        \\
          \tmfour_3
            \ar@{=>}_{d}[r]
        &
          \tmfour_4
        }
        \HS
        \xymatrix@C=.5cm@R=.5cm{
          (\lam{\var}{\tmfour_1})\sctx_1\,\tmthree_1
            \ar^-{d}[r] \ar@{=>}_{D}[d]
        &
          (\lam{\var}{\tmfour_2})\sctx_1\,\tmthree_1
            \ar@{=>}^{D}[d]
        \\
          \tmfour_3\sub{\var}{\tmthree_3}\garb{\tmthree_3}\sctx_3
            \ar@{=>}_-{d}[r]
        &
          \tmfour_4\sub{\var}{\tmthree_3}\garb{\tmthree_3}\sctx_3
        }
      \]
    \end{enumerate}
  \item
    If $\redex$ is internal to $\tmthree_1$:
    Then $\redex : \tm_1 = \tmtwo_1\,\tmthree_1 \tod{d} \tmtwo_1\,\tmthree_2$
    with $\tmthree_1 \tod{d} \tmthree_2$.
    We consider two further subcases, depending on whether
    $\redextwo$ is derived using the $\rulePAppOne$ or the $\rulePAppTwo$ rule:
    \begin{enumerate}
    \item
      If $\redextwo$ is derived using the $\rulePAppOne$ rule:
      Then $\redextwo : \tm_1 = \tmtwo_1\,\tmthree_1
                        \ptod{d} \tmtwo_3\,\tmthree_3 = \tm_3$
      with $\tmtwo_1 \ptod{d} \tmtwo_3$
      and $\tmthree_1 \ptod{d} \tmthree_3$.
      By \ih we have the diagram on the left, and we can construct the
      one on the right:
      \[
        \xymatrix@C=.5cm@R=.5cm{
          \tmthree_1
            \ar^{d}[r] \ar@{=>}_{D}[d]
        &
          \tmthree_2
            \ar@{=>}^{D}[d]
        \\
          \tmthree_3
            \ar@{=>}_{d}[r]
        &
          \tmthree_4
        }
        \HS
        \xymatrix@C=.5cm@R=.5cm{
          \tmtwo_1\,\tmthree_1
            \ar^{d}[r] \ar@{=>}_{D}[d]
        &
          \tmtwo_1\,\tmthree_2
            \ar@{=>}^{D}[d]
        \\
          \tmtwo_3\,\tmthree_3
            \ar@{=>}_{d}[r]
        &
          \tmtwo_3\,\tmthree_4
        }
      \]
    \item
      \label{commutation_of_simultaneous_reduction_by_degrees:app:right:app2}
      If $\redextwo$ is derived using the $\rulePAppTwo$ rule:
      Then $\tmtwo_1$ is a $\symg$-abstraction of degree $D$,
      \ie of the form $\tmtwo_1 = (\lam{\var}{\tmfour_1})\sctx_1$,
      and by~\rremark{simultaneous_reduction_of_terms_with_garbage}
      we have that
      $\redextwo : \tm_1 = (\lam{\var}{\tmfour_1})\sctx_1\,\tmthree_1
                   \ptod{D}
                   \tmfour_3\sub{\var}{\tmthree_3}\garb{\tmthree_3}\sctx_3$
      where $\tmfour_1 \ptod{D} \tmfour_3$
      and $\sctx_1 \ptod{D} \sctx_3$
      and $\tmthree_1 \ptod{D} \tmthree_3$.
      By \ih we have the diagram on the left, and we can construct the
      one on the right.
      On the right of the diagram, use $\rulePAppTwo$.
      On the bottom of the diagram,
      note that $\tmfour_3 \ptod{d} \tmfour_3$ by reflexivity
      so $\tmfour_3\sub{\var}{\tmthree_3} \ptod{d} \tmfour_3\sub{\var}{\tmthree_4}$:
      \[
        \xymatrix@C=.5cm@R=.5cm{
          \tmthree_1
            \ar^{d}[r] \ar@{=>}_{D}[d]
        &
          \tmthree_2
            \ar@{=>}^{D}[d]
        \\
          \tmthree_3
            \ar@{=>}_{d}[r]
        &
          \tmthree_4
        }
        \HS
        \xymatrix@C=.5cm@R=.5cm{
          (\lam{\var}{\tmfour_1})\sctx_1\,\tmthree_1
            \ar^-{d}[r] \ar@{=>}_{D}[d]
        &
          (\lam{\var}{\tmfour_1})\sctx_1\,\tmthree_2
            \ar@{=>}^{D}[d]
        \\
          \tmfour_3\sub{\var}{\tmthree_3}\garb{\tmthree_3}\sctx_3
            \ar@{=>}_-{d}[r]
        &
          \tmfour_3\sub{\var}{\tmthree_4}\garb{\tmthree_4}\sctx_3
        }
      \]
    \end{enumerate}
  \end{enumerate}
  Furthermore, note that if $d \neq D$,
  then the multi-step $\redex/\msteptwo : \tm_3 \ptod{d} \tm_4$
  at the bottom of the diagram
  must be non-empty.
  Indeed,
  case~\ref{commutation_of_simultaneous_reduction_by_degrees:app:root:app1},
  uses exactly one occurrence of the $\rulePAppTwo$ rule
  to construct $\redex/\msteptwo$.
  Case~\ref{commutation_of_simultaneous_reduction_by_degrees:app:root:app2}
  is impossible, because in such case $d = D$.
  In the remaining cases, the bottom of the diagram is constructed
  by resorting to the \ih,
  which means that $\redex/\msteptwo$ is non-empty.
  An important observation is that in
  case~\ref{commutation_of_simultaneous_reduction_by_degrees:app:right:app2}
  the argument is not erased,
  because it is always kept as a memorized term.

  Furthermore, if $d \neq D$,
  to see that the first step of $\msteptoderiv{\redex/\msteptwo}$
  determines the step $\redex$,
  consider the first step $\redexthree$ of $\msteptoderiv{\redex/\msteptwo}$
  and note that it its $\lambda$-abstraction can be uniquely traced back
  to the $\lambda$-abstraction of $\redex$ (\ie it has a unique ancestor).
  Indeed,
  in case~\ref{commutation_of_simultaneous_reduction_by_degrees:app:root:app1}
  the step at the bottom has $\redex$ as its unique ancestor.
  Case~\ref{commutation_of_simultaneous_reduction_by_degrees:app:root:app2}
  is impossible, because in such case $d = D$.
  In the remaining cases, it suffices to resort to the \ih.
\item
  \label{commutation_of_simultaneous_reduction_by_degrees:bin}
  $\tm_1 = \bin{\tmtwo_1}{\tmthree_1}$:
  We consider two subcases, depending on whether
  $\redex$ is internal to $\tmtwo_1$ or internal to $\tmthree_1$:
  \begin{enumerate}
  \item
    If $\redex$ is internal to $\tmtwo_1$:
    Then $\redex : \tm_1 = \bin{\tmtwo_1}{\tmthree_1}
                   \tod{d} \bin{\tmtwo_2}{\tmthree_1} = \tm_2$.
    Note that $\redextwo$ must be derived using the $\rulePBin$ rule,
    so $\redextwo : \tm_1 = \bin{\tmtwo_1}{\tmthree_1}
                    \ptod{d} \bin{\tmtwo_3}{\tmthree_3} = \tm_3$.
    By \ih we have the diagram on the left, and we can construct the
    one on the right:
    \[
      \xymatrix@C=.5cm@R=.5cm{
        \tmtwo_1
          \ar^{d}[r] \ar@{=>}_{D}[d]
      &
        \tmtwo_2
          \ar@{=>}^{D}[d]
      \\
        \tmtwo_3
          \ar@{=>}_{d}[r]
      &
        \tmtwo_4
      }
      \HS
      \xymatrix@C=.5cm@R=.5cm{
        \bin{\tmtwo_1}{\tmthree_1}
          \ar^-{d}[r] \ar@{=>}_{D}[d]
      &
        \bin{\tmtwo_2}{\tmthree_1}
          \ar@{=>}^{D}[d]
      \\
        \bin{\tmtwo_3}{\tmthree_3}
          \ar@{=>}_-{d}[r]
      &
        \bin{\tmtwo_4}{\tmthree_3}
      }
    \]
  \item
    If $\redex$ is internal to $\tmthree_1$:
    Similar to the previous case.
  \end{enumerate}
  Furthermore, if $d \neq D$,
  using the \ih it is easy to show that $\redex/\msteptwo$ is non-empty
  and that $\redex/\msteptwo$ determines $\redex$.
\item
  $\sctx_1 = \ctxhole$:
  Impossible, as there are no steps $\ctxhole \tod{d} \sctx_2$.
\item
  $\sctx_1 = \bin{\sctx'_1}{\tm_1}$:
  Similar to case~\ref{commutation_of_simultaneous_reduction_by_degrees:bin}.
\end{enumerate}
\end{proof}

\begin{proposition}[Commutation of reduction by degrees]
\lprop{appendix:commutation_by_degrees}
Let $d, D \in \Natz$. Then $\tod{d}$ and $\tod{D}$ commute.
More precisely,
given reduction sequences
    $\redseq : \tm_1 \rtod{d} \tm_2$
and $\redseqtwo : \tm_1 \rtod{D} \tm_3$,
there exists a term $\tm_4$
and reduction sequences
    $\redseqtwo/\redseq : \tm_2 \rtod{D} \tm_4$
and $\redseq/\redseqtwo : \tm_3 \rtod{d} \tm_4$.
Graphically:
\[
  \xymatrix@C=.5cm@R=.5cm{
    \tm_1
      \ar@{->>}^-{d}[r] \ar@{->>}_-{D}[d]
  &
    \tm_2
      \ar^-{D}@{.>>}[d]
  \\
    \tm_3
      \ar@{.>>}_-{d}[r]
  &
    \tm_4
  \\
  }
\]
The reduction sequence $\redseq/\redseqtwo$
is called the {\em projection of $\redseq$ after $\redseqtwo$}
and symmetrically for $\redseqtwo/\redseq$.
Furthermore:
\begin{enumerate}
\item
  If $d \neq D$, then $\redseq/\redseqtwo$
  contains at least as many steps as $\redseq$.
\item
  If $d \neq D$,
  then $\redseq/\redseqtwo$ determines $\redseq$.
  More precisely,
  if $\redseq_1/\redseqtwo = \redseq_2/\redseqtwo$
  then $\redseq_1 = \redseq_2$.
\end{enumerate}
\end{proposition}
\begin{proof}
Recall that
$\tod{d} \,\,\subseteq\,\, \ptod{d} \,\,\subseteq\,\, \rtod{d}$.
by \rlem{properties_of_ptod}.
We prove this in two stages.

First, given a reduction sequence $\redseq : \tm_1 \rtod{d} \tm_2$
and a multi-step $\msteptwo : \tm_1 \ptod{D} \tm_3$,
we claim that there exists a term $\tm_4$
and constructing $\redseq/\msteptwo : \tm_3 \rtod{d} \tm_4$
and $\msteptwo/\redseq : \tm_2 \ptod{D} \tm_4$ as follows,
by induction on $\redseq$,
resorting to~\rlem{commutation_of_simultaneous_reduction_by_degrees}
for the constructions of $\redex/\msteptwo$ and $\msteptwo/\redex$.
\[
  \begin{array}{rcl}
    \emptyseq/\msteptwo
  & \eqdef &
    \emptyseq
  \\
    (\redex\,\redseq')/\msteptwo
  & \eqdef &
    \msteptoderiv{\redex/\msteptwo}(\redseq'/(\msteptwo/\redex))
  \end{array}
  \HS
  \begin{array}{rcl}
    \msteptwo/\emptyseq
  & \eqdef &
    \msteptwo
  \\
    \msteptwo/(\redex\,\redseq')
  & \eqdef &
    (\msteptwo/\redex)/\redseq'
  \end{array}
\]
Recall from~\rlem{properties_of_ptod}
that if $\mstep : \tmthree \ptod{d} \tmthree'$
is a multi-step,
then $\msteptoderiv{\mstep} : \tmthree \rtod{d} \tmthree'$
denotes a reduction sequence.
The inductive cases correspond to the following diagram:
\[
  \xymatrix@C=2cm{
    \ar^{\redex}[r]
    \ar@{=>}_{\msteptwo}[d]
  & 
    \ar@{->>}^{\redseq'}[r]
    \ar@{=>}_{\msteptwo/\redex}[d]
  &
    \ar@{=>}_{(\msteptwo/\redex)/\redseq'}[d]
  \\
    \ar@{->>}_{\msteptoderiv{\redex/\msteptwo}}[r]
  &
    \ar@{->>}_{\redseq'/(\msteptwo/\redex)}[r]
  &
  }
\]
For the general case, we proceed by induction on $\redseqtwo$
resorting to the previous construction for the constructions of
$\steptomstep{\redextwo}/\redseq$ and $\redseq/\steptomstep{\redextwo}$:
\[
  \begin{array}{rcl}
    \redseq/\emptyseq
  & \eqdef &
    \redseq
  \\
    \redseq/(\redextwo\,\redseqtwo')
  & \eqdef &
    (\redseq/\steptomstep{\redextwo})/\redseqtwo'
  \end{array}
  \HS
  \begin{array}{rcl}
    \emptyseq/\redseq
  & \eqdef &
    \emptyseq
  \\
    (\redextwo\,\redseqtwo')/\redseq
  & \eqdef &
    \msteptoderiv{\steptomstep{\redextwo}/\redseq}\,
    (\redseqtwo'/(\redseq/\steptomstep{\redextwo})
  \end{array}
\]
Recall from~\rlem{properties_of_ptod}
that if $\redex : \tmthree \tod{d} \tmthree'$
is a step,
then $\steptomstep{\redex} : \tmthree \ptod{d} \tmthree'$
denotes a multi-step.
The inductive cases correspond to the following diagram:
\[
  \xymatrix@C=2cm{
    \ar@{->}_{\redextwo}[d]
    \ar@{->>}^{\redseq}[r]
  &
    \ar^{\msteptoderiv{\steptomstep{\redextwo}/\redseq}}[d]
  \\
    \ar@{->>}_{\redseqtwo'}[d]
    \ar@{->>}^{\redseq/\steptomstep{\redextwo}}[r]
  &
    \ar@{->>}^{\redseqtwo'/(\redseq/\steptomstep{\redextwo})}[d]
  \\
    \ar@{->>}^{(\redseq/\steptomstep{\redextwo})/\redseqtwo'}[r]
  &
  }
\]
Furthermore, if $d \neq D$,
note that $\redex/\msteptwo : \tm \ptodplus{d} \tmtwo$
by~\rlem{commutation_of_simultaneous_reduction_by_degrees},
so $\msteptoderiv{\redex/\msteptwo} : \tm \todplus{d} \tmtwo$
by~\rlem{properties_of_ptod}.
Then by induction on $\redseq$, we can show that
$\redseq/\msteptwo$ contains at least as many steps as $\redseq$.
Finally, by induction on $\redseqtwo$,
we can show that
$\redseq/\redseqtwo$ contains at least as many steps as $\redseq$.

Furthermore, to see that
$\redseq/\redseqtwo$ determines $\redseq$,
we proceed in stages:
\begin{enumerate}
\item
  First,
  if $\msteptoderiv{\redex_1/\msteptwo}$
  and $\msteptoderiv{\redex_2/\msteptwo}$
  start with the same step, then $\redex_1 = \redex_2$
  by~\rlem{commutation_of_simultaneous_reduction_by_degrees}.
\item
  Second, we can see that if $\redseq_1/\msteptwo = \redseq_2/\msteptwo$
  then $\redseq_1 = \redseq_2$
  by induction on $\redseq_1$.
  Note that if $\redseq_1/\msteptwo = \redseq_2/\msteptwo$
  then $\redseq_1$ and $\redseq_2$ are either both empty
  or both non-empty,
  because if $\redseq_1 = \redex_1\,\redseq'_1$
  then by definition
  $\redseq_1/\msteptwo =
   \msteptoderiv{\redex_1/\msteptwo}\,(\redseq'_1/(\msteptwo/\redex_1))$
  and~\rlem{commutation_of_simultaneous_reduction_by_degrees}
  ensures that $\redex_1/\msteptwo$ is non-empty whenever $d \neq D$,
  so $\redseq_2/\msteptwo$ is non-empty,
  and hence $\redseq_2$ is non-empty.
  The base case is immediate.
  For the induction step, when $\redseq_1 = \redex_1\,\redseq'_1$
  and $\redseq_2 = \redex_2\,\redseq'_2$
  we have that $\redseq_1/\msteptwo = \redseq_2/\msteptwo$,
  so by definition
  $\msteptoderiv{\redex_1/\msteptwo}\,(\redseq'_1/(\msteptwo/\redex_1))
  = \msteptoderiv{\redex_2/\msteptwo}\,(\redseq'_2/(\msteptwo/\redex_2))$.
  As before~\rlem{commutation_of_simultaneous_reduction_by_degrees}
  ensures that $\msteptoderiv{\redex_1/\msteptwo}$
  and $\msteptoderiv{\redex_2/\msteptwo}$ are non-empty,
  so they must start with the same step.
  Hence by~\rlem{commutation_of_simultaneous_reduction_by_degrees}
  we have that $\redex_1 = \redex_2$.
  This in turn implies that
  $\redseq'_1/(\msteptwo/\redex_1) = \redseq'_2/(\msteptwo/\redex_2)$,
  so by \ih $\redseq'_1 = \redseq'_2$.
\item
  Finally, by induction on $\redseqtwo$
  we can see that if $\redseq_1/\redseqtwo = \redseq_2/\redseqtwo$
  then $\redseq_1 = \redseq_2$,
  resorting to the previous item.
\end{enumerate}
\end{proof}

\begin{lemma}[A term reduces to its simplification, by degrees]
\llem{tm_reduces_to_simpk_by_degrees}
For every term $\tm$ and for all $d \geq 1$
we have that $\tm \rtod{d} \simpd{\tm}$.
\end{lemma}
\begin{proof}
The proof is essentially the same proof as that of~\rlem{tm_reduces_to_simpk},
noting that whenever a redex is contracted, its degree is exactly $d$.
\end{proof}

\begin{lemma}[Substitution of $\tod{d}$-normal forms]
\llem{substitution_of_tod_normal_forms}
\quad
\begin{enumerate}
\item
  \label{substitution_of_tod_normal_forms:abstractions}
  If $\tm$ and $\tmtwo$ are not $\symg$-abstractions of degree $d$,
  then $\tm\sub{\var}{\tmtwo}$ is not a $\symg$-abstraction of degree $d$.
\item
  \label{substitution_of_tod_normal_forms:preservation}
  Let $\tm$ and $\tmtwo$ be terms in $\tod{d}$-normal form
  such that $\tmtwo$ is not an abstraction of degree $d$.
  Then $\tm\sub{\var}{\tmtwo}$ is in $\tod{d}$-normal form.
\end{enumerate}
\end{lemma}
\begin{proof}
We prove the two items separately:
\begin{enumerate}
\item
  By induction on $\tm$:
  \begin{enumerate}
  \item
    $\tm = \var$: 
    Then $\tm\sub{\var}{\tmtwo} = \tmtwo$, which is not a $\symg$-abstraction of
    degree $d$ by hypothesis.
  \item
    $\tm = \vartwo \neq \var$: 
    Then $\tm\sub{\var}{\tmtwo} = \vartwo$ is not a $\symg$-abstraction.
  \item
    $\tm = \lam{\vartwo}{\tm'}$: 
    By $\alpha$-conversion
    we may assume that $\vartwo \notin \set{\var}\cup\fv{\tmtwo}$.
    Note that $\tm$ is a $\symg$-abstraction but, by hypothesis, it cannot
    be of degree $d$.
    By the substitution lemma~(\rlem{typing_substitution_lemma})
    we have that
    $\typeof{\tm\sub{\var}{\tmtwo}}
     = \typeof{\lam{\vartwo}{\tm'\sub{\var}{\tmtwo}}}
     = \typeof{\lam{\vartwo}{\tm'}}$,
    so $\tm\sub{\var}{\tmtwo} = \lam{\vartwo}{\tm'\sub{\var}{\tmtwo}}$
    is a $\symg$-abstraction, but it is not of degree $d$.
  \item
    $\tm = \tm_1\,\tm_2$: 
    Then
    $\tm\sub{\var}{\tmtwo}
     = \tm_1\sub{\var}{\tmtwo}\,\tm_2\sub{\var}{\tmtwo}$ 
    is an application,
    hence not a $\symg$-abstraction.
  \item
    $\tm = \bin{\tm_1}{\tm_2}$: 
    Since
    $\tm$ is not a $\symg$-abstraction of degree $d$,
    we have that
    $\tm_1$ is also not a $\symg$-abstraction of degree $d$.
    By \ih,
    $\tm_1\sub{\var}{\tmtwo}$ is not a $\symg$-abstraction of degree $d$,
    so
    $\tm\sub{\var}{\tmtwo}
     = \bin{\tm_1\sub{\var}{\tmtwo}}{\tm_2\sub{\var}{\tmtwo}}$
    is not a $\symg$-abstraction of degree $d$.
  \end{enumerate}
\item
  By induction $\tm$:
  \begin{enumerate}
  \item
    $\tm = \var$:
    Then $\tm\sub{\var}{\tmtwo} = \tmtwo$ is in $\tod{d}$-normal form.
  \item
    $\tm = \vartwo \neq \var$:
    Then $\tm\sub{\var}{\tmtwo} = \vartwo$ is in $\tod{d}$-normal form.
  \item
    $\tm = \lam{\vartwo}{\tm'}$:
    By $\alpha$-conversion
    we may assume that $\vartwo \notin \set{\var}\cup\fv{\tmtwo}$.
    By \ih, $\tm'\sub{\var}{\tmtwo}$ is in $\tod{d}$-normal form,
    so $\tm\sub{\var}{\tmtwo} = \lam{\vartwo}{\tm'\sub{\var}{\tmtwo}}$
    is also in $\tod{d}$-normal form.
  \item
    $\tm = \tm_1\,\tm_2$:
    By \ih, $\tm_1\sub{\var}{\tmtwo}$ and $\tm_2\sub{\var}{\tmtwo}$
    are in $\tod{d}$-normal form.
    To show that the whole term
    $\tm\sub{\var}{\tmtwo}
     = \tm_1\sub{\var}{\tmtwo}\,\tm_2\sub{\var}{\tmtwo}$
    is a normal form,
    we are only left to show that the term does not have
    a $\tod{d}$-redex at the root,
    \ie that $\tm_1\sub{\var}{\tmtwo}$ is not a $\symg$-abstraction
    of degree $d$.
    Note that $\tm_1$ cannot be a $\symg$-abstraction of degree $d$,
    for otherwise $\tm = \tm_1\,\tm_2$ would be a redex of degree $d$,
    but we know by hypothesis that $\tm$ is in $\tod{d}$-normal form.
    Hence by item~\ref{substitution_of_tod_normal_forms:abstractions}
    of this lemma,
    $\tm_1\sub{\var}{\tmtwo}$ is not a $\symg$-abstraction of degree $d$,
    as required.
  \item
    $\tm = \bin{\tm_1}{\tm_2}$:
    By \ih, $\tm_1\sub{\var}{\tmtwo}$ and $\tm_2\sub{\var}{\tmtwo}$
    are in $\tod{d}$-normal form,
    so $\tm\sub{\var}{\tmtwo}
        = \bin{\tm_1\sub{\var}{\tmtwo}}{\tm_2\sub{\var}{\tmtwo}}$
    is also in $\tod{d}$-normal form.
  \end{enumerate}
\end{enumerate}
\end{proof}

\begin{lemma}[Simplification does not create abstractions, by degrees]
\llem{simpk_does_not_create_abstractions_by_degrees}
If $\tm$ is not a $\symg$-abstraction of degree $d$,
then $\simpd{\tm}$ is not a $\symg$-abstraction of degree $d$.
\end{lemma}
\begin{proof}
By induction on $\tm$:
\begin{enumerate}
\item
  $\tm = \var$:
  Then $\simpd{\tm} = \var$ is not a $\symg$-abstraction of degree $d$.
\item
  $\tm = \lam{\vartwo}{\tmtwo}$:
  Note that $\tm$ is an abstraction but, by hypothesis, it cannot be
  of degree $d$.
  By the fact that a term reduces to its simplification~(\rlem{tm_reduces_to_simpk})
  and by the substitution lemma~(\rlem{typing_substitution_lemma})
  we know that $\typeof{\simpd{\tm}}$ is not of degree $d$,
  so in particular it cannot be a $\symg$-abstraction of degree $d$.
\item
  $\tm = \app{(\lam{\vartwo}{\tmtwo})\sctx}{\tmthree}$,
  where $(\lam{\vartwo}{\tmtwo})\sctx$ is a $\symg$-abstraction of degree $d$:
  Then
  $\simpd{\tm}
   = \simpd{\tmtwo}\sub{\var}{\simpd{\tmthree}}\garb{\simpd{\tmthree}}\simpd{\sctx}$.
  To show that this term is not a $\symg$-abstraction of degree $d$,
  it suffices to show that $\simpd{\tmtwo}\sub{\var}{\simpd{\tmthree}}$
  is not a $\symg$-abstraction of degree $d$.
  Note that the abstraction $\lam{\vartwo}{\tmtwo}$
  is of type $\typ\to\typtwo$
  where $\typeof{\tmtwo} = \typtwo$ and $\typeof{\tmthree} = \typ$.
  In particular, since the abstraction $\lam{\vartwo}{\tmtwo}$
  is of degree $d$, we have that $\height{\typ\to\typtwo} = d$.
  Furthermore,
  by the fact that a term reduces to its simplification~(\rlem{tm_reduces_to_simpk})
  and by subject reduction~(\rprop{subject_reduction}),
  we know that
  $\height{\typeof{\simpd{\tmtwo}}}
   = \height{\typeof{\tmtwo}}
   = \height{\typtwo} < d$
  and
  $\height{\typeof{\simpd{\tmthree}}}
  = \height{\typeof{\tmthree}}
  = \height{\typ} < d$.
  In particular, $\simpd{\tmtwo}$ and $\simpd{\tmthree}$
  cannot be $\symg$-abstractions of degree $d$.
  Finally
  by~\rlem{substitution_of_tod_normal_forms}(\ref{substitution_of_tod_normal_forms:abstractions})
  this means that $\simpd{\tmtwo}\sub{\var}{\simpd{\tmthree}}$
  is not a $\symg$-abstraction of degree $d$,
  as required.
\item
  $\tm = \app{\tmtwo}{\tmthree}$,
  where $\tmtwo$ is not a $\symg$-abstraction of degree $d$:
  Then
  $\simpd{\tm} = \app{\simpd{\tmtwo}}{\simpd{\tmthree}}$
  is an application, hence not a $\symg$-abstraction of degree $d$.
\item
  $\tm = \bin{\tmtwo}{\tmthree}$:
  Since $\tm$ is not a $\symg$-abstraction of degree $d$,
  we know that $\tmtwo$ is also not a $\symg$-abstraction of degree $d$.
  By \ih,
  $\simpd{\tmtwo}$ is not a $\symg$-abstraction of degree $d$,
  so
  $\simpd{\tm} = \bin{\simpd{\tmtwo}}{\simpd{\tmthree}}$
  is not a $\symg$-abstraction of degree $d$.
\end{enumerate}
\end{proof}

\begin{lemma}[The simplification of a term is normal, by degrees]
\llem{simpk_is_normal_by_degrees}
$\simpd{\tm}$ is in $\tod{d}$-normal form.
\end{lemma}
\begin{proof}
By induction on $\tm$,
generalizing the statement also for memories,
\ie showing that $\simpd{\sctx}$ is in $\tod{d}$-normal form:
\begin{enumerate}
\item
  $\tm = \var$:
  Then $\simpd{\tm} = \var$ is in $\tod{d}$-normal form.
\item
  $\tm = \lam{\var}{\tmtwo}$:
  Then $\simpd{\tm} = \lam{\var}{\simpd{\tmtwo}}$
  is in $\tod{d}$-normal form
  because $\simpd{\tmtwo}$ is in $\tod{d}$-normal form by \ih.
\item
  $\tm = \app{(\lam{\var}{\tmtwo})\sctx}{\tmthree}$,
  where $(\lam{\var}{\tmtwo})\sctx$ is a $\symg$-abstraction of degree $d$:
  Then
  $\simpd{\tm}
  = \simpd{\tmtwo}\sub{\var}{\simpd{\tmthree}}
      \garb{\simpd{\tmthree}}
      \simpd{\sctx}$.
  Note that, by \ih,
  $\simpd{\tmtwo}$, $\simpd{\tmthree}$, and $\simpd{\sctx}$
  are in $\tod{d}$-normal form.

  Since $(\lam{\var}{\tmtwo})\sctx$ is an abstraction of degree $d$,
  we know that $\height{\typeof{\tmthree}} < d$.
  Hence
  by the fact that a term reduces to its simplification~(\rlem{tm_reduces_to_simpk})
  and by subject reduction~(\rprop{subject_reduction})
  we know that $\height{\typeof{\simpd{\tmthree}}} < d$.
  In particular, $\simpd{\tmthree}$ is not an abstraction of degree $d$.
  This allows us to apply~\rlem{substitution_of_tod_normal_forms}(\ref{substitution_of_tod_normal_forms:preservation})
  to conclude that $\simpd{\tmtwo}\sub{\var}{\simpd{\tmthree}}$
  is in $\tod{d}$-normal form.
  This lets us conclude that
  $\simpd{\tmtwo}\sub{\var}{\simpd{\tmthree}}
     \garb{\simpd{\tmthree}}
     \simpd{\sctx}$
  is in $\tod{d}$-normal form.
\item
  $\tm = \app{\tmtwo}{\tmthree}$,
  where $\tmtwo$ is not a $\symg$-abstraction of degree $d$.
  Then $\simpd{\tm} = \app{\simpd{\tmtwo}}{\simpd{\tmthree}}$,
  where by \ih we have that $\simpd{\tmtwo}$ and $\simpd{\tmthree}$
  are in $\tod{d}$-normal form,
  and by \rlem{simpk_does_not_create_abstractions_by_degrees}
  we have that
  $\simpd{\tmtwo}$ is not a $\symg$-abstraction of degree $d$.
  Hence $\simpd{\tm}$ is in $\tod{d}$-normal form.
\item
  $\tm = \bin{\tmtwo}{\tmthree}$:
  Then $\simpd{\tm} = \bin{\simpd{\tmtwo}}{\simpd{\tmthree}}$
  and we conclude by \ih.
\item
  $\sctx = \ctxhole$:
  Immediate, as $\simpd{\ctxhole} = \ctxhole$ is in $\tod{d}$-normal form.
\item
  $\sctx = \bin{\sctx}{\tm}$:
  Then $\simpd{\sctx} = \bin{\simpd{\sctx}}{\simpd{\tm}}$
  and we conclude by \ih.
\end{enumerate}
\end{proof}

\begin{lemma}[Reduction does not create redexes of higher degree]
\llem{reduction_does_not_create_higher_degree_redexes}
Let $d \leq D$ and suppose that $\tm \tod{d} \tmtwo$.
\begin{enumerate}
\item
  \label{reduction_does_not_create_higher_degree_redexes:abstractions}
  If $\tm$ is not a $\symg$-abstraction of degree $D$,
  then $\tmtwo$ is not a $\symg$-abstraction of degree $D$.
\item
  \label{reduction_does_not_create_higher_degree_redexes:preservation}
  If $\tm$ is in $\tod{D}$-normal form,
  then $\tmtwo$ is also in $\tod{D}$-normal form.
\end{enumerate}
\end{lemma}
\begin{proof}
We prove the two items independently:
\begin{enumerate}
\item
  By induction on $\tm$:
  \begin{enumerate}
  \item
    $\tm = \var$:
    This case is impossible, as there are no reduction steps $\tm \tod{d} \tmtwo$.
  \item
    $\tm = \lam{\var}{\tm'}$:
    Note that $\tm$ is a $\symg$-abstraction, so by \ih it cannot be of
    degree $D$, that is, $\height{\typeof{\tm}} \neq D$.
    By subject reduction~(\rprop{subject_reduction})
    we have that $\height{\typeof{\tmtwo}} = \height{\typeof{\tm}} \neq D$,
    so $\tmtwo$ cannot be a $\symg$-abstraction of degree $D$.
  \item
    $\tm = \tm_1\,\tm_2$:
    We consider three subcases, depending on whether the reduction is
    at the root, internal to $\tm_1$, internal to $\tm_2$:
    \begin{enumerate}
    \item
      If the reduction is at the root:
      Then $\tm_1 = (\lam{\var}{\tm'_1})\sctx$ is an abstraction of degree $d$,
      and the step is of the form
      $\tm
       = (\lam{\var}{\tm'_1})\sctx\,\tm_2
       \tod{d}
         \tm'_1\sub{\var}{\tm_2}\garb{\tm_2}\sctx
       = \tm'$.
      Note that $\lam{\var}{\tm'_1}$ is an abstraction of degree $d$,
      so its type is of the form $\typ \to \typtwo$ with
      $\height{\typ \to \typtwo} = d$.
      The type of the body of the abstraction
      is $\typeof{\tm'_1} = \typtwo$,
      so $\height{\typeof{\tm'_1}} = \height{\typtwo} < d \leq D$,
      and the type of the argument of the abstraction
      is $\typeof{\tm_2} = \typ$,
      so $\height{\typeof{\tm_2}} = \height{\typ} < d \leq D$.
      This means that $\tm'_1$ and $\tm_2$ cannot be
      $\symg$-abstractions of degree $D$.
      Hence by \rlem{substitution_of_tod_normal_forms}(\ref{substitution_of_tod_normal_forms:abstractions})
      we have that $\tm'_1\sub{\var}{\tm_2}$
      is not a $\symg$-abstractions of degree $D$.
      From this we conclude that $\tm' = \tm'_1\sub{\var}{\tm_2}\garb{\tm_2}\sctx$
      is not a $\symg$-abstractions of degree $D$.
    \item
      If the reduction is internal to $\tm_1$:
      Then the step is of the form
      $\tm = \tm_1\,\tm_2 \tod{d} \tmtwo_1\,\tm_2 = \tmtwo$
      with $\tm_1 \tod{d} \tmtwo_1$.
      Note that $\tmtwo$ is an application,
      and hence not a $\symg$-abstraction of degree $D$.
    \item
      If the reduction is internal to $\tm_2$:
      Then the step is of the form
      $\tm = \tm_1\,\tm_2 \tod{d} \tm_1\,\tmtwo_2 = \tmtwo$
      with $\tm_2 \tod{d} \tmtwo_2$.
      Note that $\tmtwo$ is an application,
      and hence not a $\symg$-abstraction of degree $D$.
    \end{enumerate}
  \item
    $\tm = \bin{\tm_1}{\tm_2}$:
    We consider two subcases, depending on whether the reduction 
    is internal to $\tm_1$ or internal to $\tm_2$:
    \begin{enumerate}
    \item
      If the reduction is internal to $\tm_1$:
      Then the step is of the form
      $\tm = \bin{\tm_1}{\tm_2} \tod{d} \bin{\tmtwo_1}{\tm_2} = \tmtwo$
      with $\tm_1 \tod{d} \tmtwo_1$.
      By hypothesis $\tm$ is not a $\symg$-abstraction of degree $D$,
      so $\tm_1$ is also not a $\symg$-abstraction of degree $D$.
      By \ih $\tmtwo_1$ is not a $\symg$-abstraction of degree $D$,
      so we conclude that
      $\tmtwo = \bin{\tmtwo_1}{\tm_2}$
      is not a $\symg$-abstraction of degree $D$.
    \item
      If the reduction is internal to $\tm_2$:
      Then the step is of the form
      $\tm = \bin{\tm_1}{\tm_2} \tod{d} \bin{\tm_1}{\tmtwo_2} = \tmtwo$
      with $\tm_2 \tod{d} \tmtwo_2$.
      By hypothesis $\tm$ is not a $\symg$-abstraction of degree $D$,
      so $\tm_1$ is also not a $\symg$-abstraction of degree $D$.
      Hence $\tmtwo = \bin{\tm_1}{\tmtwo_2}$
      is not a $\symg$-abstraction of degree $D$.
    \end{enumerate}
  \end{enumerate}
\item
  By induction on $\tm$:
  \begin{enumerate}
  \item
    $\tm = \var$:
    This case is impossible,
    as there are no reduction steps $\tm \tod{d} \tmtwo$.
  \item
    $\tm = \lam{\var}{\tm'}$:
    Straightforward resorting to the \ih.
  \item
    $\tm = \tm_1\,\tm_2$:
    We consider three subcases, depending on whether the reduction is
    at the root, internal to $\tm_1$, internal to $\tm_2$:
    \begin{enumerate}
    \item
      If the reduction is at the root:
      Then $\tm_1 = (\lam{\var}{\tm'_1})\sctx$ is an abstraction of degree $d$,
      and the step is of the form
      $\tm
       = (\lam{\var}{\tm'_1})\sctx\,\tm_2
       \tod{d}
         \tm'_1\sub{\var}{\tm_2}\garb{\tm_2}\sctx
       = \tmtwo$.
      Note that by hypothesis, $\tm = (\lam{\var}{\tm'_1})\sctx\,\tm_2$
      is in $\tod{D}$-normal form,
      which means in particular that $\tm'_1$, $\sctx$ and $\tm_2$
      are in $\tod{D}$-normal form.
      Moreover, since $\lam{\var}{\tm'_1}$ is an abstraction of degree $d$,
      its type is of the form $\typ\to\typtwo$ with
      $\height{\typ\to\typtwo} = d$.
      Moreover, its argument $\tm_2$ is such that
      $\typeof{\tm_2} = \typ$,
      so $\height{\typeof{\tm_2}} = \height{\typ} < d \leq D$.
      In particular, $\tm_2$ cannot be an abstraction of degree $D$.
      By~\rlem{substitution_of_tod_normal_forms}
      this implies that $\tm'_1\sub{\var}{\tm_2}$ is in $\tod{D}$-normal form.
      Finally, this means that $\tm'_1\sub{\var}{\tm_2}\garb{\tm_2}\sctx$
      must also be in $\tod{D}$-normal form.
    \item
      If the reduction is internal to $\tm_1$:
      Then the step is of the form
      $\tm = \tm_1\,\tm_2 \tod{d} \tmtwo_1\,\tm_2 = \tmtwo$.
      By hypothesis $\tm = \tm_1\,\tm_2$ is in $\tod{D}$-normal form,
      so we kwow that $\tm_1$ and $\tm_2$ must be in $\tod{D}$-normal form
      and, moreover, that $\tm_1$ is not a $\symg$-abstraction of degree $D$. 
      By \ih, we have that $\tmtwo_1$ is a $\tod{D}$-normal form.
      Moreover, by item~\ref{reduction_does_not_create_higher_degree_redexes:abstractions}
      of this lemma, we have that $\tmtwo_1$ is not a $\symg$-abstraction
      of degree $D$.
      Hence we conclude that $\tmtwo = \tmtwo_1\,\tm_2$
      is in $\tod{D}$-normal form.
    \item
      If the reduction is internal to $\tm_2$:
      Then the step is of the form
      $\tm = \tm_1\,\tm_2 \tod{d} \tm_1\,\tmtwo_2 = \tmtwo$.
      By hypothesis $\tm = \tm_1\,\tm_2$ is in $\tod{D}$-normal form,
      so we kwow that $\tm_1$ and $\tm_2$ must be in $\tod{D}$-normal form
      and, moreover, that $\tm_1$ is not a $\symg$-abstraction of degree $D$. 
      By \ih, we have that $\tmtwo_2$ is a $\tod{D}$-normal form.
      Hence we conclude that $\tmtwo = \tmtwo_1\,\tm_2$
      is in $\tod{D}$-normal form.
    \end{enumerate}
  \item
    $\tm = \bin{\tm_1}{\tm_2}$:
    Straightforward resorting to the \ih.
  \end{enumerate}
\end{enumerate}
\end{proof}

\begin{proposition}[Lifting property for lower steps]
\lprop{appendix:retraction_of_higher_degree_steps}
Let $d < D$ and suppose that $\tm \tod{d} \tmtwo \rtod{D} \tmtwo'$.
Then there exist a term $\tm'$ and a term $\tmtwo''$
such that $\tm \rtod{D} \tm'$
and $\tmtwo' \rtod{D} \tmtwo''$
and $\tm' \todplus{d} \tmtwo''$ in at least one step.
Graphically:
\[
  \xymatrix@C=.5cm@R=.5cm{
    \tm \ar^-{d}[r]
        \ar@{.>>}_{D}[dd]
  &
    \tmtwo \ar@{->>}^{D}[d]
  \\
  &
    \tmtwo' \ar@{.>>}^{D}[d]
  \\
    \tm' \ar@{.>>}^{d}[r]
  &
    \tmtwo''
  }
\]
\end{proposition}
\begin{proof}
Take $\tm' := \simp{D}{\tm}$.
By the fact that a term reduces
to its simplification~(\rlem{tm_reduces_to_simpk_by_degrees})
we have that $\tm \rtod{D} \simp{D}{\tm}$.
Appyling commutation~(\rprop{commutation_by_degrees})
on the reduction sequences
$\tm \tod{d} \tmtwo$
and
$\tm \rtod{d} \simp{D}{\tm}$,
we have that
there exists a term $\tmthree$ such that
$\tmtwo \rtod{D} \tmthree$
and
$\simp{D}{\tm} \todplus{d} \tmthree$ in at least one step.
Applying the commutation theorem again,
this time on the reduction sequences
$\tmtwo \rtod{D} \tmthree$
and
$\tmtwo \rtod{D} \tmtwo'$
we have that there exists a term $\tmtwo''$
such that $\tmthree \rtod{D} \tmtwo''$
and $\tmtwo' \rtod{D} \tmtwo''$.
The situation is:
\[
  \xymatrix@C=.5cm@R=.5cm{
    \tm \ar^-{d}[rr]
        \ar@{->>}_{D}[dd]
  &
  &
    \tmtwo \ar@{->>}^{D}[d]
           \ar@{->>}_{D}[ldd]
  \\
  &&
    \tmtwo' \ar@{->>}^{D}[d]
  \\
    \simp{D}{\tm} \ar@{->>}^{d}[r]
  &
    \tmthree \ar@{->>}^{D}[r]
  &
    \tmtwo''
  }
\]
By~\rlem{simpk_is_normal_by_degrees}
we know that $\simp{D}{\tm}$ is in $\tod{D}$-normal form,
and since $\simp{D}{\tm} \rtod{d} \tmthree$ with $d < D$,
by~\rlem{reduction_does_not_create_higher_degree_redexes}~(\ref{reduction_does_not_create_higher_degree_redexes:preservation})
we have that $\tmthree$ is in $\tod{D}$-normal form,
so $\tmthree = \tmtwo''$, which concludes the proof.
\end{proof}

\begin{lemma}[Local postponement of forgetful reduction]
\llem{tog_shone_reverse_commutation}
If $\redex : \tm \shone \tmtwo$ is a forgetful step and
$\redextwo : \tmtwo \tod{d} \tmtwo'$ is a reduction step of degree $d$,
there exists a term $\tm'$,
a forgetful reduction $\protract{\redex}{\redextwo} : \tm' \shone^* \tmtwo'$
and a step $\retract{\redextwo}{\redex} : \tm \tod{d} \tm'$.
Graphically:
\[
  \xymatrix@C=.1cm@R=.5cm{
    \tm \ar_{d}@{.>}[d] & \shone & \tmtwo \ar^{d}[d] \\
    \tm' & \shone^* & \tmtwo' \\
  }
\]
Furthermore,
the step $\retract{\redextwo}{\redex}$ determines the step $\redextwo$.
More precisely,
if $\retract{\redextwo}{\redex} = \retract{\redexthree}{\redex}$
then $\redextwo = \redexthree$.
\end{lemma}
\begin{proof}
By induction on $\tm$:
\begin{enumerate}
\item
  $\tm = \var$:
  Impossible, as there are no reduction steps $\var \shone \tmtwo$.
\item
  $\tm = \lam{\var}{\tm_1}$:
  The steps must be of the form
  $\redex : \tm = \lam{\var}{\tm_1} \shone \lam{\var}{\tmtwo_1} = \tmtwo$
  with $\tm_1 \shone \tmtwo_1$,
  and
  $\redextwo : \tmtwo = \lam{\var}{\tmtwo_1} \tod{d} \lam{\var}{\tmtwo'_1} = \tmtwo'$
  with $\tmtwo_1 \tod{d} \tmtwo'_1$.
  By \ih we have the diagram on the left, so we can construct the one
  on the right:
  \[
  \xymatrix@C=.1cm@R=.5cm{
    \tm_1 \ar_{d}@{.>}[d] & \shone & \tmtwo_1 \ar^{d}[d] \\
    \tm'_1 & \shone^* & \tmtwo'_1 \\
  }
  \hspace{1cm}
  \xymatrix@C=.1cm@R=.5cm{
    \lam{\var}{\tm_1} \ar_{d}@{.>}[d] & \shone & \lam{\var}{\tmtwo_1} \ar^{d}[d] \\
    \lam{\var}{\tm'_1} & \shone^* & \lam{\var}{\tmtwo'_1} \\
  }
  \]
  By \ih, the step $\tm_1 \tod{d} \tm'_1$
  determines the step $\tmtwo_1 \tod{d} \tmtwo'_1$,
  which implies that
  the step $\lam{\var}{\tm_1} \tod{d} \lam{\var}{\tm'_1}$
  determines the step $\lam{\var}{\tmtwo_1} \tod{d} \lam{\var}{\tmtwo'_1}$.
\item
  If $\tm = \tm_1\,\tm_2$:
  We consider two subcases, depending on whether
  the step $\redex : \tm \shone \tmtwo$
  is internal to $\tm_1$ or internal to $\tm_2$.
  \begin{enumerate}
  \item
    If $\redex$ is internal to $\tm_1$,
    then $\redex : \tm = \tm_1\,\tm_2 \shone \tmtwo_1\,\tm_2 = \tmtwo$
    where $\tm_1 \shone \tmtwo_1$.
    We consider three further subcases, depending on whether
    the step $\redextwo : \tmtwo = \tmtwo_1\,\tm_2 \tod{d} \tmtwo'$
    is at the root, internal to $\tmtwo_1$, or internal to $\tm_2$:
    \begin{enumerate}
    \item
      If the $\redextwo$ is at the root of $\tmtwo = \tmtwo_1\,\tm_2$:
      Then $\tmtwo_1$
      is a $\symg$-abstraction of degree $d$,
      \ie of the form $\tmtwo_1 = (\lam{\var}{\tmtwo_{11}})\sctx$,
      and $\redextwo$ is of the form
      $\redextwo : (\lam{\var}{\tmtwo_{11}})\sctx\,\tm_2
                   \tod{d}
                   \tmtwo_{11}\sub{\var}{\tm_2}\garb{\tm_2}\sctx$.
      We consider three subcases, depending on the form of the step
      $\redex_1 : \tm_1 \shone (\lam{\var}{\tmtwo_{11}})\sctx$:
      \begin{enumerate}
      \item
        If $\redex_1$ is of the form
        $\tm_1 = (\lam{\var}{\tm_{11}})\sctx
                 \shone
                 (\lam{\var}{\tmtwo_{11}})\sctx
                 = \tmtwo_1$
        where $\tm_{11} \shone \tmtwo_{11}$,
        we can choose
        $\tm' := \tm_{11}\sub{\var}{\tm_2}\garb{\tm_2}\sctx$,
        according to the diagram:
        \[
        \xymatrix@C=.1cm@R=.5cm{
          (\lam{\var}{\tm_{11}})\sctx\,\tm_2
          \ar_{d}@{.>}[d]
        & \shone &
          (\lam{\var}{\tmtwo_{11}})\sctx\,\tm_2
          \ar^{d}[d]
        \\
          \tm_{11}\sub{\var}{\tm_2}\garb{\tm_2}\sctx
        & \shone &
          \tmtwo_{11}\sub{\var}{\tm_2}\garb{\tm_2}\sctx
        }
        \]
        Here we use the fact that
        $\tm_{11} \shone \tmtwo_{11}$
        implies
        $\tm_{11}\sub{\var}{\tm_2} \shone \tmtwo_{11}\sub{\var}{\tm_2}$,
        as stated in \rlem{properties_of_shrinking}.
      \item
        If $\redex_1$
        is of the form
        $\tm_1 = (\lam{\var}{\tmtwo_{11}})\sctx_1\garb{\tm_3}\sctx_2
                 \shone
                 (\lam{\var}{\tmtwo_{11}})\sctx_1\garb{\tmtwo_3}\sctx_2
                 = \tmtwo_1$
        with $\tm_3 \shone \tmtwo_3$,
        we can choose
        $\tm' := \tm_{11}\sub{\var}{\tm_2}\garb{\tm_2}\sctx_1\garb{\tm_3}\sctx_2$,
        according to the diagram:
        \[
        \xymatrix@C=.1cm@R=.5cm{
          (\lam{\var}{\tm_{11}})\sctx_1\garb{\tm_3}\sctx_2\,\tm_2
          \ar_{d}@{.>}[d]
        & \shone &
          (\lam{\var}{\tm_{11}})\sctx_1\garb{\tmtwo_3}\sctx_2\,\tm_2
          \ar^{d}[d]
        \\
          \tm_{11}\sub{\var}{\tm_2}\garb{\tm_2}\sctx_1\garb{\tm_3}\sctx_2
        & \shone &
          \tm_{11}\sub{\var}{\tm_2}\garb{\tm_2}\sctx_1\garb{\tmtwo_3}\sctx_2
        }
        \]
      \item
        If $\redex_1$
        is of the form
        $\tm_1 = (\lam{\var}{\tmtwo_{11}})\sctx_1\garb{\tm_3}\sctx_2
                 \shone
                 (\lam{\var}{\tmtwo_{11}})\sctx_1\sctx_2
                 = \tmtwo_1$,
        we can choose
        $\tm' := \tm_{11}\sub{\var}{\tm_2}\garb{\tm_2}\sctx_1\garb{\tm_3}\sctx_2$,
        according to the diagram:
        \[
        \xymatrix@C=.1cm@R=.5cm{
          (\lam{\var}{\tm_{11}})\sctx_1\garb{\tm_3}\sctx_2\,\tm_2
          \ar_{d}@{.>}[d]
        & \shone &
          (\lam{\var}{\tm_{11}})\sctx_1\sctx_2\,\tm_2
          \ar^{d}[d]
        \\
          \tm_{11}\sub{\var}{\tm_2}\garb{\tm_2}\sctx_1\garb{\tm_3}\sctx_2
        & \shone &
          \tm_{11}\sub{\var}{\tm_2}\garb{\tm_2}\sctx_1\sctx_2
        }
        \]
      \end{enumerate}
    \item
      If $\redextwo$ is internal to $\tmtwo_1$:
      Then $\redextwo$ must be of the form
      $\redextwo : \tmtwo = \tmtwo_1\,\tm_2 \tod{d} \tmtwo'_1\,\tm_2$
      with $\tmtwo_1 \tod{d} \tmtwo'_1$.
      By \ih we have the diagram on the left, so we can construct the one
      on the right:
      \[
      \xymatrix@C=.1cm@R=.5cm{
        \tm_1 \ar_{d}@{.>}[d] & \shone & \tmtwo_1 \ar^{d}[d] \\
        \tm'_1 & \shone^* & \tmtwo'_1 \\
      }
      \hspace{1cm}
      \xymatrix@C=.1cm@R=.5cm{
        \tm_1\,\tm_2 \ar_{d}@{.>}[d] & \shone & \tmtwo_1\,\tm_2 \ar^{d}[d] \\
        \tm'_1\,\tm_2 & \shone^* & \tmtwo'_1\,\tm_2 \\
      }
      \]
    \item
      If $\redextwo$ is internal to $\tm_2$:
      Then $\redextwo$ must be of the form
      $\redextwo : \tmtwo = \tmtwo_1\,\tm_2 \tod{d} \tmtwo_1\,\tm'_2$
      with $\tm_2 \tod{d} \tm'_2$.
      Then we can choose $\tm' := \tm_1\,\tm'_2$,
      according to the diagram:
      \[
      \xymatrix@C=.1cm@R=.5cm{
        \tm_1\,\tm_2
        \ar_{d}@{.>}[d]
      & \shone &
        \tmtwo_1\,\tm_2
        \ar^{d}[d]
      \\
        \tm_1\,\tm'_2
      & \shone &
        \tmtwo_1\,\tm'_2
      }
      \]
    \end{enumerate}
  \item
    If $\redex$ is internal to $\tm_2$,
    then $\redex : \tm = \tm_1\,\tm_2 \shone \tm_1\,\tmtwo_2 = \tmtwo$
    where $\tm_2 \shone \tmtwo_2$.
    We consider three further subcases, depending on whether
    the step $\redextwo$ is at the root,
    internal to $\tm_1$ or internal to $\tmtwo_2$:
    \begin{enumerate}
    \item
      If $\redextwo$ is at the root of $\tm_1\,\tmtwo_2$:
      Then $\tm_1$ is a $\symg$-abstraction of degree $d$,
      \ie of the form $\tm_1 = (\lam{\var}{\tm_{11}})\sctx$,
      and the step is of the form
      $\redextwo : \tmtwo
                   = (\lam{\var}{\tm_{11}})\sctx\,\tmtwo_2
                   \tod{d} \tm_{11}\sub{\var}{\tmtwo_2}\garb{\tmtwo_2}\sctx
                   = \tmtwo'$.
      Then we can choose
      $\tm' := \tm_{11}\sub{\var}{\tm_2}\garb{\tm_2}\sctx$,
      according to the diagram:
      \[
      \xymatrix@C=.1cm@R=.5cm{
        (\lam{\var}{\tm_{11}})\sctx\,\tm_2
        \ar_{d}@{.>}[d]
      & \shone &
        (\lam{\var}{\tm_{11}})\sctx\,\tmtwo_2
        \ar^{d}[d]
      \\
        \tm_{11}\sub{\var}{\tm_2}\garb{\tm_2}\sctx
      & \shone^* &
        \tm_{11}\sub{\var}{\tmtwo_2}\garb{\tmtwo_2}\sctx
      }
      \]
      Here we use the fact that
      $\tm_2 \shone \tmtwo_2$
      implies
      $\tm_{11}\sub{\var}{\tm_2} \shone^* \tm_{11}\sub{\var}{\tmtwo_2}$,
      as stated in \rlem{properties_of_shrinking}.
    \item
      If $\redextwo$ is internal to $\tm_1$:
      Then $\redextwo : \tmtwo = \tm_1\,\tmtwo_2 \tod{d} \tm'_1\,\tmtwo_2 = \tmtwo'$
      with $\tm_1 \tod{d} \tm'_1$,
      and we can choose $\tm' := \tm'_1\,\tm_2$,
      according to the diagram:
      \[
      \xymatrix@C=.1cm@R=.5cm{
        \tm_1\,\tm_2
        \ar_{d}@{.>}[d]
      & \shone &
        \tm_1\,\tmtwo_2
        \ar^{d}[d]
      \\
        \tm'_1\,\tm_2
      & \shone &
        \tm'_1\,\tmtwo_2
      }
      \]
    \item
      If $\redextwo$ is internal to $\tmtwo_2$:
      Then $\redextwo : \tm_1\,\tmtwo_2 \tod{d} \tm_1\,\tmtwo'_2$
      with $\tmtwo_2 \tod{d} \tmtwo'_2$.
      By \ih we have the diagram on the left, so we can construct the one
      on the right:
      \[
      \xymatrix@C=.1cm@R=.5cm{
        \tm_2 \ar_{d}@{.>}[d] & \shone & \tmtwo_2 \ar^{d}[d] \\
        \tm'_2 & \shone^* & \tmtwo'_2 \\
      }
      \hspace{1cm}
      \xymatrix@C=.1cm@R=.5cm{
        \tm_1\,\tm_2 \ar_{d}@{.>}[d] & \shone & \tm_1\,\tmtwo_2 \ar^{d}[d] \\
        \tm_1\,\tm'_2 & \shone^* & \tm_1\,\tmtwo'_2 \\
      }
      \]
    \end{enumerate}
  \end{enumerate}
  Furthermore, to see that the step
  $\retract{\redextwo}{\redex} : \tm \tod{d} \tm'$
  determines the step
  $\redextwo : \tm \tod{d} \tm'$,
  it suffices to note that there are no overlappings between the
  diagrams,
  \ie if the step $\redextwo$ and the step $\retract{\redextwo}{\redex}$
  are fixed, then no more than one of the cases above applies.
\item
  If $\tm = \bin{\tm_1}{\tm_2}$:
  We consider three subcases, depending on whether
  the step $\redex : \tm \shone \tmtwo$ is at the root,
  internal to $\tm_1$ or internal to $\tm_2$:
  \begin{enumerate}
  \item
    If the step $\redex$ is at the root:
    Then $\redex$ is of the form
    $\redex : \tm = \bin{\tm_1}{\tm_2} \shone \tm_1 = \tmtwo$
    and $\redextwo$ is of the form
    $\redextwo : \tmtwo \tod{d} \tmtwo'$.
    Then we can choose $\tm' := \bin{\tmtwo'}{\tm_2}$,
    according to the diagram:
    \[
    \xymatrix@C=.1cm@R=.5cm{
      \tmtwo\,\tm_2
      \ar_{d}@{.>}[d]
    & \shone &
      \tmtwo
      \ar^{d}[d]
    \\
      \bin{\tmtwo'}{\tm_2}
    & \shone &
      \tmtwo'\,\tm_2
    }
    \]
  \item
    If the step $\redex$ is internal to $\tm_1$:
    Then $\redex$ is of the form
    $\redex : \tm = \bin{\tm_1}{\tm_2} \shone \bin{\tmtwo_1}{\tm_2} = \tmtwo$
    with $\tm_1 \shone \tmtwo_1$.
    We consider two subcases, depending on whether the step $\redextwo$
    is internal to $\tmtwo_1$ or internal to $\tm_2$:
    \begin{enumerate}
    \item
      If $\redex$ is internal to $\tmtwo_1$:
      Then $\redex : \tmtwo
                     = \bin{\tmtwo_1}{\tm_2}
                     \tod{d} \bin{\tmtwo'_1}{\tm_2}
                     = \tmtwo'$
      with $\tmtwo_1 \tod{d} \tmtwo'_1$.
      By \ih we have the diagram on the left, so we can construct the one
      on the right:
      \[
      \xymatrix@C=.1cm@R=.5cm{
        \tm_1 \ar_{d}@{.>}[d] & \shone & \tmtwo_1 \ar^{d}[d] \\
        \tm'_1 & \shone^* & \tmtwo'_1 \\
      }
      \hspace{1cm}
      \xymatrix@C=.1cm@R=.5cm{
        \tm_1\garb{\tm_2} \ar_{d}@{.>}[d] & \shone & \tm_1\garb{\tmtwo_2} \ar^{d}[d] \\
        \tm_1\garb{\tm'_2} & \shone^* & \tm_1\garb{\tmtwo'_2} \\
      }
      \]
    \item
      If $\redex$ is internal to $\tm_2$:
      Then $\redex : \tmtwo
                     = \bin{\tmtwo_1}{\tm_2}
                     \tod{d} \bin{\tmtwo_1}{\tm'_2}
                     = \tmtwo'$
      with $\tm_2 \tod{d} \tm'_2$
      and we can choose $\tm' = \tm_1\garb{\tm'_2}$,
      according to the diagram:
      \[
      \xymatrix@C=.1cm@R=.5cm{
        \tm_1\garb{\tm_2}
        \ar_{d}@{.>}[d]
      & \shone &
        \tmtwo_1\garb{\tm_2}
        \ar^{d}[d]
      \\
        \tm_1\garb{\tm'_2}
      & \shone &
        \tmtwo_1\garb{\tm'_2}
      }
      \]
    \end{enumerate}
  \item
    If the step $\redex$ is internal to $\tm_2$:
    Symmetric to the previous case.
  \end{enumerate}
  Furthermore, to see that the step
  $\retract{\redextwo}{\redex} : \tm \tod{d} \tm'$
  determines the step
  $\redextwo : \tm \tod{d} \tm'$,
  it suffices to note that there are no overlappings between the
  diagrams,
  \ie if the step $\redextwo$ and the step $\retract{\redextwo}{\redex}$
  are fixed, then no more than one of the cases above applies.
\end{enumerate}
\end{proof}

\begin{proposition}[Postponement of forgetful reduction]
\lprop{appendix:retraction_before_shrinking}
Let $\redseq : \tm \shone^* \tm'$ be a forgetful reduction sequence and
let $\redseqtwo : \tm' \rtod{d} \tmtwo'$ be a reduction sequence of degree $d$.
Then there exist a term $\tmtwo$
and reduction sequences $\protract{\redseq}{\redseqtwo} : \tmtwo \shone^* \tmtwo'$
and $\retract{\redseqtwo}{\redseq} : \tm \rtod{d} \tmtwo$.
Graphically:
\[
  \xymatrix@C=.2cm@R=.5cm{
    \tm \ar@{.>>}_{d}[d] & \shone^* & \tm' \ar@{->>}^{d}[d] \\
    \tmtwo & \shone^* & \tmtwo' \\
  }
\]
Furthermore, $\retract{\redseqtwo}{\redseq}$ determines $\redseqtwo$,
that is,
More precisely,
$\retract{\redseqtwo_1}{\redseq} = \retract{\redseqtwo_2}{\redseq}$
then $\redseqtwo_1 = \redseqtwo_2$.
\end{proposition}
\begin{proof}
First,
if $\redseq : \tm \shone^* \tm'$ is a forgetful reduction sequence and
and $\redextwo : \tm' \tod{d} \tmtwo'$ is a single step of degree $d$,
we can construct a forgetful reduction sequence $\protract{\redseq}{\redextwo}$
and a step $\retract{\redextwo}{\redseq}$ of degree $d$ 
by induction on $\redseq$
as follows, resorting to~\rlem{tog_shone_reverse_commutation}
for the constructions of $\protract{\redex}{(\retract{\redextwo}{\redseq'})}$
and $\retract{(\retract{\redextwo}{\redseq'})}{\redex}$:
\[
  \begin{array}{rcl}
    \protract{\emptyseq}{\redextwo}
  & \eqdef &
    \emptyseq
  \\
    \protract{(\redex\,\redseq')}{\redextwo}
  & \eqdef &
    (\protract{\redex}{(\retract{\redextwo}{\redseq'})})
    (\protract{\redseq'}{\redextwo})
  \\
  \end{array}
  \HS
  \begin{array}{rcl}
    \retract{\redextwo}{\emptyseq}
  & \eqdef &
    \redextwo
  \\
    \retract{\redextwo}{(\redex\,\redseq')}
  & \eqdef &
    \retract{(\retract{\redextwo}{\redseq'})}{\redex}
  \\
  \end{array}
\]
The inductive cases correspond to the following diagram:
\[
  \xymatrix@C=2cm{
    \ar^{\redex}[r]
    \ar@{->}_{\retract{(\retract{\redextwo}{\redseq'})}{\redex}}[d]
  & 
    \ar@{->>}^{\redseq'}[r]
    \ar@{->}_{\retract{\redextwo}{\redseq'}}[d]
  &
    \ar_{\redextwo}[d]
  \\
    \ar@{->>}_{\protract{\redex}{(\retract{\redextwo}{\redseq'})}}[r]
  &
    \ar@{->>}_{\protract{\redseq'}{\redextwo}}[r]
  &
  }
\]
For the general case, we proceed by induction on $\redseqtwo$,
resorting to the previous construction for the constructions of
$\protract{\redseq}{\redextwo}$ and $\retract{\redextwo}{\redseq}$:
\[
  \begin{array}{rcl}
    \protract{\redseq}{\emptyseq}
  & \eqdef &
    \redseq
  \\
    \protract{\redseq}{(\redextwo\,\redseqtwo')}
  & \eqdef &
    \protract{(\protract{\redseq}{\redextwo})}{\redseqtwo'}
  \end{array}
  \HS
  \begin{array}{rcl}
    \retract{\emptyseq}{\redseq}
  & \eqdef &
    \emptyseq
  \\
    \retract{(\redextwo\,\redseqtwo')}{\redseq}
  & \eqdef &
    (\retract{\redextwo}{\redseq})
    (\retract{\redseqtwo'}{(\protract{\redseq}{\redextwo})})
  \end{array}
\]
The inductive cases correspond to the following diagram:
\[
  \xymatrix@C=2cm{
    \ar@{->}_{\retract{\redextwo}{\redseq}}[d]
    \ar@{->>}^{\redseq}[r]
  &
    \ar^{\redextwo}[d]
  \\
    \ar@{->}_{\retract{\redseqtwo'}{(\protract{\redseq}{\redextwo})}}[d]
    \ar@{->>}^{\protract{\redseq}{\redextwo}}[r]
  &
    \ar@{->>}^{\redseqtwo'}[d]
  \\
    \ar@{->>}^{\protract{(\protract{\redseq}{\redextwo})}{\redseqtwo'}}[r]
  &
  }
\]
Furthermore, to see that $\retract{\redseqtwo}{\redseq}$ determines $\redseqtwo$,
we proceed in three stages:
\begin{enumerate}
\item
  First, if
  $\retract{\redextwo_1}{\redex} = \retract{\redextwo_2}{\redex}$
  then $\redextwo_1 = \redextwo_2$
  by~\rlem{tog_shone_reverse_commutation}.
\item
  Second, by induction on $\redseq$,
  it is easy to see that
  if $\retract{\redextwo_1}{\redseq} = \retract{\redextwo_2}{\redseq}$
  then $\redextwo_1 = \redextwo_2$.
\item
  Finally,
  we can see that
  if $\retract{\redseqtwo_1}{\redseq} = \retract{\redseqtwo_2}{\redseq}$
  then $\redseqtwo_1 = \redseqtwo_2$
  by induction on $\redseqtwo_1$.
  Note that $\retract{\redseqtwo_1}{\redseq} = \retract{\redseqtwo_2}{\redseq}$
  then $\redseqtwo_1$ and $\redseqtwo_2$
  are either both empty or both non-empty.
  The base case is immediate.
  For the induction step,
  we have that $\redseqtwo_1 = \redextwo_1\,\redseqtwo'_1$
  and $\redseqtwo_2 = \redextwo_2\,\redseqtwo'_2$;
  then note that if
  $\retract{(\redextwo_1\,\redseqtwo'_1)}{\redseq} =
   \retract{(\redextwo_2\,\redseqtwo'_2)}{\redseq}$
  then by definition
  $(\retract{\redextwo_1}{\redseq})\,(\retract{{\redseqtwo'_1}}{(\protract{\redseq}{\redextwo})})
  = (\retract{\redextwo_2}{\redseq})\,(\retract{{\redseqtwo'_2}}{(\protract{\redseq}{\redextwo})})$
  so we have that
  $\retract{\redextwo_1}{\redseq} = \retract{\redextwo_2}{\redseq}$
  which, resorting to the previous item,
  means that $\redextwo_1 = \redextwo_2$,
  and
  we also have that
  $\retract{{\redseqtwo'_1}}{(\protract{\redseq}{\redextwo})} =
   \retract{{\redseqtwo'_2}}{(\protract{\redseq}{\redextwo})}$
  which by \ih implies $\redseqtwo'_1 = \redseqtwo'_2$.
\end{enumerate}
\end{proof}

\subsection{Proofs of \rsec{reduction_by_degrees} --- The $\amesym$-measure}
\lsec{appendix:a_measure}

In this section we give detailed proofs of the results
about reduction by degrees stated in \rsec{a_measure}.

\begin{lemma}[Properties of the pointwise multiset order]
\llem{properties_of_mgtmap}
\quad
\begin{enumerate}
\item
  If $\mset_1 \mgtmap \msettwo_1$
  and $\mset_2 \mgtmap \msettwo_2$
  then $\mset_1 + \mset_2 \mgtmap \msettwo_1 + \msettwo_2$.
\item
  If $\mset \mgtmap \msettwo$
  then for all $k \in \Natz$ we have that $\mset \mgeq k \mtimes \msettwo$.
  In particular, taking $k = 1$, $\mset \mgeq \msettwo$.
\item
  If $\mset \mgtmap \msettwo$
  and $\mset$ is non-empty
  then $\mset \mgt \msettwo$.
\end{enumerate}
\end{lemma}
\begin{proof}
The first item is straightforward.
For the second item, 
suppose that $\mset \mgtmap \msettwo$
and proceed by induction on the cardinality of $\mset$.
If $\mset$ is empty, then $\mset = \msempty = \msettwo$,
so $\mset = \msempty \mgeq \msempty = k \mtimes \msempty = k \mtimes \msettwo$.
If $\mset$ is non-empty,
then we can write $\mset = \ms{x} + \mset'$
and $\msettwo = \ms{y} + \msettwo'$
in such a way that $x > y$ and $\mset' \mgtmap \msettwo'$.
By \ih we have that $\mset' \mgeq k \mtimes \msettwo'$,
so
$\mset
 = \ms{x} + \mset'
 \mgt k \mtimes \ms{y} + \mset'
 \mgeq k \mtimes \ms{y} + k \mtimes \msettwo'
 = k \mtimes (\ms{y} + \msettwo')
 = k \mtimes \msettwo$.
The third item is similar to the second.
\end{proof}

\begin{lemma}[Higher substitution lemma]
\llem{appendix:upper_substitution_lemma}
Let $\tm,\tmtwo$ be typable terms and let $\var$ be a variable.
Then $\eme{d}{\tm_0}{\tm} \mleq \eme{d}{\tm_0}{\tm\sub{\var}{\tmtwo}}$.
\end{lemma}
\begin{proof}
We generalize the lemma for the case
in which $\tm$ may also be a memory.
That is, we prove that
if $\anon$ is a term or a memory,
$\tmtwo$ is a term, and $\var$ is a variable
then $\eme{d}{\tm_0}{\anon} \mleq \eme{d}{\tm_0}{\anon\sub{\var}{\tmtwo}}$.
We proceed by induction on $\anon$:
\begin{enumerate}
\item
  $\tm = \var$:
  Then
  $\eme{d}{\tm_0}{\var}
   = \msempty
   \mleq \eme{d}{\tm_0}{\tmtwo}
   = \eme{d}{\tm_0}{\var\sub{\var}{\tmtwo}}$.
\item
  $\tm = \vartwo \neq \var$:
  Then
  $\eme{d}{\tm_0}{\vartwo}
   \mleq \eme{d}{\tm_0}{\vartwo}
   = \eme{d}{\tm_0}{\vartwo\sub{\var}{\tmtwo}}$.
\item
  $\tm = \lam{\vartwo}{\tm'}$:
  By $\alpha$-conversion, we may assume
  that $\vartwo \notin \set{\var} \cup \fv{\tmtwo}$.
  Then
  $\eme{d}{\tm_0}{\lam{\vartwo}{\tm'}}
    = \eme{d}{\tm_0}{\tm'}
    \mleq \eme{d}{\tm_0}{\tm'\sub{\var}{\tmtwo}}
    = \eme{d}{\tm_0}{(\lam{\vartwo}{\tm'})\sub{\var}{\tmtwo}}$
  by \ih.
\item
  If $\tm = \app{(\lam{\var}{\tm_1})\sctx}{\tm_2}$ is a redex of degree $d$:
  Then
  $\eme{d}{\tm_0}{\app{(\lam{\var}{\tm_1})\sctx}{\tm_2}}
   =   \eme{d}{\tm_0}{\tm_1}
     + \eme{d}{\tm_0}{\sctx}
     + \eme{d}{\tm_0}{\tm_2}
     + \ms{(d,\bme{d}{\tm_0})}
   \mleq
       \eme{d}{\tm_0}{\tm_1\sub{\var}{\tmtwo}}
     + \eme{d}{\tm_0}{\sctx\sub{\var}{\tmtwo}}
     + \eme{d}{\tm_0}{\tm_2\sub{\var}{\tmtwo}}
     + \ms{(d,\bme{d}{\tm_0})}
   =  \eme{d}{\tm_0}{((\lam{\vartwo}{\tm_1})\sctx\,\tm_2)\sub{\var}{\tmtwo}}$
  by \ih.
\item
  \label{upper_substitution_lemma:application_non_redex}
  If $\tm = \app{\tm_1}{\tm_2}$ is not a redex of degree $d$:
  Then
  $\eme{d}{\tm_0}{\tm_1\,\tm_2}
   =   \eme{d}{\tm_0}{\tm_1}
     + \eme{d}{\tm_0}{\tm_2}
   \mleq
       \eme{d}{\tm_0}{\tm_1\sub{\var}{\tmtwo}}
     + \eme{d}{\tm_0}{\tm_2\sub{\var}{\tmtwo}}
   = \eme{d}{\tm_0}{(\tm_1\,\tm_2)\sub{\var}{\tmtwo}}$
  by \ih.
\item
  $\tm = \bin{\tm_1}{\tm_2}$:
  Similar to case~\ref{upper_substitution_lemma:application_non_redex}.
\item
  $\sctx = \ctxhole$:
  Then
  $\eme{d}{\tm_0}{\ctxhole}
   \mleq \eme{d}{\tm_0}{\ctxhole}
   = \eme{d}{\tm_0}{\ctxhole\sub{\var}{\tmtwo}}$.
\item
  $\sctx = \bin{\sctx_1}{\tm}$:
  Similar to case~\ref{upper_substitution_lemma:application_non_redex}.
\end{enumerate}
\end{proof}

\begin{proposition}[High/increase]
\lprop{appendix:upper_reduction}
Let $D \in \Natz$. Then the following hold:
\begin{enumerate}
\item
  \label{upper_reduction:bme_increase}
  If $1 \leq d < D$
  and $\tm \tod{D} \tm'$
  then $\bme{d}{\tm} \mleq \bme{d}{\tm'}$.
\item
  \label{upper_reduction:eme_left_increase}
  If $0 \leq d < D$
  and $\tm_0 \tod{D} \tm'_0$
  then $\eme{d}{\tm_0}{\tm} \mleq \eme{d}{\tm'_0}{\tm}$.
\item
  \label{upper_reduction:eme_right_increase}
  If $0 \leq d < D$
  and $\tm_0 \tod{D} \tm'_0$
  and $\tm \tod{D} \tm'$
  then $\eme{d}{\tm_0}{\tm} \mleq \eme{d}{\tm'_0}{\tm'}$.
\item
  \label{upper_reduction:ame_increase}
  If $0 \leq d < D$
  and $\tm \tod{D} \tm'$
  then $\ame{d}{\tm} \mleq \ame{d}{\tm'}$.
\end{enumerate}
\end{proposition}
\begin{proof}
We prove a more general version of the statement:
in items
\ref{upper_reduction:eme_left_increase} and
\ref{upper_reduction:eme_right_increase}
we allow $\tm$ to be either a term or a memory.
For example, the statement of item~\ref{upper_reduction:eme_left_increase}
is generalized as follows:
  if $1 \leq d < D$
  and $\tm_0 \tod{d} \tm'_0$
  then $\eme{d}{\tm_0}{\anon} \mleq \eme{d}{\tm'_0}{\anon}$,
  where $\anon$ is either a term or a memory.

We prove all items simultaneously by induction on $d$.
Note that:
item~1. resorts to the \ih;
item~2. resorts to item~1. (without decreasing $d$);
item~3. resorts to items~1. and~2. (without decreasing $d$);
item~4. resorts to item~3. (without necessarily decreasing $d$).
\begin{enumerate}
\item
  Let $1 \leq d < D$
  and $\tm \tod{D} \tm'$.
  We argue that $\bme{d}{\tm} \mleq \bme{d}{\tm'}$.
  Let $X$ and $Y$ be the sets of reduction sequences
  $X = \set{\redseq \ST (\exists \tmtwo)\ \redseq : \tm \rtod{d} \tmtwo}$
  and
  $Y = \set{\redseqtwo \ST (\exists \tmtwo')\ \redseqtwo : \tm' \rtod{d} \tmtwo'}$.
  Note that, by definition,
  $\bme{d}{\tm} = \msb{\ame{d-1}{\tgt{\redseq}}}{\redseq \in X}$
  and
  $\bme{d}{\tm'} = \msb{\ame{d-1}{\tgt{\redseqtwo}}}{\redseqtwo \in Y}$.
  We construct a function $\varphi : X \to Y$ as follows.
  Consider a reduction step $\redex : \tm \tod{D} \tm'$;
  note that there may be more than one such step, but we know by hypothesis
  that there is at least one.
  By commutation~(\rprop{commutation_by_degrees}),
  given a reduction sequence $\redseq \in X$,
  \ie $\redseq : \tm \rtod{d} \tmtwo$
  there exists a term $\tmtwo'_\redseq$
  and reduction sequences $\redseq/\redex : \tm' \rtod{d} \tmtwo'_\redseq$
  and $\redex/\redseq : \tmtwo \rtod{D} \tmtwo'_\redseq$.
  In particular, $\redseq/\redex \in Y$,
  and we can define $\varphi(\redseq) := \redseq/\redex$.
  Moreover, $\varphi$ is injective,
  because if $\redseq_1,\redseq_2 \in X$
  are such that $\redseq_1/\redex = \redseq_2/\redex$
  then by commutation~(\rprop{commutation_by_degrees})
  we have that $\redseq_1 = \redseq_2$, given that $d < D$.

  First, we claim that
  $\ame{d-1}{\tgt{\redseq}} \mleq \ame{d-1}{\tgt{\varphi(\redseq)}}$
  for every $\redseq \in X$.
  If $d = 1$, this is immediate
  since
  $\ame{d-1}{\tgt{\redseq}}
   = \ame{0}{\tgt{\redseq}}
   = \msempty
   = \ame{0}{\tgt{\varphi(\redseq)}}
   = \ame{d - 1}{\tgt{\varphi(\redseq)}}$.
  Assume now that $d > 1$.
  Then we have that:
  \[
    \begin{array}{rcll}
      \ame{d-1}{\tgt{\redseq}}
    & = &
      \ame{d-1}{\tmtwo}
    \\
    & \mleq &
      \ame{d-1}{\tmtwo'_\redseq}
      & \text{by item~\ref{upper_reduction:ame_increase} of the \ih}
    \\
    & = &
      \ame{d-1}{\tgt{(\redseq/\redex)}}
    \\
    & = &
      \ame{d-1}{\varphi(\redseq)}
    \end{array}
  \]
  To be able to apply item~\ref{upper_reduction:ame_increase} of the \ih,
  observe that $1 \leq d - 1 < D$ holds because $1 \leq d < D$.
  We resort to the \ih as many times as the length of the reduction
  $\tmtwo \rtod{D} \tmtwo'_\redseq$.
  To conclude the proof, let $Z = Y \setminus \varphi(X)$,
  so that $Y = \varphi(X) \uplus Z$, and note that:
  \[
    \begin{array}{rcll}
      \bme{d}{\tm}
    & = &
      \msb{\ame{d-1}{\tgt{\redseq}}}{\redseq \in X}
    \\
    & \mleq &
      \msb{\ame{d-1}{\tgt{\varphi(\redseq)}}}{\redseq \in X}
      \HS\text{($\star$)}
    \\
    & = &
      \msb{\ame{d-1}{\tgt{\redseqtwo}}}{\redseqtwo \in \varphi(X)}
      \HS\text{($\star\star$)}
    \\
    & \mleq &
      \msb{\ame{d-1}{\tgt{\redseqtwo}}}{\redseqtwo \in \varphi(X)}
      +
      \msb{\ame{d-1}{\tgt{\redseqtwo}}}{\redseqtwo \in Z}
    \\
    & = &
      \msb{\ame{d-1}{\tgt{\redseqtwo}}}{\redseqtwo \in \varphi(X) \uplus Z}
    \\
    & = &
      \msb{\ame{d-1}{\tgt{\redseqtwo}}}{\redseqtwo \in Y}
    \\
    & = &
      \bme{d}{\tm'}
    \end{array}
  \]
  To justify the step marked with ($\star$),
  note that
  $\msb{\ame{d-1}{\tgt{\redseq}}}{\redseq \in X}
   = \sum_{\redseq \in X} \ms{\ame{d-1}{\tgt{\redseq}}}
   \mleq \sum_{\redseq \in X} \ms{\ame{d-1}{\tgt{\varphi(\redseq)}}}
   = \msb{\ame{d-1}{\tgt{\varphi(\redseq)}}}{\redseq \in X}$
  because $\ame{d-1}{\tgt{\redseq}} \mleq \ame{d-1}{\tgt{\varphi(\redseq)}}$,
  as we have already claimed.
  To justify the step marked with ($\star\star$),
  note that $\varphi$ is injective, so $X$ and $\varphi(X)$ have the
  same cardinality.
\item
  Let $0 \leq d < D$
  and $\tm_0 \tod{D} \tm'_0$.
  We argue that $\eme{d}{\tm_0}{\anon} \mleq \eme{d}{\tm'_0}{\anon}$,
  where $\anon$ is either a term ($\anon = \tm$)
  or a memory ($\anon = \sctx$).
  We proceed by induction on $\anon$:
  \begin{enumerate}
  \item
    $\tm = \var$:
    Then
    $\eme{d}{\tm_0}{\var}
     = \msempty
     \mleq \msempty
     = \eme{d}{\tm'_0}{\var}$.
  \item
    $\tm = \lam{\var}{\tmtwo}$:
    Then
    $\eme{d}{\tm_0}{\lam{\var}{\tmtwo}}
    = \eme{d}{\tm_0}{\tmtwo}
    \mleq \eme{d}{\tm'_0}{\tmtwo}
    = \eme{d}{\tm'_0}{\lam{\var}{\tmtwo}}$
    by the internal \ih.
  \item
    If $\tm = (\lam{\var}{\tmtwo})\sctx\,\tmthree$ is a redex of degree $d$:
    Then:
    \[
      \begin{array}{rcll}
        \eme{d}{\tm_0}{\tm}
      & = &
        \eme{d}{\tm_0}{(\lam{\var}{\tmtwo})\sctx\,\tmthree}
      \\
      & = &
          \eme{d}{\tm_0}{\tmtwo}
        + \eme{d}{\tm_0}{\sctx}
        + \eme{d}{\tm_0}{\tmthree}
        + \ms{(d,\bme{d}{\tm_0})}
      \\
      & \mleq &
          \eme{d}{\tm'_0}{\tmtwo}
        + \eme{d}{\tm'_0}{\sctx}
        + \eme{d}{\tm'_0}{\tmthree}
        + \ms{(d,\bme{d}{\tm_0})}
        & \text{by the internal \ih}
      \\
        & \mleq &
          \eme{d}{\tm'_0}{\tmtwo}
        + \eme{d}{\tm'_0}{\sctx}
        + \eme{d}{\tm'_0}{\tmthree}
        + \ms{(d,\bme{d}{\tm'_0})}
        & \text{by item~\ref{upper_reduction:bme_increase}}
      \\
        & = &
          \eme{d}{\tm'_0}{(\lam{\var}{\tmtwo})\sctx\,\tmthree}
      \\
        & = &
          \eme{d}{\tm'_0}{\tm}
      \end{array}
    \]
  \item
    \label{upper_reduction:eme_left_increase:application_non_redex}
    If $\tm = \tmtwo\,\tmthree$ is not a redex of degree $d$:
    Then:
    \[
      \begin{array}{rcll}
        \eme{d}{\tm_0}{\tm}
      & = &
        \eme{d}{\tm_0}{\tmtwo\,\tmthree}
      \\
      & = &
          \eme{d}{\tm_0}{\tmtwo}
        + \eme{d}{\tm_0}{\tmthree}
      \\
        & \mleq &
          \eme{d}{\tm'_0}{\tmtwo}
        + \eme{d}{\tm'_0}{\tmthree}
        & \text{by the internal \ih}
      \\
        & = &
          \eme{d}{\tm'_0}{\tmtwo\,\tmthree}
      \end{array}
    \]
  \item
    $\tm = \bin{\tmtwo}{\tmthree}$:
    Similar to case~\ref{upper_reduction:eme_left_increase:application_non_redex}.
  \item
    $\sctx = \ctxhole$:
    Then
    $\eme{d}{\tm_0}{\ctxhole}
     = \msempty \mleq \msempty
     = \eme{d}{\tm'_0}{\ctxhole}$.
  \item
    $\sctx = \bin{\sctx_1}{\tm}$:
    Similar to case~\ref{upper_reduction:eme_left_increase:application_non_redex}.
  \end{enumerate}
\item
  Let $0 \leq d < D$
  and $\tm_0 \tod{D} \tm'_0$
  and let $\anon,\anon'$ be 
  and $\anon \tod{D} \anon'$
  where $\anon,\anon'$
  are either terms ($\anon = \tm$ and $\anon = \tm'$)
  or memories ($\anon = \sctx$ and $\anon = \sctx'$).
  We argue that $\eme{d}{\tm_0}{\anon} \mleq \eme{d}{\tm'_0}{\anon'}$.
  We proceed by induction on $\anon$:
  \begin{enumerate}
  \item
    $\tm = \var$:
    Impossible, as there are no steps $\var \tod{D} \tm'$.
  \item
    $\tm = \lam{\var}{\tmtwo}$:
    Then the step is of the form
    $\tm
     = \lam{\var}{\tmtwo}
     \tod{D} \lam{\var}{\tmtwo'}
     = \tm'$
    with $\tmtwo \tod{D} \tmtwo'$,
    so
    $\eme{d}{\tm_0}{\lam{\var}{\tmtwo}}
    = \eme{d}{\tm_0}{\tmtwo}
    \mleq \eme{d}{\tm'_0}{\tmtwo'}
    = \eme{d}{\tm'_0}{\lam{\var}{\tmtwo'}}$
    by the internal \ih.
  \item
    If $\tm = \app{(\lam{\var}{\tmtwo})\sctx}{\tmthree}$
    is the redex of degree $D$ contracted by the step $\tm \tod{D} \tm'$:
    Then the step is of the form
    $\tm
     = \app{(\lam{\var}{\tmtwo})\sctx}{\tmthree}
     \tod{D} \tmtwo\sub{\var}{\tmthree}\garb{\tmthree}\sctx
     = \tm'$.
    Note that $\tm$ is not a redex of degree $d$ because $d < D$,
    so:
    \[
      \begin{array}{rcll}
        \eme{d}{\tm_0}{\tm}
      & = &
        \eme{d}{\tm_0}{\app{(\lam{\var}{\tmtwo})\sctx}{\tmthree}}
      \\
      & = &
          \eme{d}{\tm_0}{\tmtwo}
        + \eme{d}{\tm_0}{\sctx}
        + \eme{d}{\tm_0}{\tmthree}
      \\
      & \mleq &
          \eme{d}{\tm'_0}{\tmtwo}
        + \eme{d}{\tm'_0}{\sctx}
        + \eme{d}{\tm'_0}{\tmthree}
        & \text{by item~\ref{upper_reduction:eme_left_increase}}
      \\
      & = &
          \eme{d}{\tm'_0}{\tmtwo}
        + \eme{d}{\tm'_0}{\tmthree}
        + \eme{d}{\tm'_0}{\sctx}
      \\
      & \mleq &
          \eme{d}{\tm'_0}{\tmtwo\sub{\var}{\tmthree}}
        + \eme{d}{\tm'_0}{\tmthree}
        + \eme{d}{\tm'_0}{\sctx}
        & \text{by~\rlem{upper_substitution_lemma}}
      \\
      & = &
        \eme{d}{\tm'_0}{\tmtwo\sub{\var}{\tmthree}\garb{\tmthree}\sctx}
      \\
      & = &
        \eme{d}{\tm'_0}{\tm'}
      \end{array}
    \]
  \item
    If $\tm = \app{(\lam{\var}{\tmtwo})\sctx}{\tmthree}$
    is a redex of degree $d$:
    Note that $\tm$ is not a redex of degree $D$ because $d < D$.
    We consider three subcases, depending on whether the step
    $\tm \tod{D} \tm'$ is internal to $\tmtwo$,
    internal to $\sctx$, or internal to $\tmthree$.
    All these subcases are similar;
    we only give the proof for the case in which
    the step is internal to $\tmtwo$.
    Then:
    \[
      \begin{array}{rcll}
        \eme{d}{\tm_0}{\tm}
      & = &
        \eme{d}{\tm_0}{\app{(\lam{\var}{\tmtwo})\sctx}{\tmthree}}
      \\
      & = &
          \eme{d}{\tm_0}{\tmtwo}
        + \eme{d}{\tm_0}{\sctx}
        + \eme{d}{\tm_0}{\tmthree}
        + \ms{(d,\bme{d}{\tm_0})}
      \\
      & \mleq &
          \eme{d}{\tm'_0}{\tmtwo'}
        + \eme{d}{\tm_0}{\sctx}
        + \eme{d}{\tm_0}{\tmthree}
        + \ms{(d,\bme{d}{\tm_0})}
        & \text{by the internal \ih}
      \\
      & \mleq &
          \eme{d}{\tm'_0}{\tmtwo'}
        + \eme{d}{\tm'_0}{\sctx}
        + \eme{d}{\tm'_0}{\tmthree}
        + \ms{(d,\bme{d}{\tm_0})}
        & \text{by item~\ref{upper_reduction:eme_left_increase}}
      \\
      & \mleq &
          \eme{d}{\tm'_0}{\tmtwo'}
        + \eme{d}{\tm'_0}{\sctx}
        + \eme{d}{\tm'_0}{\tmthree}
        + \ms{(d,\bme{d}{\tm'_0})}
        & \text{by item~\ref{upper_reduction:bme_increase}}
      \\
      & = &
        \eme{d}{\tm'_0}{\app{(\lam{\var}{\tmtwo'})\sctx}{\tmthree}}
      \\
      & = &
        \eme{d}{\tm'_0}{\tm'}
      \end{array}
    \]
  \item
    \label{upper_reduction:eme_right_increase:application_non_redex}
    If $\tm = \app{\tmtwo}{\tmthree}$
    is not the redex contracted by the step $\tm \tod{D} \tm'$
    nor a redex of degree $d$:
    We consider two subcases, depending on whether the step
    $\tm \tod{D} \tm'$ is internal to $\tmtwo$ or internal to $\tmthree$:
    \begin{enumerate}
    \item
      If the step is internal to $\tmtwo$,
      then the step is of the form
      $\tm
       = \tmtwo\,\tmthree
       \tod{D} \tmtwo'\,\tmthree
       = \tm'$
      with $\tmtwo \tod{D} \tmtwo'$.
      We know that $\tmtwo$ is not a $\symg$-abstraction of degree $d$,
      but $d < D$, so it may be the case that $\tmtwo'$
      is a $\symg$-abstraction of degree $d$,
      \ie reduction at degree $D$ may create an abstraction of degree $d < D$.
      We consider two further subcases, depending on whether
      $\tmtwo'$ is a $\symg$-abstraction of degree $d$ or not:
      \begin{enumerate}
      \item
        If $\tmtwo' = (\lam{\var}{\tmtwo''})\sctx$ 
        is a $\symg$-abstraction of degree $d$, then:
        \[
          \begin{array}{rcll}
            \eme{d}{\tm_0}{\tm}
          & = &
            \eme{d}{\tm_0}{\tmtwo\,\tmthree}
          \\
          & = &
              \eme{d}{\tm_0}{\tmtwo}
            + \eme{d}{\tm_0}{\tmthree}
          \\
          & \mleq &
              \eme{d}{\tm'_0}{(\lam{\var}{\tmtwo''})\sctx}
            + \eme{d}{\tm_0}{\tmthree}
            & \text{by the internal \ih}
          \\
          & \mleq &
              \eme{d}{\tm'_0}{(\lam{\var}{\tmtwo''})\sctx}
            + \eme{d}{\tm'_0}{\tmthree}
            & \text{by item~\ref{upper_reduction:eme_left_increase}}
          \\
          & \mleq &
              \eme{d}{\tm'_0}{(\lam{\var}{\tmtwo''})\sctx}
            + \eme{d}{\tm'_0}{\tmthree}
            + \ms{(d,\bme{d}{\tm'_0})}
          \\
          & = &
              \eme{d}{\tm'_0}{(\lam{\var}{\tmtwo''})\sctx\,\tmthree}
          \\
          & = &
              \eme{d}{\tm'_0}{\tm'}
          \end{array}
        \]
      \item
        \label{upper_reduction:eme_right_increase:application_non_redex:non_creation}
        If $\tmtwo'$ is not a $\symg$-abstraction of degree $d$, then:
        \[
          \begin{array}{rcll}
            \eme{d}{\tm_0}{\tm}
          & = &
            \eme{d}{\tm_0}{\tmtwo\,\tmthree}
          \\
          & = &
              \eme{d}{\tm_0}{\tmtwo}
            + \eme{d}{\tm_0}{\tmthree}
          \\
          & \mleq &
              \eme{d}{\tm'_0}{\tmtwo'}
            + \eme{d}{\tm_0}{\tmthree}
            & \text{by the internal \ih}
          \\
          & \mleq &
              \eme{d}{\tm'_0}{\tmtwo'}
            + \eme{d}{\tm'_0}{\tmthree}
            & \text{by item~\ref{upper_reduction:eme_left_increase}}
          \\
          & = &
              \eme{d}{\tm'_0}{\tmtwo'\,\tmthree}
          \\
          & = &
              \eme{d}{\tm'_0}{\tm'}
          \end{array}
        \]
      \end{enumerate}
    \item
      If the step is internal to $\tmthree$:
      Similar to case~\ref{upper_reduction:eme_right_increase:application_non_redex:non_creation}.
    \end{enumerate}
  \item
    $\tm = \bin{\tmtwo}{\tmthree}$:
    Similar to case~\ref{upper_reduction:eme_right_increase:application_non_redex}.
  \item
    $\sctx = \ctxhole$:
    Impossible, as there are no reduction steps $\ctxhole \tod{D} \sctx'$.
  \item
    $\sctx = \bin{\sctx_1}{\tm}$:
    Similar to case~\ref{upper_reduction:eme_right_increase:application_non_redex}.
  \end{enumerate}
\item
  Let $0 \leq d < D$
  and $\tm \tod{D} \tm'$.
  We argue that $\ame{d}{\tm} \mleq \ame{d}{\tm'}$.
  Indeed:
  \[
    \begin{array}{rcll}
      \ame{d}{\tm}
    & = &
      \sum_{i=1}^{d} \eme{i}{\tm}{\tm}
    \\
    & \mleq &
      \sum_{i=1}^{d} \eme{i}{\tm'}{\tm'}
      & \text{by item~\ref{upper_reduction:eme_right_increase},
              resorting to the \ih when $i < d$}
    \\
    & = &
      \ame{d}{\tm'}
    \end{array}
  \]
  Note that for the value $i = d$, we resort directly to
  item~\ref{upper_reduction:eme_right_increase} and not to the \ih.
\end{enumerate}
\end{proof}

\begin{lemma}[Substitution of degree $d$ does not create abstractions]
\llem{substitution_of_degree_d_does_not_create_abstractions}
If $\tm$ and $\tmtwo$ are not $\symg$-abstractions of degree $d$,
then $\tm\sub{\var}{\tmtwo}$ is not a $\symg$-abstraction of degree $d$.
\end{lemma}
\begin{proof}
By induction on $\tm$:
\begin{enumerate}
\item
  $\tm = \var$:
  Then $\tm\sub{\var}{\tmtwo} = \tmtwo$ is not a $\symg$-abstraction
  of degree $d$.
\item
  $\tm = \vartwo \neq \var$:
  Then $\tm\sub{\var}{\tmtwo} = \vartwo$ is not a $\symg$-abstraction
  of degree $d$.
\item
  $\tm = \lam{\vartwo}{\tm'}$:
  Then $\tm$ is a $\symg$-abstraction
  but, by hypothesis, we know that it cannot be of degree $d$.
  Hence
  $\tm\sub{\var}{\tmtwo} 
   = \lam{\vartwo}{\tm'\sub{\var}{\tmtwo}}$.
  By the substitution lemma~(\rlem{typing_substitution_lemma})
  $\tm$ and $\tm\sub{\var}{\tmtwo}$ have the same type,
  so $\lam{\vartwo}{\tm'\sub{\var}{\tmtwo}}$
  is a $\symg$-abstraction but it is not of degree $d$.
\item
  $\tm = \tm_1\,\tm_2$:
  Then
  $\tm\sub{\var}{\tmtwo} 
  = \tm_1\sub{\var}{\tmtwo}\,\tm_2\sub{\var}{\tmtwo}$
  is an application, hence not a $\symg$-abstraction of degree $d$.
\item
  $\tm = \bin{\tm_1}{\tm_2}$:
  Since $\tm$ is not a $\symg$-abstraction of degree $d$,
  we have that $\tm_1$ is also not a $\symg$-abstraction of degree $d$.
  By \ih, $\tm_1\sub{\var}{\tmtwo}$
  is not a $\symg$-abstraction of degree $d$.
  So
  $\tm\sub{\var}{\tmtwo} = \bin{\tm_1\sub{\var}{\tmtwo}}{\tm_2\sub{\var}{\tmtwo}}$
  is not a $\symg$-abstraction of degree $d$.
\end{enumerate}
\end{proof}

\begin{lemma}[Lower substitution lemma]
\llem{appendix:lower_substitution_lemma}
Let $\tm,\tmtwo$ be typable terms and let $\var$ be a variable
such that $\tmtwo$ is not a $\symg$-abstraction of degree $d$.
Then there exists $k \in \Natz$ such that
$\eme{d}{\tm_0}{\tm\sub{\var}{\tmtwo}}
 = \eme{d}{\tm_0}{\tm} + k \mtimes \eme{d}{\tm_0}{\tmtwo}$.
\end{lemma}
\begin{proof}
We generalize the lemma for the case
in which $\tm$ may also be a memory.
That is, we prove that
if $\anon$ is a term or a memory,
$\tmtwo$ is a term, and $\var$ is a variable
such that $\tmtwo$ is not a $\symg$-abstraction of degree $d$,
then there exists $k \in \Natz$ such that
$\eme{d}{\tm_0}{\anon\sub{\var}{\tmtwo}}
 = \eme{d}{\tm_0}{\anon} + k \mtimes \eme{d}{\tm_0}{\tmtwo}$.
We proceed by induction on $\anon$:
\begin{enumerate}
\item
  $\tm = \var$:
  Take $k := 1$.
  Then
  $\eme{d}{\tm_0}{\tm\sub{\var}{\tmtwo}}
  = \eme{d}{\tm_0}{\var\sub{\var}{\tmtwo}}
  = \eme{d}{\tm_0}{\tmtwo}
  = \msempty + 1 \mtimes \eme{d}{\tm_0}{\tmtwo}
  = \eme{d}{\tm_0}{\var} + 1 \mtimes \eme{d}{\tm_0}{\tmtwo}
  = \eme{d}{\tm_0}{\tm} + 1 \mtimes \eme{d}{\tm_0}{\tmtwo}$.
\item
  $\tm = \vartwo \neq \var$:
  Take $k := 0$.
  Then
  $\eme{d}{\tm_0}{\tm\sub{\var}{\tmtwo}}
  = \eme{d}{\tm_0}{\vartwo\sub{\var}{\tmtwo}}
  = \eme{d}{\tm_0}{\vartwo}
  = \eme{d}{\tm_0}{\vartwo} + 0 \mtimes \eme{d}{\tm_0}{\tmtwo}
  = \eme{d}{\tm_0}{\tm} + 0 \mtimes \eme{d}{\tm_0}{\tmtwo}$.
\item
  $\tm = \lam{\vartwo}{\tm'}$:
  By $\alpha$-conversion we may assume
  that $\vartwo \notin \set{\var} \cup \fv{\tmtwo}$.
  Resorting to the \ih,
  we have
  $\eme{d}{\tm_0}{\tm\sub{\var}{\tmtwo}}
  = \eme{d}{\tm_0}{\lam{\vartwo}{\tm'\sub{\var}{\tmtwo}}}
  = \eme{d}{\tm_0}{\tm'\sub{\var}{\tmtwo}}
  = \eme{d}{\tm_0}{\tm'} + k \mtimes \eme{d}{\tm_0}{\tmtwo}
  = \eme{d}{\tm_0}{\lam{\vartwo}{\tm'}} + k \mtimes \eme{d}{\tm_0}{\tmtwo}
  = \eme{d}{\tm_0}{\tm} + k \mtimes \eme{d}{\tm_0}{\tmtwo}$.
\item
  $\tm = (\lam{\var}{\tm_1})\sctx\,\tm_2$
  where $(\lam{\var}{\tm_1})\sctx$ is a $\symg$-abstraction of degree $d$.
  Note that, by the substitution lemma~(\rlem{typing_substitution_lemma})
  we have that
  $(\lam{\var}{\tm_1\sub{\var}{\tmtwo}})(\sctx\sub{\var}{\tmtwo})$
  is also an abstraction of degree $d$.
  Then by \ih there exist $k_1, k_2, k_3 \in \Natz$
  such that:
  \[
    \begin{array}{rcl}
    &&
      \eme{d}{\tm_0}{\tm\sub{\var}{\tmtwo}}
    \\
    & = &
      \eme{d}{\tm_0}{
        (\lam{\var}{\tm_1\sub{\var}{\tmtwo}})
        (\sctx\sub{\var}{\tmtwo})
        (\tm_2\sub{\var}{\tmtwo})
      }
    \\
    & = &
        \eme{d}{\tm_0}{\lam{\var}{\tm_1\sub{\var}{\tmtwo}}}
      + \eme{d}{\tm_0}{\sctx\sub{\var}{\tmtwo}}
      + \eme{d}{\tm_0}{\tm_2\sub{\var}{\tmtwo}}
      + \ms{(d,\bme{d}{\tm_0})}
    \\
    & = &
        (\eme{d}{\tm_0}{\lam{\var}{\tm_1}} + k_1 \mtimes \eme{d}{\tm_0}{\tmtwo})
      +  (\eme{d}{\tm_0}{\sctx} + k_2 \mtimes \eme{d}{\tm_0}{\tmtwo})
      \\&&
      +\ (\eme{d}{\tm_0}{\tm_2} + k_3 \mtimes \eme{d}{\tm_0}{\tmtwo})
      +  \ms{(d,\bme{d}{\tm_0})}
    \\
    & = &
        \eme{d}{\tm_0}{\lam{\var}{\tm_1}}
      + \eme{d}{\tm_0}{\sctx}
      + \eme{d}{\tm_0}{\tm_2}
      +  \ms{(d,\bme{d}{\tm_0})}
      + (k_1 + k_2 + k_3) \mtimes \eme{d}{\tm_0}{\tmtwo}
    \\
    & = &
        \eme{d}{\tm_0}{(\lam{\var}{\tm_1})\sctx\,\tm_2}
      + (k_1 + k_2 + k_3) \mtimes \eme{d}{\tm_0}{\tmtwo}
    \\
    & = &
        \eme{d}{\tm_0}{\tm}
      + (k_1 + k_2 + k_3) \mtimes \eme{d}{\tm_0}{\tmtwo}
    \end{array}
  \]
  So taking $k := k_1 + k_2 + k_3$ we are done.
\item
  \label{lower_substitution_lemma:application_non_redex}
  $\tm = \tm_1\,\tm_2$
  where $\tm_1$ is not a $\symg$-abstraction of degree $d$:
  Note by \rlem{substitution_of_degree_d_does_not_create_abstractions}
  that $\tm_1\sub{\var}{\tm_2}$
  is not a $\symg$-abstraction of degree $d$.
  Then by \ih there exist $k_1, k_2 \in \Natz$ such that:
  \[
    \begin{array}{rcl}
      \eme{d}{\tm_0}{\tm\sub{\var}{\tmtwo}}
    & = &
      \eme{d}{\tm_0}{\tm_1\sub{\var}{\tmtwo}\,\tm_2\sub{\var}{\tmtwo}}
    \\
    &&
        \eme{d}{\tm_0}{\tm_1\sub{\var}{\tmtwo}}
      + \eme{d}{\tm_0}{\tm_2\sub{\var}{\tmtwo}}
    \\
    & = &
        (\eme{d}{\tm_0}{\tm_1} + k_1 \mtimes \eme{d}{\tm_0}{\tmtwo})
      + (\eme{d}{\tm_0}{\tm_2} + k_2 \mtimes \eme{d}{\tm_0}{\tmtwo})
    \\
    & = &
        \eme{d}{\tm_0}{\tm_1} + \eme{d}{\tm_0}{\tm_2}
      + (k_1 + k_2) \mtimes \eme{d}{\tm_0}{\tmtwo})
    \\
    & = &
        \eme{d}{\tm_0}{\tm_1\,\tm_2}
      + (k_1 + k_2) \mtimes \eme{d}{\tm_0}{\tmtwo})
    \\
    \end{array}
  \]
  So taking $k := k_1 + k_2$ we are done.
\item
  $\tm = \bin{\tm_1}{\tm_2}$:
  Similar to the previous case.
\item
  $\sctx = \ctxhole$:
  Take $k := 0$.
  Then
  $\eme{d}{\tm_0}{\ctxhole\sub{\var}{\tmtwo}}
  = \eme{d}{\tm_0}{\ctxhole}
  = \eme{d}{\tm_0}{\ctxhole} + 0 \mtimes \eme{d}{\tm_0}{\tmtwo}$.
\item
  $\sctx = \bin{\sctx'}{\tm}$:
  Similar to
  case~\ref{lower_substitution_lemma:application_non_redex}.
\end{enumerate}
\end{proof}

\begin{proposition}[Low/decrease]
\lprop{appendix:lower_reduction}
Let $D \in \Natz$. Then the following hold:
\begin{enumerate}
\item
  \label{lower_reduction:bme_decrease}
  If $1 \leq d \leq j \leq D$
  and $\tm \tod{d} \tm'$
  then $\bme{j}{\tm} \mgt \bme{j}{\tm'}$.
\item
  \label{lower_reduction:eme_left_decrease}
  If $1 \leq d \leq j \leq D$
  and $\tm_0 \tod{d} \tm'_0$
  then $\eme{j}{\tm_0}{\tm} \mgtmap \eme{j}{\tm'_0}{\tm}$.
\item
  \label{lower_reduction:eme_right_equal_decrease}
  If $1 \leq d \leq D$
  and $\tm_0 \tod{d} \tm'_0$
  and $\tm \tod{d} \tm'$,
  then for all $\mset \in \ameset{d-1}$
  we have
  $\eme{d}{\tm_0}{\tm} \mgt \eme{d}{\tm'_0}{\tm'} + \mset$.
\item
  \label{lower_reduction:eme_right_nonequal_decrease}
  If $1 \leq d < j \leq D$
  and $\tm_0 \tod{d} \tm'_0$
  and $\tm \tod{d} \tm'$
  then $\eme{j}{\tm_0}{\tm} \mgeq \eme{j}{\tm'_0}{\tm'}$.
\item
  \label{lower_reduction:ame_decrease}
  If $1 \leq d \leq D$
  and $\tm \tod{d} \tm'$
  then $\ame{D}{\tm} \mgt \ame{D}{\tm'}$.
\end{enumerate}
\end{proposition}
\begin{proof}
We prove a more general version of the statement:
in items 
\ref{lower_reduction:eme_left_decrease},
\ref{lower_reduction:eme_right_equal_decrease}, and
\ref{lower_reduction:eme_right_nonequal_decrease}
we allow $\tm$ to be either a term or a memory.
For example, the statement of item~\ref{lower_reduction:eme_left_decrease}
is generalized as follows:
  if $1 \leq d \leq j \leq D$
  and $\tm_0 \tod{d} \tm'_0$
  then $\eme{j}{\tm_0}{\anon} \mgtmap \eme{j}{\tm'_0}{\anon}$,
  where $\anon$ is either a term or a memory.

We prove all items simultaneously by induction on $D$.
Note that:
item~1. resorts to the \ih;
item~2. resorts to item~1. (without decreasing $D$);
items~3. and~4. resort to items~1. and~2. (without decreasing $D$);
item~5. resorts to items~3. and~4. (without decreasing $D$).
\begin{enumerate}
\item
  Let $1 \leq d \leq j \leq D$
  and $\tm \tod{d} \tm'$.
  We argue that $\bme{j}{\tm} \mgt \bme{j}{\tm'}$.
  Let $X$ and $Y$
  be the sets of reduction sequences
  $X := \set{\redseq \ST (\exists{\tmtwo})\ \redseq : \tm \rtod{j} \tmtwo}$
  and
  $Y := \set{\redseqtwo \ST (\exists{\tmtwo'})\ \redseqtwo : \tm' \rtod{j} \tmtwo'}$.
  Note that, by definition,
  $\bme{j}{\tm} =
   \msb{\ame{j-1}{\tgt{\redseq}}}{\redseq \in X}$
  and
  $\bme{j}{\tm'} =
   \msb{\ame{j-1}{\tgt{\redseqtwo}}}{\redseqtwo \in Y}$.
  We consider two subcases, depending on whether $d = j$ or $d < j$:
  \begin{enumerate}
  \item
    If $d = j$,
    let $\redex$ be the step $\redex : \tm \tod{d} \tm'$.
    We construct a function $\varphi : Y \to X$
    given by $\varphi(\redseqtwo) = \redex\,\redseqtwo$.
    Observe that $\varphi$ is injective
    and that if $\redseqtwo : \tm' \rtod{d} \tmtwo'$
    then $\redex\,\redseqtwo : \tm \rtod{d} \tmtwo'$,
    and in particular $\redseqtwo$ and $\varphi(\redseqtwo)$
    have the same target.
    Let $Z = X \setminus \varphi(Y)$, so that $X = \varphi(Y) \uplus Z$.
    Note that:
    \[
      \begin{array}{rcll}
        \bme{j}{\tm}
        & = &
        \msb{\ame{j-1}{\tgt{\redseq}}}{\redseq \in X}
      \\
        & = &
        \msb{\ame{j-1}{\tgt{\redseq}}}{\redseq \in \varphi(Y) \uplus Z}
      \\
        & = &
           \msb{\ame{j-1}{\tgt{\redseq}}}{\redseq \in \varphi(Y)}
         + \msb{\ame{j-1}{\tgt{\redseq}}}{\redseq \in Z}
      \\
        & = &
           \msb{\ame{j-1}{\tgt{\redex\redseqtwo}}}{\redseqtwo \in Y}
         + \msb{\ame{j-1}{\tgt{\redseq}}}{\redseq \in Z}
        & \text{since $\varphi$ is injective}
      \\
        & = &
           \msb{\ame{j-1}{\tgt{\redseqtwo}}}{\redseqtwo \in Y}
         + \msb{\ame{j-1}{\tgt{\redseq}}}{\redseq \in Z}
      \\
        & = &
           \bme{j}{\tm'}
         + \msb{\ame{j-1}{\tgt{\redseq}}}{\redseq \in Z}
      \end{array}
    \]
    By this chain of equations,
    in order to conclude that $\bme{j}{\tm} \mgt \bme{j}{\tm'}$,
    it suffices to show that $Z$ is non-empty.
    Indeed, let $\emptyseq : \tm \rtod{d} \tm$ be the empty reduction sequence.
    Then $\emptyseq \in X \setminus \varphi(Y)$, so $\emptyseq \in Z$.
  \item
    If $d < j$,
    we construct a function $\varphi : Y \to X$
    using ~\rprop{retraction_of_higher_degree_steps}.
    More precisely, since $d < j$ and $\tm \tod{d} \tm'$,
    for each reduction sequence $\redseqtwo : \tm' \rtod{j} \tmtwo'$
    by~\rprop{retraction_of_higher_degree_steps}
    there exist a term $\tmtwo_\redseqtwo$ and a term $\tmthree_\redseqtwo$,
    such that there is a reduction sequence
    $\varphi(\redseqtwo) : \tm \rtod{j} \tmtwo_\redseqtwo$
    and such that $\tmtwo' \rtod{j} \tmthree_\redseqtwo$,
    and $\tmtwo_\redseqtwo \todplus{d} \tmthree_\redseqtwo$ in at least one step.

    First, we claim that
    $\ame{j-1}{\tgt{\varphi(\redseqtwo)}} \mgt \ame{j-1}{\tgt{\redseqtwo}}$
    for every $\redseqtwo \in Y$.
    Indeed:
    \[
      \begin{array}{rcll}
        \ame{j-1}{\tgt{\varphi(\redseqtwo)}}
        & = &
        \ame{j-1}{\tmtwo_\redseqtwo}
      \\
        & \mgt &
        \ame{j-1}{\tmthree_\redseqtwo}
        & \text{by item \ref{lower_reduction:ame_decrease} of the \ih}
      \\
        & \mgeq &
        \ame{j-1}{\tmtwo'}
        & \text{by high/increase~(\rprop{upper_reduction}(\ref{upper_reduction:ame_increase}))}
      \\
        & = &
        \ame{j-1}{\tgt{\redseqtwo}}
      \end{array}
    \]
    To be able to apply item \ref{lower_reduction:ame_decrease} of the \ih,
    observe that
    we have that $1 \leq d \leq j - 1 < D$
    holds because we know $d < j \leq D$.
    We resort to the \ih
    as many times as the length of the reduction
    $\tmtwo_\redseqtwo \todplus{d} \tmthree_\redseqtwo$.
    The inequality is strict because this reduction contains at least one step.
    To be able to apply the high/increase property,
    observe that $0 \leq j - 1 < j$.
    We resort to this lemma as many times as the length
    of the reduction $\tmtwo' \rtod{j} \tmthree_\redseqtwo$,
    which may be empty.
    To conclude the proof,
    let $Z = X \setminus \varphi(Y)$, so that $X = \varphi(Y) \uplus Z$,
    and note that:
    \[
      \begin{array}{rcll}
        \bme{j}{\tm}
        & = &
        \msb{\ame{j-1}{\tgt{\redseq}}}{\redseq \in X}
      \\
        & = &
        \msb{\ame{j-1}{\tgt{\redseq}}}{\redseq \in \varphi(Y) \uplus Z}
      \\
        & = &
           \msb{\ame{j-1}{\tgt{\redseq}}}{\redseq \in \varphi(Y)}
         + \msb{\ame{j-1}{\tgt{\redseq}}}{\redseq \in Z}
      \\
        & = &
           \msb{\ame{j-1}{\tgt{\varphi(\redseqtwo)}}}{\redseqtwo \in Y}
         + \msb{\ame{j-1}{\tgt{\redseq}}}{\redseq \in Z}
      \\
        & \mgeq &
           \msb{\ame{j-1}{\tgt{\varphi(\redseqtwo)}}}{\redseqtwo \in Y}
      \\
        & \mgt &
           \msb{\ame{j-1}{\tgt{\redseqtwo}}}{\redseqtwo \in Y}
        & \text{($\star$)}
      \\
        & = &
           \bme{j}{\tm'}
      \end{array}
    \]
    For the step marked with ($\star$),
    note that
    $\msb{\ame{j-1}{\tgt{\varphi(\redseqtwo)}}}{\redseqtwo \in Y}
     \mgtmap
     \msb{\ame{j-1}{\tgt{\redseqtwo}}}{\redseqtwo \in Y}$
    because
    $\ame{j-1}{\tgt{\varphi(\redseqtwo)}} \mgt \ame{j-1}{\tgt{\redseqtwo}}$
    holds by the claim above.
    Moreover, $Y$ is non-empty because the empty reduction sequence
    $\emptyseq : \tm' \rtod{j} \tm'$ is in $Y$,
    so we may resort to \rlem{properties_of_mgtmap}.
  \end{enumerate}
\item
  Let $1 \leq d \leq j \leq D$
  and $\tm_0 \tod{d} \tm'_0$.
  We argue that $\eme{j}{\tm_0}{\anon} \mgtmap \eme{j}{\tm'_0}{\anon}$,
  where $\anon$
  is either a term ($\anon = \tm$)
  or a memory ($\anon = \sctx$).
  We proceed by induction on $\anon$:
  \begin{enumerate}
  \item
    $\tm = \var$:
    Then
    $\eme{j}{\tm_0}{\var}
     = \msempty
     \mgtmap \msempty
     = \eme{j}{\tm'_0}{\var}$.
  \item
    $\tm = \lam{\var}{\tmtwo}$:
    Then
    $\eme{j}{\tm_0}{\lam{\var}{\tmtwo}}
    = \eme{j}{\tm_0}{\tmtwo}
    = \eme{j}{\tm'_0}{\tmtwo}
    = \eme{j}{\tm'_0}{\lam{\var}{\tmtwo}}$
    by the internal \ih.
  \item
    If $\tm = \app{(\lam{\var}{\tmtwo})\sctx}{\tmthree}$
    is a redex of degree $j$:
    Then:
      \[
        \begin{array}{rcll}
        &&
          \eme{j}{\tm_0}{(\lam{\var}{\tmtwo})\sctx\,\tmthree}
        \\
        & = &
            \eme{j}{\tm_0}{\tmtwo}
          + \eme{j}{\tm_0}{\sctx}
          + \eme{j}{\tm_0}{\tmthree}
          + \ms{(j,\bme{j}{\tm_0})}
        \\
        & \mgtmap &
            \eme{j}{\tm'_0}{\tmtwo}
          + \eme{j}{\tm'_0}{\sctx}
          + \eme{j}{\tm'_0}{\tmthree}
          + \ms{(j,\bme{j}{\tm_0})}
          \\&&\HS\text{by the internal \ih}
        \\
        & \mgtmap &
            \eme{j}{\tm'_0}{\tmtwo}
          + \eme{j}{\tm'_0}{\sctx}
          + \eme{j}{\tm'_0}{\tmthree}
          + \ms{(j,\bme{j}{\tm'_0})}
          \\&&\HS\text{
                    since $\bme{j}{\tm_0} \mgt \bme{j}{\tm'_0}$
                    by item \ref{lower_reduction:bme_decrease}}
        \\
        & = &
            \eme{j}{\tm'_0}{(\lam{\var}{\tmtwo})\sctx\,\tmthree}
        \end{array}
      \]
      Note that we can resort to
      item~\ref{lower_reduction:bme_decrease},
      because $1 \leq d \leq j \leq D$
      and $\tm_0 \tod{d} \tm'_0$.
  \item
    If $\tm = \app{\tmtwo}{\tmthree}$ is not a redex of degree $j$:
    Then
    $\eme{j}{\tm_0}{\tmtwo\,\tmthree}
     = \eme{j}{\tm_0}{\tmtwo} + \eme{j}{\tm_0}{\tmthree}
     \mgtmap \eme{j}{\tm'_0}{\tmtwo} + \eme{j}{\tm'_0}{\tmthree}
     = \eme{j}{\tm'_0}{\tmtwo\,\tmthree}$
    by the internal \ih.
  \item
    $\tm = \bin{\tmtwo}{\tmthree}$:
    Then
    $\eme{j}{\tm_0}{\bin{\tmtwo}{\tmthree}}
     = \eme{j}{\tm_0}{\tmtwo} + \eme{j}{\tm_0}{\tmthree}
     \mgtmap \eme{j}{\tm'_0}{\tmtwo} + \eme{j}{\tm'_0}{\tmthree}
     = \eme{j}{\tm'_0}{\bin{\tmtwo}{\tmthree}}$
    by the internal \ih.
  \item
    $\sctx = \ctxhole$:
    Then
    $\eme{j}{\tm_0}{\ctxhole}
    = \msempty
    \mgtmap \msempty
    =
    \eme{j}{\tm'_0}{\ctxhole}$.
  \item
    $\sctx = \bin{\sctx_1}{\tm}$:
    $\eme{j}{\tm_0}{\bin{\sctx_1}{\tm}}
    = \eme{j}{\tm_0}{\sctx_1} + \eme{j}{\tm_0}{\tm}
    \mgtmap \eme{j}{\tm'_0}{\sctx_1} + \eme{j}{\tm'_0}{\tm}$
    by \ih.
  \end{enumerate}
\item
  Let $1 \leq d \leq D$
  and $\tm_0 \tod{d} \tm'_0$
  and $\anon \tod{d} \anon'$,
  where $\anon,\anon'$
  are either terms ($\anon = \tm$ and $\anon = \tm'$)
  or memories ($\anon = \sctx$ and $\anon = \sctx'$).
  We argue that for all $\mset \in \ameset{d-1}$
  we have
  $\eme{d}{\tm_0}{\anon} \mgt \eme{d}{\tm'_0}{\anon'} + \mset$.
  We proceed by induction on $\anon$:
  \begin{enumerate}
  \item
    $\tm = \var$:
    Impossible, as there are no reduction steps $\var \tod{d} \tm'$.
  \item
    $\tm = \lam{\var}{\tmtwo}$:
    Then the step is of the form
    $\tm
     = \lam{\var}{\tmtwo}
     \tod{d} \lam{\var}{\tmtwo'}
     = \tm'$
    with $\tmtwo \tod{d} \tmtwo'$.
    Let $\mset \in \ameset{d-1}$.
    Then
    $\eme{d}{\tm_0}{\tm}
    = \eme{d}{\tm_0}{\lam{\var}{\tmtwo}}
    = \eme{d}{\tm_0}{\tmtwo}
    \mgt \eme{d}{\tm'_0}{\tmtwo'} + \mset
    = \eme{d}{\tm'_0}{\lam{\var}{\tmtwo'}} + \mset
    = \eme{d}{\tm'_0}{\tm'} + \mset$
    by the internal \ih.
  \item
    If $\tm = \app{(\lam{\var}{\tmtwo})\sctx}{\tmthree}$
    is the redex of degree $d$ contracted by the step $\tm \tod{d} \tm'$:
    %%%% This is a very interesting case.
    If the reduction step is at the root,
    then $(\lam{\var}{\tmtwo})\sctx$
    is a $\symg$-abstraction of degree $d$,
    and the step is of the form
    $\tm
     = (\lam{\var}{\tmtwo})\sctx\,\tmthree
     \tod{d} \tmtwo\sub{\var}{\tmthree}\garb{\tmthree}\sctx
     = \tm'$.
    Given that $(\lam{\var}{\tmtwo})\sctx$
    is a $\symg$-abstraction of degree $d$,
    the type of its argument $\tmthree$ is of height strictly less than $d$,
    that is, $\height{\typeof{\tmthree}} < d$.
    In particular, $\tmthree$ is not an abstraction of degree $d$
    so we may apply
    the lower substitution lemma~(\rlem{lower_substitution_lemma})
    which ensures that there exists $k \in \Natz$
    such that
    $\eme{d}{\tm'_0}{\tmtwo\sub{\var}{\tmthree}} =
     \eme{d}{\tm'_0}{\tmtwo} + k \mtimes \eme{d}{\tm'_0}{\tmthree}$.
    Furthermore, observe that
    $\eme{d}{\tm_0}{\tmthree}
     \mgeq
     (1 + k) \mtimes \eme{d}{\tm'_0}{\tmthree}$.
    Indeed,
    by item~\ref{lower_reduction:eme_left_decrease}
    we have that $\eme{d}{\tm_0}{\tmthree} \mgtmap \eme{d}{\tm'_0}{\tmthree}$,
    so by \rlem{properties_of_mgtmap}
    $\eme{d}{\tm_0}{\tmthree}
     \mgeq (1 + k) \mtimes \eme{d}{\tm'_0}{\tmthree}$.
    To conclude this case:
    \[
      \begin{array}{rcll}
      &&
        \eme{d}{\tm_0}{\tm}
      \\
      & = &
        \eme{d}{\tm_0}{(\lam{\var}{\tmtwo})\sctx\,\tmthree}
      \\
      & = &
          \eme{d}{\tm_0}{\tmtwo}
        + \eme{d}{\tm_0}{\sctx}
        + \eme{d}{\tm_0}{\tmthree}
        + \ms{(d,\bme{d}{\tm_0})}
      \\
      & \mgeq &
          \eme{d}{\tm'_0}{\tmtwo}
        + \eme{d}{\tm'_0}{\sctx}
        + \eme{d}{\tm_0}{\tmthree}
        + \ms{(d,\bme{d}{\tm_0})}
        \\&&\HS\text{by item~\ref{lower_reduction:eme_left_decrease}
                     and \rlem{properties_of_mgtmap}}
      \\
      & \mgeq &
          \eme{d}{\tm'_0}{\tmtwo}
        + \eme{d}{\tm'_0}{\sctx}
        + (1 + k) \mtimes \eme{d}{\tm'_0}{\tmthree}
        + \ms{(d,\bme{d}{\tm_0})}
        \\&&\HS\text{since
                     $\eme{d}{\tm_0}{\tmthree}
                      \mgeq (1 + k) \mtimes \eme{d}{\tm'_0}{\tmthree}$}
      \\
      & \mgt &
          \eme{d}{\tm'_0}{\tmtwo}
        + \eme{d}{\tm'_0}{\sctx}
        + (1 + k) \mtimes \eme{d}{\tm'_0}{\tmthree}
        + \mset
        \\&&\HS\text{since $\mset \in \ameset{d-1}$,
                     so $\ms{(d,\bme{d}{\tm_0})} \mgt \mset$}
      \\
      & = &
          \eme{d}{\tm'_0}{\tmtwo}
        + k \mtimes \eme{d}{\tm'_0}{\tmthree}
        + \eme{d}{\tm'_0}{\tmthree}
        + \eme{d}{\tm'_0}{\sctx}
        + \mset
      \\
      & = &
          \eme{d}{\tm'_0}{\tmtwo\sub{\var}{\tmthree}}
        + \eme{d}{\tm'_0}{\tmthree}
        + \eme{d}{\tm'_0}{\sctx}
        + \mset
      \\
      & = &
          \eme{d}{\tm'_0}{\tmtwo\sub{\var}{\tmthree}\garb{\tmthree}\sctx}
        + \mset
      \\
      & = &
          \eme{d}{\tm'_0}{\tm'}
        + \mset
      \end{array}
    \]
  \item
    If $\tm = \app{(\lam{\var}{\tmtwo})\sctx}{\tmthree}$
    is a redex of degree $d$,
    but not the redex contracted by the step $\tm \tod{d} \tm'$:
    There are three subcases, depending on whether the 
    step $\tm \tod{d} \tm'$ is internal to $\tmtwo$, internal to $\sctx$,
    or internal to $\tmthree$.
    All these subcases are similar;
    we only give the proof for the case in which
    the step is internal to $\tmtwo$.
    Then the step is of the form
    $\tm
     = \app{(\lam{\var}{\tmtwo})\sctx}{\tmthree}
     \tod{d} \app{(\lam{\var}{\tmtwo'})\sctx}{\tmthree}
     = \tm'$
    with $\tmtwo \tod{d} \tmtwo'$,
    and we have:
    \[
      \begin{array}{rcll}
      &&
        \eme{d}{\tm_0}{\tm}
      \\
      & = &
        \eme{d}{\tm_0}{\app{(\lam{\var}{\tmtwo})\sctx}{\tmthree}}
      \\
      & = &
          \eme{d}{\tm_0}{\tmtwo}
        + \eme{d}{\tm_0}{\sctx}
        + \eme{d}{\tm_0}{\tmthree}
        + \ms{(d,\bme{d}{\tm_0})}
      \\
      & \mgt &
          \eme{d}{\tm'_0}{\tmtwo'}
        + \eme{d}{\tm_0}{\sctx}
        + \eme{d}{\tm_0}{\tmthree}
        + \ms{(d,\bme{d}{\tm_0})}
        + \mset
        & \text{by the internal \ih}
      \\
      & \mgeq &
          \eme{d}{\tm'_0}{\tmtwo'}
        + \eme{d}{\tm'_0}{\sctx}
        + \eme{d}{\tm'_0}{\tmthree}
        + \ms{(d,\bme{d}{\tm_0})}
        + \mset
        & \text{by item~\ref{lower_reduction:eme_left_decrease}
                and \rlem{properties_of_mgtmap}}
      \\
      & \mgeq &
          \eme{d}{\tm'_0}{\tmtwo'}
        + \eme{d}{\tm'_0}{\sctx}
        + \eme{d}{\tm'_0}{\tmthree}
        + \ms{(d,\bme{d}{\tm'_0})}
        + \mset
        & \text{by item~\ref{lower_reduction:bme_decrease}}
      \\
      & = &
          \eme{d}{\tm'_0}{(\lam{\var}{\tmtwo'})\sctx\,\tmthree}
        + \mset
      \\
      & = &
          \eme{d}{\tm'_0}{\tm'}
        + \mset
      \end{array}
    \]
  \item
    \label{lower_reduction:eme_right_equal_decrease:application_non_redex}
    If $\tm = \app{\tmtwo}{\tmthree}$
    is not a redex of degree $d$:
    There are two subcases, depending on whether the 
    step $\tm \tod{d} \tm'$ is internal to $\tmtwo$ or internal to $\tmthree$:
    \begin{enumerate}
    \item
      If the step is internal to $\tmtwo$,
      then the step is of the form
      $\tm
       = \app{\tmtwo}{\tmthree}
       \tod{d} \app{\tmtwo'}{\tmthree}
       = \tm'$
      with $\tmtwo \tod{d} \tmtwo'$.
      Note that $\tmtwo$ is not a $\symg$-abstraction of degree $d$
      (because $\tmtwo\,\tmthree$ is not a redex of degree $d$).
      Hence by \rlem{reduction_does_not_create_higher_degree_redexes}
      $\tmtwo'$ is not a $\symg$-abstraction of degree $d$.
      Then:
      \[
        \begin{array}{rcll}
          \eme{d}{\tm_0}{\tm}
        & = &
          \eme{d}{\tm_0}{\tmtwo\,\tmthree}
        \\
        & = &
          \eme{d}{\tm_0}{\tmtwo} + \eme{d}{\tm_0}{\tmthree}
        \\
        & \mgt &
          \eme{d}{\tm'_0}{\tmtwo'} + \eme{d}{\tm_0}{\tmthree} + \mset
          & \text{by the internal \ih}
        \\
        & \mgeq &
          \eme{d}{\tm'_0}{\tmtwo'} + \eme{d}{\tm'_0}{\tmthree} + \mset
          & \text{by item~\ref{lower_reduction:eme_left_decrease}
                  and \rlem{properties_of_mgtmap}}
        \\
        & = &
          \eme{d}{\tm'_0}{\tmtwo'\,\tmthree} + \mset
        \\
        & = &
          \eme{d}{\tm'_0}{\tm'} + \mset
        \end{array}
      \]
    \item
      If the step is internal to $\tmthree$:
      Similar to the previous case.
    \end{enumerate}
  \item
    $\tm = \bin{\tmtwo}{\tmthree}$:
    Similar to case~\ref{lower_reduction:eme_right_equal_decrease:application_non_redex}.
  \item
    $\sctx = \ctxhole$:
    Impossible, as there are no reduction steps $\ctxhole \tod{d} \sctx'$.
  \item
    $\sctx = \bin{\sctx_1}{\tm}$:
    Similar to case~\ref{lower_reduction:eme_right_equal_decrease:application_non_redex}
  \end{enumerate}
\item
  Let $1 \leq d < j \leq D$
  and $\tm_0 \tod{d} \tm_0'$
  and $\anon \tod{d} \anon'$,
  where $\anon,\anon'$
  are either terms ($\anon = \tm$ and $\anon = \tm'$)
  or memories ($\anon = \sctx$ and $\anon = \sctx'$).
  We argue that $\eme{j}{\tm_0}{\anon} \mgeq \eme{j}{\tm'_0}{\anon'}$.
  We proceed by induction on $\anon$:
  \begin{enumerate}
  \item
    $\tm = \var$:
    Impossible, as there are no reduction steps $\var \tod{d} \tm'$.
  \item
    $\tm = \lam{\var}{\tmtwo}$:
    Then the step is of the form
    $\tm
     = \lam{\var}{\tmtwo}
     \tod{d} \lam{\var}{\tmtwo'}
     = \tm'$
    with $\tmtwo \tod{d} \tmtwo'$.
    Then
    $\eme{j}{\tm_0}{\tm}
     = \eme{j}{\tm_0}{\lam{\var}{\tmtwo}}
     = \eme{j}{\tm_0}{\tmtwo}
     \mgeq \eme{j}{\tm'_0}{\tmtwo'}
     = \eme{j}{\tm'_0}{\lam{\var}{\tmtwo'}}
     = \eme{j}{\tm'_0}{\tm'}$
    by the internal \ih.
  \item
    If $\tm = (\lam{\var}{\tmtwo})\sctx\,\tmthree$
    is the redex of degree $d$ contracted by the step $\tm \tod{d} \tm'$:
    Then the step is of the form
    $\tm
     = (\lam{\var}{\tmtwo})\sctx\,\tmthree
     \tod{d} \tmtwo\sub{\var}{\tmthree}\garb{\tmthree}\sctx
     = \tm'$.
    Recall that by hypothesis $d < j$,
    and note that the abstraction $(\lam{\var}{\tmtwo})\sctx$
    is of degree $d$,
    so the type of the argument $\tmthree$ must be of height less than $d$,
    that is, $\height{\typeof{\tmthree}} < d < j$.
    In particular, the argument $\tmthree$ cannot be a $\symg$-abstraction
    of degree $j$,
    so we may apply the lower substitution lemma~(\rlem{lower_substitution_lemma}),
    which ensures that there exists $k \in \Natz$
    such that
    $\eme{j}{\tm'_0}{\tmtwo\sub{\var}{\tmthree}}
     = \eme{j}{\tm'_0}{\tmtwo} + k \mtimes \eme{j}{\tm'_0}{\tmthree}$.
    Furthermore, observe that
    $\eme{j}{\tm_0}{\tmthree}
     \mgeq (1 + k) \mtimes \eme{j}{\tm'_0}{\tmthree}$.
    Indeed by item~\ref{lower_reduction:eme_left_decrease}
    $\eme{j}{\tm_0}{\tmthree}
     \mgtmap \eme{j}{\tm'_0}{\tmthree}$
    so by \rlem{properties_of_mgtmap}
    $\eme{j}{\tm_0}{\tmthree}
     \mgeq (1 + k) \mtimes \eme{j}{\tm'_0}{\tmthree}$.
    Moreover, note that $\tm$ is not a redex of degree $j$, so:
    \[
      \begin{array}{rcll}
        \eme{j}{\tm_0}{\tm}
      & = &
        \eme{j}{\tm_0}{(\lam{\var}{\tmtwo})\sctx\,\tmthree}
      \\
      & = &
          \eme{j}{\tm_0}{\tmtwo}
        + \eme{j}{\tm_0}{\sctx}
        + \eme{j}{\tm_0}{\tmthree}
      \\
      & \mgeq &
          \eme{j}{\tm'_0}{\tmtwo}
        + \eme{j}{\tm'_0}{\sctx}
        + \eme{j}{\tm_0}{\tmthree}
        \HS\text{by item~\ref{lower_reduction:eme_left_decrease}
                     and~\rlem{properties_of_mgtmap}}
      \\
      & \mgeq &
          \eme{j}{\tm'_0}{\tmtwo}
        + \eme{j}{\tm'_0}{\sctx}
        + (1 + k) \mtimes \eme{j}{\tm'_0}{\tmthree}
        \\&&\HS\text{since
                     $\eme{j}{\tm_0}{\tmthree}
                      \mgeq (1 + k) \mtimes \eme{j}{\tm'_0}{\tmthree}$}
      \\
      & = &
          \eme{j}{\tm'_0}{\tmtwo}
        + k \mtimes \eme{j}{\tm'_0}{\tmthree}
        + \eme{j}{\tm'_0}{\tmthree}
        + \eme{j}{\tm'_0}{\sctx}
      \\
      & = &
          \eme{j}{\tm'_0}{\tmtwo\sub{\var}{\tmthree}}
        + \eme{j}{\tm'_0}{\tmthree}
        + \eme{j}{\tm'_0}{\sctx}
      \\
      & = &
        \eme{j}{\tm'_0}{\tmtwo\sub{\var}{\tmthree}\garb{\tmthree}\sctx}
      \\
      & = &
        \eme{j}{\tm_0}{\tm'}
      \end{array}
    \]
  \item
    If $\tm = (\lam{\var}{\tmtwo})\sctx\,\tmthree$
    is a redex of degree $j$:
    Note that the step $\tm \tod{d} \tm'$ cannot be at the root,
    because $d < j$, so the redex at the root is not of degree $d$.
    There are three subcases, depending on whether the step is
    internal to $\tmtwo$, internal to $\sctx$, or internal to $\tmthree$.
    All these subcases are similar; we only give the proof
    for the case in which the step is internal to $\tmtwo$.
    Then the step is of the form
    $\tm
     = (\lam{\var}{\tmtwo})\sctx\,\tmthree
     \tod{d} (\lam{\var}{\tmtwo'})\sctx\,\tmthree
     = \tm'$
    with $\tmtwo \tod{d} \tmtwo'$, and we have:
    \[
      \begin{array}{rcll}
      &&
        \eme{j}{\tm_0}{\tm}
      \\
      & = &
        \eme{j}{\tm_0}{(\lam{\var}{\tmtwo})\sctx\,\tmthree}
      \\
      & = &
          \eme{j}{\tm_0}{\tmtwo}
        + \eme{j}{\tm_0}{\sctx}
        + \eme{j}{\tm_0}{\tmthree}
        + \ms{(j,\bme{j}{\tm_0})}
      \\
      & \mgeq &
          \eme{j}{\tm'_0}{\tmtwo'}
        + \eme{j}{\tm_0}{\sctx}
        + \eme{j}{\tm_0}{\tmthree}
        + \ms{(j,\bme{j}{\tm_0})}
        & \text{by the internal \ih}
      \\
      & \mgeq &
          \eme{j}{\tm'_0}{\tmtwo'}
        + \eme{j}{\tm'_0}{\sctx}
        + \eme{j}{\tm'_0}{\tmthree}
        + \ms{(j,\bme{j}{\tm_0})}
        & \text{by item~\ref{lower_reduction:eme_left_decrease}
                    and~\rlem{properties_of_mgtmap}}
      \\
      & \mgeq &
          \eme{j}{\tm'_0}{\tmtwo'}
        + \eme{j}{\tm'_0}{\sctx}
        + \eme{j}{\tm'_0}{\tmthree}
        + \ms{(j,\bme{j}{\tm'_0})}
        & \text{by item~\ref{lower_reduction:bme_decrease}}
      \\
      & = &
        \eme{j}{\tm'_0}{(\lam{\var}{\tmtwo'})\sctx\,\tmthree}
      \\
      & = &
        \eme{j}{\tm'_0}{\tm'}
      \end{array}
    \]
  \item
    \label{lower_reduction:eme_right_nonequal_decrease:application_non_redex}
    If $\tm = \tmtwo\,\tmthree$
    is not the redex contracted by the step $\tm \tod{d} \tm'$
    nor a redex of degree $j$:
    There are two subcases, depending on whether the step
    $\tm \tod{d} \tm'$ is internal to $\tmtwo$ or internal to $\tmthree$:
    \begin{enumerate}
    \item
      If the step is internal to $\tmtwo$,
      then the step is of the form
      $\tm
       = \tmtwo\,\tmthree
       \tod{d} \tmtwo'\,\tmthree
       = \tm'$
      with $\tmtwo \tod{d} \tmtwo'$.
      Note that $\tmtwo$ is not a $\symg$-abstraction of degree $j$,
      (because $\tmtwo\,\tmthree$ is not a redex of degree $j$).
      Moreover $d < j$
      so by \rlem{reduction_does_not_create_higher_degree_redexes}
      we have that $\tmtwo'$ is not a $\symg$-abstraction of degree $j$.
      Then:
      \[
        \begin{array}{rcll}
          \eme{j}{\tm_0}{\tm}
        & = &
          \eme{j}{\tm_0}{\tmtwo\,\tmthree}
        \\
        & = &
          \eme{j}{\tm_0}{\tmtwo} + \eme{j}{\tm_0}{\tmthree}
        \\
        & \mgeq &
          \eme{j}{\tm'_0}{\tmtwo'} + \eme{j}{\tm_0}{\tmthree}
          & \text{by the internal \ih}
        \\
        & \mgeq &
          \eme{j}{\tm'_0}{\tmtwo'} + \eme{j}{\tm'_0}{\tmthree}
          & \text{by item~\ref{lower_reduction:eme_left_decrease}
                      and~\rlem{properties_of_mgtmap}}
        \\
        & = &
          \eme{j}{\tm'_0}{\tmtwo'\,\tmthree}
        \\
        & = &
          \eme{j}{\tm'_0}{\tm'}
        \end{array}
      \]
    \item
      If the step is internal to $\tmthree$:
      Similar to the previous case.
    \end{enumerate}
  \item
    $\tm = \bin{\tmtwo}{\tmthree}$:
    Similar to case~\ref{lower_reduction:eme_right_nonequal_decrease:application_non_redex}.
  \item
    $\sctx = \ctxhole$:
    Impossible, as there are no reduction steps $\ctxhole \tod{d} \sctx'$.
  \item
    $\sctx = \bin{\sctx_1}{\tm}$:
    Similar to case~\ref{lower_reduction:eme_right_nonequal_decrease:application_non_redex}.
  \end{enumerate}
\item
  Let $1 \leq d \leq D$
  and $\tm \tod{d} \tm'$.
  We argue that $\ame{D}{\tm} \mgt \ame{D}{\tm'}$.
  Indeed:
  \[
    \begin{array}{rcll}
      \ame{D}{\tm}
    & = &
      \sum_{i=1}^{D} \eme{i}{\tm}{\tm}
    \\
    & = &
        \ame{d-1}{\tm}
      + \eme{d}{\tm}{\tm}
      + (\sum_{j=d+1}^{D} \eme{j}{\tm}{\tm})
    \\
    & \mgeq &
        \eme{d}{\tm}{\tm}
      + (\sum_{j=d+1}^{D} \eme{j}{\tm}{\tm})
      \\&&\HS\text{removing the first term}
    \\
    & \mgt &
        \ame{d-1}{\tm'}
      + \eme{d}{\tm'}{\tm'}
      + (\sum_{j=d+1}^{D} \eme{j}{\tm}{\tm})
      \\&&\HS\text{by item~\ref{lower_reduction:eme_right_equal_decrease},
                   taking $\mset := \ame{d-1}{\tm'}$}
    \\
    & \mgeq &
        \ame{d-1}{\tm'}
      + \eme{d}{\tm'}{\tm'}
      + (\sum_{j=d+1}^{D} \eme{j}{\tm'}{\tm'})
      \\&&\HS\text{by item~\ref{lower_reduction:eme_right_nonequal_decrease}}
    \\
    & = &
      \ame{D}{\tm'}
    \end{array}
  \]
\end{enumerate}
\end{proof}

\begin{proposition}[Forget/decrease]
\lprop{appendix:shrinking_ame}
Let $d \in \Natz$. Then the following hold:
\begin{enumerate}
\item
  \label{shrinking_ame:bme_decrease}
  If $\tm \shone \tm'$
  then $\bme{d}{\tm} \mgeq \bme{d}{\tm'}$.
\item
  \label{shrinking_ame:eme_left_decrease}
  If $\tm_0 \shone \tm'_0$
  then $\eme{d}{\tm_0}{\tm} \mgeq \eme{d}{\tm'_0}{\tm}$.
\item
  \label{shrinking_ame:eme_right_decrease}
  If $\tm_0 \shone \tm'_0$
  and $\tm \shone \tm'$
  then $\eme{d}{\tm_0}{\tm} \mgeq \eme{d}{\tm'_0}{\tm'}$.
\item
  \label{shrinking_ame:ame_decrease}
  If $\tm \shone \tm'$
  then $\ame{d}{\tm} \mgeq \ame{d}{\tm'}$.
\end{enumerate}
\end{proposition}
\begin{proof}
We prove a more general version of the statement:
in items 
\ref{shrinking_ame:eme_left_decrease} and
\ref{shrinking_ame:eme_right_decrease}
we allow $\tm$ to be either a term or a memory.
For example, the statement of item~\ref{lower_reduction:eme_left_decrease}
is generalized as follows:
  if $\tm_0 \shone \tm'_0$
  then $\eme{j}{\tm_0}{\anon} \mgeq \eme{j}{\tm'_0}{\anon}$,
  where $\anon$ is either a term or a memory.

We prove all items simultaneously by induction on $d$.
Note that:
item~1. resorts to the \ih;
item~2. resorts to item~1. (without decreasing $d$);
item~3. resorts to items~1. and~2. (without decreasing $d$);
item~4. resorts to item~3. (without necessarily decreasing $d$).
\begin{enumerate}
\item
  Let $\tm \shone \tm'$.
  We argue that $\bme{d}{\tm} \mgeq \bme{d}{\tm'}$.
  Let $X$ and $Y$
  be the sets of reduction sequences
  $X := \set{\redseq \ST (\exists{\tmtwo})\ \redseq : \tm \rtod{d} \tmtwo}$
  and
  $Y := \set{\redseqtwo \ST (\exists{\tmtwo'})\ \redseqtwo : \tm' \rtod{d} \tmtwo'}$.
  Note that, by definition,
  $\bme{j}{\tm} =
   \msb{\ame{d-1}{\tgt{\redseq}}}{\redseq \in X}$
  and
  $\bme{j}{\tm'} =
   \msb{\ame{d-1}{\tgt{\redseqtwo}}}{\redseqtwo \in Y}$.
  We construct a function $\varphi : Y \to X$ as follows.
  Consider a forgetful step $\redex : \tm \shone \tm'$;
  there may be more than one such step, but there is at least one
  by hypothesis.
  By postponement of forgetful reduction~\rprop{retraction_before_shrinking},
  for each reduction sequence $\redseqtwo : \tm' \rtod{d} \tmtwo'$
  there exists a term $\tmtwo_\redseqtwo$
  such that $\tmtwo_\redseqtwo \shone^* \tmtwo'$
  and
  a reduction sequence
  $\retract{\redseqtwo}{\redex} : \tm \rtod{d} \tmtwo_\redseqtwo$.
  In particular, $(\retract{\redseqtwo}{\redex}) \in X$,
  and we can define $\varphi(\redseqtwo) := \retract{\redseqtwo}{\redex}$.
  Moreover, $\varphi$ is injective
  because if $\redseqtwo_1,\redseqtwo_2 \in Y$
  are such that $\retract{\redseqtwo_1}{\redex} = \retract{\redseqtwo_2}{\redex}$
  then by~\rprop{retraction_before_shrinking}
  we have that $\redseqtwo_1 = \redseqtwo_2$.

  First, we claim that
  $\ame{d-1}{\tgt{\varphi(\redseqtwo)}} \mgeq \ame{d-1}{\tgt{\redseqtwo}}$
  for every $\redseqtwo \in Y$. Indeed:
  \[
    \begin{array}{rcll}
      \ame{d-1}{\tgt{\varphi(\redseqtwo)}}
    & = &
      \ame{d-1}{\tmtwo_\redseqtwo}
    \\
    & \mgeq &
      \ame{d-1}{\tmtwo'}
      & \text{by item~\ref{shrinking_ame:ame_decrease} of the \ih}
    \\
    & = &
      \ame{d-1}{\tgt{\redseqtwo}}
    \end{array}
  \]
  To be able to apply item~\ref{shrinking_ame:ame_decrease} of the \ih,
  observe that we have $d - 1 < d$.
  We apply the \ih as many times as the number of forgetful steps
  in $\tmtwo_\redseqtwo \shone^* \tmtwo'$.

  To conclude the proof, let $Z = X \setminus \varphi(Y)$,
  so that $X = \varphi(Y) \uplus Z$,
  and note that:
  \[
    \begin{array}{rcll}
      \bme{d}{\tm}
    & = &
      \msb{\ame{d-1}{\tgt{\redseq}}}{\redseq \in X}
    \\
    & = &
      \msb{\ame{d-1}{\tgt{\redseq}}}{\redseq \in \varphi(Y) \uplus Z}
    \\
    & = &
        \msb{\ame{d-1}{\tgt{\redseq}}}{\redseq \in \varphi(Y)}
      + \msb{\ame{d-1}{\tgt{\redseq}}}{\redseq \in Z}
    \\
    & \mgeq &
      \msb{\ame{d-1}{\tgt{\redseq}}}{\redseq \in \varphi(Y)}
    \\
    & = &
      \msb{\ame{d-1}{\tgt{\varphi(\redseqtwo)}}}{\redseqtwo \in Y}
      & \text{($\star$)}
    \\
    & \mgeq &
      \msb{\ame{d-1}{\tgt{\redseqtwo}}}{\redseqtwo \in Y}
      & \text{($\star\star$)}
    \\
    & = &
      \bme{d}{\tm'}
    \end{array}
  \]
  To justify the step marked with ($\star$), note that $\varphi$ is injective,
  so $Y$ and $\varphi(Y)$ have the same cardinality.
  To justify the step marked with ($\star\star$),
  note that
  $\msb{\ame{d-1}{\tgt{\varphi(\redseqtwo)}}}{\redseqtwo \in Y}
   = \sum_{\redseqtwo \in Y} \ms{\ame{d-1}{\tgt{\varphi(\redseqtwo)}}}
   \mgeq \sum_{\redseqtwo \in Y} \ms{\ame{d-1}{\tgt{\redseqtwo}}}
   = \msb{\ame{d-1}{\tgt{\redseqtwo}}}{\redseqtwo \in Y}$
  because 
  $\ame{d-1}{\tgt{\varphi(\redseqtwo)}} \mgeq \ame{d-1}{\tgt{\redseqtwo}}$,
  as we have already justified.
\item
  Let $\tm_0 \shone \tm'_0$ and
  let $\anon$ be either a term or a memory.
  We argue that $\eme{d}{\tm_0}{\anon} \mgeq \eme{d}{\tm'_0}{\anon}$.
  We proceed by induction on $\anon$:
  \begin{enumerate}
  \item
    $\tm = \var$:
    Then
    $\eme{d}{\tm_0}{\var}
     = \msempty
     \mgeq \msempty
     = \eme{d}{\tm'_0}{\var}$.
  \item
    $\tm = \lam{\var}{\tmtwo}$:
    Then
    $\eme{d}{\tm_0}{\lam{\var}{\tmtwo}}
     = \eme{d}{\tm_0}{\tmtwo}
     \mgeq \eme{d}{\tm'_0}{\tmtwo}
     = \eme{d}{\tm'_0}{\lam{\var}{\tmtwo}}$
    by the internal \ih.
  \item
    If $\tm = (\lam{\var}{\tmtwo})\sctx\,\tmthree$ is a redex of degree $d$:
    Then:
    \[
      \begin{array}{rcll}
        \eme{d}{\tm_0}{\tm}
      & = &
        \eme{d}{\tm_0}{(\lam{\var}{\tmtwo})\sctx\,\tmthree}
      \\
      & = &
          \eme{d}{\tm_0}{\tmtwo}
        + \eme{d}{\tm_0}{\sctx}
        + \eme{d}{\tm_0}{\tmthree}
        + \ms{(d,\bme{d}{\tm_0})}
      \\
      & \mgeq &
          \eme{d}{\tm'_0}{\tmtwo}
        + \eme{d}{\tm'_0}{\sctx}
        + \eme{d}{\tm'_0}{\tmthree}
        + \ms{(d,\bme{d}{\tm_0})}
        & \text{by the internal \ih}
      \\
      & \mgeq &
          \eme{d}{\tm'_0}{\tmtwo}
        + \eme{d}{\tm'_0}{\sctx}
        + \eme{d}{\tm'_0}{\tmthree}
        + \ms{(d,\bme{d}{\tm'_0})}
        & \text{by item~\ref{shrinking_ame:bme_decrease}}
      \\
      & = &
        \eme{d}{\tm'_0}{(\lam{\var}{\tmtwo})\sctx\,\tmthree}
      \\
      & = &
        \eme{d}{\tm'_0}{\tm}
      \end{array}
    \]
  \item
    \label{shrinking_ame:eme_left_decrease:application_non_redex}
    If $\tm = \tmtwo\,\tmthree$ is not a redex of degree $d$:
    Then:
    \[
      \begin{array}{rcll}
        \eme{d}{\tm_0}{\tm}
      & = &
        \eme{d}{\tm_0}{\tmtwo\,\tmthree}
      \\
      & = &
          \eme{d}{\tm_0}{\tmtwo}
        + \eme{d}{\tm_0}{\tmthree}
      \\
      & \mgeq &
          \eme{d}{\tm'_0}{\tmtwo}
        + \eme{d}{\tm'_0}{\tmthree}
        & \text{by the internal \ih}
      \\
      & = &
        \eme{d}{\tm'_0}{\tm}
      \end{array}
    \]
  \item
    $\tm = \bin{\tmtwo}{\tmthree}$:
    Similar to case~\ref{shrinking_ame:eme_left_decrease:application_non_redex}.
  \item
    $\sctx = \ctxhole$:
    Then
    $\eme{d}{\tm_0}{\ctxhole}
     = \msempty
     \mgeq \msempty
     = \eme{d}{\tm'_0}{\ctxhole}$.
  \item
    $\sctx = \bin{\sctx_1}{\tm}$:
    Similar to case~\ref{shrinking_ame:eme_left_decrease:application_non_redex}.
  \end{enumerate}
\item
  Let $\tm_0 \shone \tm'_0$
  and $\anon \shone \anon'$,
  where $\anon$ and $\anon'$
  are either terms ($\anon = \tm$ and $\anon = \tm'$)
  or memories ($\anon = \sctx$ and $\anon = \sctx'$).
  We argue that $\eme{d}{\tm_0}{\anon} \mgeq \eme{d}{\tm'_0}{\anon'}$.
  We proceed by induction on $\anon$:
  \begin{enumerate}
  \item
    $\tm = \var$: 
    Impossible, as there are no forgetful steps $\var \shone \tm'$.
  \item
    $\tm = \lam{\var}{\tmtwo}$: 
    Then
    $\tm
     = \lam{\var}{\tmtwo}
     \shone \lam{\var}{\tmtwo'}
     = \tm'$
    with $\tmtwo \shone \tmtwo'$
    and
    $\eme{d}{\tm_0}{\tm}
    = \eme{d}{\tm_0}{\lam{\var}{\tmtwo}}
    = \eme{d}{\tm_0}{\tmtwo}
    \mgeq \eme{d}{\tm'_0}{\tmtwo'}
    \mgeq \eme{d}{\tm'_0}{\lam{\var}{\tmtwo'}}
    \mgeq \eme{d}{\tm'_0}{\tm'}$
    by the internal \ih.
  \item
    If $\tm = (\lam{\var}{\tmtwo})\sctx\,\tmthree$ is a redex of degree $d$: 
    There are three subcases, depending on whether the forgetful step
    $\tm \shone \tm'$ is internal to $\tmtwo$, internal to $\sctx$,
    or internal to $\tmthree$.
    All these subcases are similar; we only give the proof
    for the case in which the step is internal to $\tmtwo$.
    Then
    $\tm
     = (\lam{\var}{\tmtwo})\sctx\,\tmthree
     \shone (\lam{\var}{\tmtwo'})\sctx\,\tmthree
     = \tm'$
    with $\tmtwo \shone \tmtwo'$
    and:
    \[
      \begin{array}{rcll}
        \eme{d}{\tm_0}{\tm}
      & = &
        \eme{d}{\tm_0}{(\lam{\var}{\tmtwo})\sctx\,\tmthree}
      \\
      & = &
          \eme{d}{\tm_0}{\tmtwo}
        + \eme{d}{\tm_0}{\sctx}
        + \eme{d}{\tm_0}{\tmthree}
        + \ms{(d,\bme{d}{\tm_0})}
      \\
      & \mgeq &
          \eme{d}{\tm'_0}{\tmtwo'}
        + \eme{d}{\tm_0}{\sctx}
        + \eme{d}{\tm_0}{\tmthree}
        + \ms{(d,\bme{d}{\tm_0})}
        & \text{by the internal \ih}
      \\
      & \mgeq &
          \eme{d}{\tm'_0}{\tmtwo'}
        + \eme{d}{\tm'_0}{\sctx}
        + \eme{d}{\tm'_0}{\tmthree}
        + \ms{(d,\bme{d}{\tm_0})}
        & \text{by item~\ref{shrinking_ame:eme_left_decrease}}
      \\
      & \mgeq &
          \eme{d}{\tm'_0}{\tmtwo'}
        + \eme{d}{\tm'_0}{\sctx}
        + \eme{d}{\tm'_0}{\tmthree}
        + \ms{(d,\bme{d}{\tm'_0})}
        & \text{by item~\ref{shrinking_ame:bme_decrease}}
      \\
      & = &
        \eme{d}{\tm'_0}{(\lam{\var}{\tmtwo'})\sctx\,\tmthree}
      \\
      & = &
        \eme{d}{\tm'_0}{\tm'}
      \end{array}
    \]
  \item
    \label{shrinking_ame:eme_right_decrease:application_non_redex}
    If $\tm = \tmtwo\,\tmthree$ is not a redex of degree $d$: 
    There are two subcases, depending on whether the forgetful step
    $\tm \shone \tm'$ is internal to $\tmtwo$ or internal to $\tmthree$.
    These subcases are similar; we only give the proof
    for the case in which the step is internal to $\tmtwo$.
    Then $\tm = \tmtwo\,\tmthree \shone \tmtwo'\,\tmthree = \tm'$
    with $\tmtwo \shone \tmtwo'$ and:
    \[
      \begin{array}{rcll}
        \eme{d}{\tm_0}{\tm}
      & = &
        \eme{d}{\tm_0}{\tmtwo\,\tmthree}
      \\
      & = &
          \eme{d}{\tm_0}{\tmtwo}
        + \eme{d}{\tm_0}{\tmthree}
      \\
      & \mgeq &
          \eme{d}{\tm'_0}{\tmtwo'}
        + \eme{d}{\tm_0}{\tmthree}
        & \text{by the internal \ih}
      \\
      & \mgeq &
          \eme{d}{\tm'_0}{\tmtwo'}
        + \eme{d}{\tm'_0}{\tmthree}
        & \text{by item~\ref{shrinking_ame:eme_left_decrease}}
      \\
      & = &
        \eme{d}{\tm'_0}{\tmtwo'\,\tmthree}
      \\
      & = &
        \eme{d}{\tm'_0}{\tm'}
      \end{array}
    \]
  \item
    \label{shrinking_ame:eme_right_decrease:bin}
    $\tm = \bin{\tmtwo}{\tmthree}$:
    There are three subcases, depending on whether the forgetful step
    $\tm \shone \tm'$ is at the root, internal to $\tmtwo$, or internal
    to $\tmthree$:
    \begin{enumerate}
    \item
      If the step is at the root:
      Then the step is of the form
      $\tm = \bin{\tmtwo}{\tmthree} \shone \tmtwo = \tm'$
      and:
      \[
        \begin{array}{rcll}
          \eme{d}{\tm_0}{\tm}
        & = &
          \eme{d}{\tm_0}{\bin{\tmtwo}{\tmthree}}
        \\
        & = &
            \eme{d}{\tm_0}{\tmtwo}
          + \eme{d}{\tm_0}{\tmthree}
        \\
        & \mgeq &
          \eme{d}{\tm_0}{\tmtwo}
        \\
        & \mgeq &
          \eme{d}{\tm'_0}{\tmtwo}
          & \text{by item~\ref{shrinking_ame:eme_left_decrease}}
        \\
        & = &
          \eme{d}{\tm'_0}{\tm'}
        \end{array}
      \]
    \item
      If the step is internal to $\tmtwo$:
      Similar to case~\ref{shrinking_ame:eme_right_decrease:application_non_redex}.
    \item
      If the step is internal to $\tmthree$:
      Similar to case~\ref{shrinking_ame:eme_right_decrease:application_non_redex}.
    \end{enumerate}
  \item
    $\sctx = \ctxhole$:
    Impossible, as there are no forgetful steps $\ctxhole \shone \sctx'$.
  \item
    $\sctx = \bin{\sctx_1}{\tm}$:
    Similar to case~\ref{shrinking_ame:eme_right_decrease:bin}.
  \end{enumerate}
\item
  Let $\tm \shone \tm'$.
  We argue that $\ame{d}{\tm} \mgeq \ame{d}{\tm'}$.
  Indeed:
  \[
    \begin{array}{rcll}
      \ame{d}{\tm}
    & = &
      \sum_{i=1}^{d} \eme{d}{\tm}{\tm}
    \\
    & \mgeq &
      \sum_{i=1}^{d} \eme{d}{\tm'}{\tm'}
      & \text{by item~\ref{shrinking_ame:eme_right_decrease},
              resorting to the \ih when $i < d$}
    \\
    & = &
      \ame{d}{\tm'}
    \end{array}
  \]
  Note that for the value $i = d$, we resort directly to
  item~\ref{shrinking_ame:eme_right_decrease} and not to the \ih.
\end{enumerate}
\end{proof}

\end{document}